\documentclass[sn-mathphys-num]{sn-jnl}% Math and Physical Sciences Numbered Reference Style 
\usepackage{graphicx}%
\usepackage{multirow}%
\usepackage{amsmath,amssymb,amsfonts}%
\usepackage{amsthm}%
\usepackage{mathrsfs}%
\usepackage[title]{appendix}%
\usepackage{xcolor}%
\usepackage{textcomp}%
\usepackage{manyfoot}%
\usepackage{booktabs}%
\usepackage{algorithm}%
\usepackage{algorithmicx}%
\usepackage{algpseudocode}%
\usepackage{listings}%

\usepackage{xcolor}
\usepackage{proof}
\usepackage{amsfonts}
\usepackage{amsmath}
\usepackage{amssymb}
\usepackage{latexsym}
\usepackage{graphicx}
\usepackage{enumitem}
\usepackage{cleveref}
\usepackage{natbib}

\newtheorem{example}{Example}
\newtheorem{definition}{Definition}
\newtheorem{lemma}{Lemma}

\newcommand{\ehi}{
}

\newcommand{\romc}{\mathrm{C}}

\newcommand{\rome}{\mathrm{E}}
\newcommand{\romr}{\mathrm{R}}
\newcommand{\romp}{\mathrm{P}}
\newcommand{\romo}{\mathrm{O}}

\newcommand{\romi}{\mathrm{I}}

\newcommand{\romw}{\mathrm{W}}
\newcommand{\romu}{\mathrm{U}}

\newcommand{\calf}{\mathcal{F}}
\newcommand{\calm}{\mathcal{M}}
\newcommand{\caln}{\mathcal{N}}

\newcommand{\calp}{\mathcal{P}}
\newcommand{\calt}{\mathcal{T}}

\newcommand{\canm}{\calm^\romc}
\newcommand{\canw}{\romw^\romc}
\newcommand{\canr}{\romr^\romc}
\newcommand{\canu}{\romu^\romc}
\newcommand{\cani}{\romi^\romc}
\newcommand{\cane}{\rome^\romc}

\newcommand{\know}{\mathrm{K}}
\newcommand{\trust}{\mathrm{T}}

\newcommand{\lid}{\doteq}

\theoremstyle{thmstyleone}%
\newtheorem{theorem}{Theorem}%  meant for continuous numbers
\theoremstyle{thmstyletwo}%

\theoremstyle{thmstylethree}%

\newtheorem{corollary}{Corollary}

\begin{document}

\title[A Logic of Knowledge and Justifications, with an Application to Computational Trust]{A Logic of Knowledge and Justifications, with an Application to Computational Trust}


\author*{\fnm{Francesco A.} \sur{Genco}}\email{francesco.genco@unito.it}

\affil{\orgdiv{Center for Logic, Language and Cognition, Department of Philosophy and Education Sciences},
\orgname{University of Turin},
{\street{Via Santo'Ottavio 20}, \city{Turin}, \postcode{10124}, \country{Italy}}}
%
%\affil[2]{\orgdiv{Department}, \orgname{Organization}, \orgaddress{\street{Street}, \city{City}, \postcode{10587}, \state{State}, \country{Country}}}
%
%\affil[3]{\orgdiv{Department}, \orgname{Organization}, \orgaddress{\street{Street}, \city{City}, \postcode{610101}, \state{State}, \country{Country}}}


\abstract{We present a logical framework that enables us to define a formal theory of computational trust in which this notion is analysed in terms of epistemic attitudes towards the possible objects of trust and in relation to existing evidence in favour of the trustworthiness of these objects. The framework is based on a quantified epistemic and justification logic featuring a non-standard handling of identities. Thus, the theory is able to account for the hyperintensional nature of computational trust. We present a proof system and a frame semantics for the logic, we prove soundness and completeness results and we introduce the syntactical machinery required to define a theory of trust.}

\keywords{Epistemic logic, justification logic, trust, computation, intensionality, hyperintensionality.}



\maketitle

\section{Introduction}\label{sec:intro}

I can trust someone with my life, but if they come to me unrecognisably disguised\dots Well, one might argue that I do indeed trust them, disguised or not, and that I can simply fail to realise that I do. My trust in them still holds, because they are, after all, a person that I trust, I just cannot see that they are, momentarily. One might argue, on the other hand, that we should simply admit that trust is a fleeting attitude that does not really hold invariably between individuals, but actually depends on the way one individual is presented to another one. The second line of thought would definitely lead to a linguistic analysis of trust as an intensional---possibly even hyperintensional\footnote{Hyperintensionality has been introduced by \citet{cre02} for the first time with respect to notions that are not far from that of trust. A hyperintensional relation, formally, is one that can hold between two terms $t$ and $s$ and not hold between $t$ and $s'$ even though $s$ and $s'$ are provably equivalent.}---relation. Trust, according to this analysis, would be a relation that can hold between two terms $t$ and $s$ and, at the same time, not hold between the two terms $t$ and $s'$ even though $s$ and $s'$ denote the same individual. This discrepancy would arise if $s$ and $s'$ constituted different ways of presenting their common denotation. This distinction between ways of denoting objects can be traced back to the essay {\it On Sense and Reference} by \citet{frephil} and possibly even further back to the work by \citet{an96}.\footnote{Here, the relevant distinction is the one between the {\it comprehension} and the {\it extension} of an idea.}
%If $s$ and $t$ are different descriptions of the same person, does the fact that someone trusts $s$ implies that that person also trusts $t$? In other words, if someone is willing to state that they trust $s$ but not that they trust $t$, are they simply mistakenly assuming that $s$ and $t$ are not the same but they actually trust both? Perhaps so, but one might say that we are legitimated in thinking that they indeed trust $s$ but not $t$, simply because the way $s$ presents its denotation does induce trust in the agent while the way $t$ presents it does not. 

If we consider computational trust \cite{sim20}---that is, in brief, the relation that holds between an agent and a computational entity in case the former trusts the latter to accomplish a given computational task---the issue seems less problematic, at least from a practical perspective. Indeed, an agent might trust a software they bought on an official website and not trust a pirated version of the software, even though the code of the two is exactly the same. Since a computer program can have different implementations based on the exact same code---making them mathematically indiscernible but physically different---and since it often happens that the internal workings of a software are not epistemically accessible to an agent or not understandable to them, we might very well have cases in which two instances of the same computer program\footnote{We can also talk of different ways of presenting the same implementation. Indeed, we might very well think of two distinct ways of accessing and employing the same exact implementation of the same software. For instance, different websites might redirect a user to the exact same implementation of a software.} trigger different attitudes of the same agent in terms of trust. The agent might trust one and not the other even though the two correctly implement the same algorithm, the same program and even the same code. And this is not only a conceptual point; if we wish to account for the actual behaviour of an agent, this is obviously a distinction we must make.

%As for formal analyses aiming at characterising the epistemic and subjective facets of trust and computational trust, we can mention the hyperintensional systems for trust and knowledge presented in \cite{fp21,sed21,pri22}. 

%The logical system is meant to enable an analysis of computational trust based, for instance, on the analysis of trustworthiness conducted by \citet{dp21}
A rather conspicuous amount of work has been devoted to the conceptual analysis of computational trust---see, e.g., \cite{cof07,cf10,key16,sim20,ess20,gmw20,sul20,esc21,cof21}---and to the formal analysis of notions that can be described as positive or negative forms of trust or trustworthiness---see, for instance,  \cite{arh98,dem04,zl05,sin11,pt12,pr14,pk16,prbt17,pri20,dp21,dgp22,pcd23,kp24,kra15}. 
%Moreover, a non-monotonic semantics has been introduced by \citet{kra15}. 
A discussion of this literature will be presented in Section \ref{sec:related-work}.

The aim of this work is to present a logical framework that enables us to define a formal theory of computational trust. The framework is supposed to provide an analysis of this notion in terms of epistemic attitudes towards the possible objects of trust and in relation to existing evidence in favour of the trustworthiness of these objects. A further objective of the work is to account for the hyperintensional nature of trust with respect to the ways a possible object of trust is presented. That is, we would like the logical theory to admit the possibility that an agent trusts some object when presented in a certain way but does not trust it when presented in another way, even though it is provable that the two presentations of it actually have the same denotation. Finally, we would like to keep the logical language as simple as possible and avoid the introduction of {\it ad hoc} constructs explicitly representing linguistic elements different from those that we aim at analysing.\footnote{As we will discuss later, several very expressive intensional systems have been introduced in the modal logic literature. These systems feature formal mechanisms to explicitly denote intensions. This would probably do the trick for us, but at the cost of having in the language and in the semantics two sorts of objects, one for intensions and one for denotations. Since, nevertheless, we do not aim at an analysis of the notion of intension, but we simply would like to formalise trust by a hyperintensional operator, we avoid this enrichment and complication of the language and endeavour instead in the task of achieving this by only employing the traditional linguistic apparatus of first order modal logic. For the same reason, we do not employ the semantical framework introduced by \citet{kra15} but introduce a simpler one tailored on the language employed.}

In order to define such a theory, we will introduce a quantified modal logic, acting as base framework, which combines an epistemic logic and a justification logic. As we will argue later in more depth, these two subsystems of our base logic are reasonably simple but also considerably different---in non trivial ways---with respect to other, more traditional logical systems of the same kind.   
Since the main focus of the presented work is on {\it computational} trust, we include in the language all we need to formulate statements about the properties of programs and computations, and about the attitudes of agents towards them. This is why the term language of the presented logic contains the whole language of pure $\lambda$-calculus  \cite{chu32,chu36,bar84}, a well known and extremely versatile model of computation. Thanks to the reliance on pure $\lambda$-calculus, and without extending its simple and yet expressive language, we will be able to formulate definitions enabling us to express also statements about {\it opaque probabilistic} programs and computations. From a technical perspective, some ingenuities will be employed in order to treat certain formal elements of $\lambda$-calculus as probabilistic programs. Programs of this kind, indeed, can neither be represented by using the primitive constructs of $\lambda$-calculus  nor be encoded in the core version of the formalism without properly extending it \citep[Section 11.1]{bar84}.     

%The reader which is not favourably inclined towards models of computation need not worry: the language of $\lambda$-calculus is not much different with respect to the traditional term language of classical propositional logic featuring variables constants and function symbols. The main differences, when only pure logic---and no specific theory---is concerned, are that instead of constants we use variables, and that terms themselves can be applied to other terms and there is no need to distinguish between terms and symbols that can be applied to terms. No part of this will be particularly relevant, though, before we introduce the formal theory of trust: from a purely logical perspective,  a term is a term and all that matters is whether it is or it is not identical to another term. As for the theory, we will take care, in due time, of explaining almost from scratch the features of $\lambda$-calculus that will matter with respect to our topic. 

As for the peculiar features of the epistemic operator of the presented logic, we aim at introducing a  simple and robust logical formalism for representing epistemic attitudes that are intensional with respect to individual objects. That is to say, we wish to be able to represent cases in which, even though terms $t$ and $s$ denote the same object, an agent knows that $t$ enjoys a property $P$ but does not know that $s$ does. In other words, we wish the formal system to feature agent-relative---we might even say {\it subjective}---knowledge of identity. Technically, we restrict the substitution of identicals in such a way that it cannot permute with respect to epistemic operators. Contingent identity, we might say, is then opaque with respect to the subjective epistemic perspective  of individual agents. In the semantics, each epistemic state can then feature a different theory of identity, and agents having access to different epistemic states can have considerably different knowledge concerning the identity of objects in the actual world.

%Since the long-term objective of this work is the formalisation of attitudes towards computational systems and processes, such as, for instance,  computational trust, we adopt as terms of our language the set of terms of $\lambda$-calculus. 
%
%The intended application of the logic, as already mentioned, is to formalise more complex attitudes, of only partial epistemic nature, towards computational systems and processes---or, more generally, agents---such as trust. Any reasonably rich formalisation of trust cannot simply be based on knowledge, as many analyses of this notion in different contexts clearly witness \cite{luh00}. The idea that this work aims at capturing is that of formalising trust as a combination of knowledge and of the existence of justifications for the trustworthiness of a system or process. 

As for the justification logic fragment of the presented logic, it is meant to provide a hyperintensional system expressing state of affairs about the relations between formulae and pieces of evidence that one can regard as warranting their truth. This fragment of the presented logic is very similar to well-known existing justification logic systems \cite{art95,an05,art08,fit08,ks12,fs20,af21} but presents some non negligible differences. Indeed, the fragment is a first-order justification logic, as is the logic introduced by \citet{fs20}; but, unlike the system by \citet{fs20}, the first-order justification logic embedded in our system does not feature a necessitation rule enabling us to conclude that the truth of each theorem of our system is witnessed by a justification. The absence of such a rule for justifications makes the formalised notion of trust hyperintensional. Indeed, if we had a rule enabling us to infer that $A$ is justified for each theorem $A$, then we would also have that, if an agent $a$ knows that there is a justification for $A[t/x]$ and $t=s$ is provable, then $a$ also knows that there is a justification for $A[s/x]$.\footnote{Technically, if we had a justification rule of the form\[\infer{\Rightarrow j:A}{\Rightarrow A}\]or even of the form\[\infer{ j_1:A_1 , \dots , j_n:A_n \Rightarrow j^{j_1 , \dots , j_n}:A}{j_1:A_1 , \dots , j_n:A_n\Rightarrow A }\]for some justification $j^{j_1 , \dots , j_n}$ possibly constructed  on the basis of the justifications $j_1 , \dots , j_n$ for the formulae $A_1 , \dots , A_n$, respectively;    
or any other rule of a similar kind enabling us to introduce a justification for any formula that can be derived from justified hypothesis, then we would validate the following derivability judgement: $j_1:A[t/x], j_2:t=s\vdash j^{j_1,j_2}:A[s/x]$  because of factivity of justifications and of the validity of the following judgement: $A[t/x], t=s\vdash A[s/x]$. This would, in turn, entail that, if an agent $a$ knows that there is a justification for $A[t/x]$ and $t=s$ is provable, then $a$ also know that there is a justification for $A[s/x]$. Indeed, if $t=s$ is provable then $j_2:t=s$ would be too by necessitation, and if the judgement\[j_1:A[t/x], j_2:t=s\vdash j^{j_1,j_2}:A[s/x]\] is valid, then so is the judgement\[\know _a j_1:A[t/x], \know _a j_2:t=s\vdash \know _a  j^{j_1,j_2}:A[s/x]\]because of the presence in our system of the necessitation rule\[\infer{\know _a A_1 , \dots \know _a A_n\Rightarrow \know _a A}{\know _a A_1 , \dots \know _a A_n\Rightarrow A}\]for the epistemic operator $\know $.}
Since we will not introduce a necessitation rule for justifications, and considering---as we will see later---that trust will require that the considered agent knows that there exists a justification for a suitable formula, we formalise a hyperintensional notion of trust: even though $A[t/x]$ and $A[s/x]$ are provably equivalent---that is, both $A[t/x]\vdash A[s/x] $ and $A[s/x]\vdash A[t/x]$ hold---under suitable conditions, there might exist an agent that trusts $t$ with respect to a certain task---formally specified by $A[x/x]$---but does not trust $s$ with respect to the same task. 
\footnote{Indeed, as argued, $\know _a j_1:A[t/x]$ and $\know _a j^{j_1,j_2}:A[s/x]$ might not be provably equivalent even though $s=t$ is provable. These would not be provably equivalent, for instance, when $\know _a s=t$ and $\know _a j_2:s=t$ are not among our hypotheses.}

%This feature of the system also implies, for similar reasons, that the notion of trust is hyperintensional with respect to provable identities of terms. Indeed, even if $\vdash s=t$ holds, it does not mean that there is a justification $j_2$ such that $\vdash j_2: s=t$ holds. And hence, even though $s$ and $t$ are provably identical, an agent might trust $t$ with respect to a certain task but not trust $s$ with respect to the same task. 

%Since, nevertheless, our system does not contain any necessitation-like rule for justifications, the judgement $j_1:A[t/x], j_2:t=s\vdash k:A[s/x]$ is not valid in general and we can have an agent that ignores that a .

The absence of a necessitation rule for justifications for instance of the form\[\infer{ j_1:A_1 , \dots , j_n:A_n \Rightarrow j^{j_1 , \dots , j_n}:A}{j_1:A_1 , \dots , j_n:A_n\Rightarrow A }\]can be seen as odd if we think of justifications as connected with formal proofs in general. Indeed, the intuition behind this rule might be that, if something---that is, $A$---can be deduced from the assumption that each element of a set of formulae $A_1 , \dots , A_n $ is justified, then we can conclude that a justification for what we are deriving from these formulae must exist, that is, $j^{j_1 , \dots , j_n}:A$. In this sense, a justification would be considered as any formal conclusive proof that the formula holds and, therefore, the existence of a derivation of $A$ from $j_1:A_1 , \dots , j_n:A_n$ would be {\it ipso facto} a witness of the fact that a justification of $A$ exists.
Such a direct connection between derivations and justifications bears a strong appeal in the context of justification logic---has been, for instance, enforced in the first-order justification logic introduced by \citet{fs20}---and 
closely relates certain justification logics with proofs-as-programs correspondences \cite{how80,sor06} and with the BHK interpretation \cite{kol32,hey34}, that are very intimately connected with these logics since their very inception \cite{art95,art08,af21}. Nevertheless, another possible and very well established reading of justification logics tends to interpret justifications as non-conclusive and possibly admissible pieces of evidence in favour of a formula \cite{af21}. This second view on the nature of justifications is certainly better suited to understand the present work. Obviously, by stating this, we do not mean to make any general claim on the nature of justifications as a whole. We just mean that the particular kind of justifications that we are formalising and that are effective with respect to the formalisation of the notion of trust we are considering are, generally, closer in nature to pieces of evidence---just as those that would be considered during a trial in a court of law, as it were---than to formal mathematical proofs. Let me also remark that in the previous sentence we specified ``in general'' since, a mathematical proof, if pertinent, would make a brilliant piece of evidence also in a court of law, but only if available and understandable; just as a mathematical proof, if pertinent and available, would constitute a brilliant incentive with respect to the trust an agent might or might not have towards another, possibly computational, agent or system. Nevertheless, in both cases, the mere existence of a mathematical proof might not realistically produce the effect it might have in principle on an agent's epistemic state.

It is also relevant to remark that the idea of exploiting the expressive power of justification logics that do not validate any sort of necessitation rule is not new to the literature. Two justification logics of this kind have been  introduced by \citet{rs21}\footnote{The justification logics by \citet{rs21} have been introduced in order to act as explicit counterparts of non-normal deontic logics---a modal logic is {\it non-normal} if the necessitation rule is not admissible in it. Explicit counterparts of deontic logics have the advantage that well-known deontic paradoxes can be avoided thanks to the additional discriminatory power granted by the use of possibly distinct justification terms rather than different occurrences of the same modal operator.}. While these systems are hyperintensional, they are not first-order, and hence cannot be directly used for the purposes of the present work. First-order non-normal modal logics do exist in the literature and are, indeed, hyperintensional---see, for instance, \cite{ap06}---but they involve non-negligible complications as far as the definition of a semantics is concerned. Indeed, the standard way of defining a semantics for non-normal modal logics is by employing  neighbourhood models. While it is true that these models would enable us to invalidate several logical principles governing the interactions between the modal operators and the logical connectives---which constitutes a great gain in expressive power and can be very useful in some cases---for achieving the objective of the present work it is enough to define a system in which {\it identity} is handled in a hyperintensional way. Hence, we opt for a first-order justification logic, which admits a much simpler frame semantics.

%These models have the benefit---as far as expressivity is concerned---of invalidating several logical principles governing the interactions between the modal operators and the logical connectives, this is not really required in order to achieve the objective of the present work, which is simply a hyperintensional handling of identity. 

%See \cite{mon69,mon70,fit06} on intensional logics with a first-order fragment.

Even thought the introduced system is not an {\it intensional logic} in the very specific sense that this expression has acquired in the modal logic literature \cite{mon69,mon70,fit06,gal16,fm23}, it is a rather strongly intensional system if we compare it with  more traditional modal logics. Indeed, traditional modal logics---{\it normal} ones, for instance\footnote{By a {\it normal} modal logic, as usual, we mean a logic in which the necessitation rule is valid and the axiom K is a theorem.}---can be called {\it intensional} in the sense that the modal operators are semantically defined on the basis of the truth of their subformulae at certain {\it possible worlds} or {\it states}, and not simply in a truth-functional manner, or {\it extensionally} \cite[Chapter V]{car88}. Nevertheless, in these systems, modal contexts are often not opaque with respect to identity. Hence, if we consider two formulae $A$ and $B$, the fact that $A\leftrightarrow B$ is true at a state of a model does not imply that also $\square A\leftrightarrow \square B$ is true at that state. This points at the fact that, under the modal operator $\square$, formulae are considered intensionally: it does not matter what is their truth value {\it at the present state}---their {\it extension}, some might say---but what is their truth value in all the relevant states---that is, their {\it intension}, in some acceptations of the word. Thus: even though $A$ and $B$ have the same truth value at the present state (that is, $A\leftrightarrow B$ is true), they might have different truth values at some of the relevant states (that is, $\neg (\square A\leftrightarrow \square B)$ might be true).  But if we consider, instead, two terms $t$ and $s$, the fact that $t=s$ holds does imply, in most of these systems, that all properties that hold necessarily for $t$ also hold necessarily for $s$ (that is, for any formula $A$, $\square A[t/x]\leftrightarrow \square A[s/x]$).\footnote{For instance, if $\square P(t)$ is true at a state $w$ and $t=s$ is true at $w$ as well, then also $\square P(s) $. This is due to the fact that the {\it necessity of identity} principle obtains: $t=s\rightarrow \square t=s$ is a theorem of the logic---see, for instance, the systems presented by \citet[Lemma 2.1]{cor02}, \citet[Section 3.3]{fit04} and \citet[Section 11.4]{fm23}. Technically, the connection is quite a direct one: if $\square P(t)$ and $t=s$ are true at $w$, then by necessity of identity also $\square t=s$ will. Thus, for any state $v$ which is accessible from $w$, we will have that both $P(t)$ and $s=t$ are true at $v$. This clearly implies that $P(s)$ is true at $v$. But since this means that $P(s)$ is verified by  all successors of $w$, we can conclude that $\square P(s)$ is true at $w$.} If we define our knowledge operator as a $\square$-like operator of this kind, there is no chance of representing a situation in which an agent ignores that two terms are identical in case they are. 

%The quantified modal logic in \cite[11.4]{fm23} verifies the necessity of identity and difference principles: $\forall x.\forall y.x=y\rightarrow \square x=y$ and $\forall x.\forall y.\neg x=y\rightarrow \neg \square x=y$.
%
%
%The system in \cite[3.3]{fit04} validates the formulae $x=y\rightarrow \square x=y$ and $x\neq y\rightarrow \square x\neq y$.

More subtle mechanisms that make it possible to treat terms intensionally exist \cite{mon69,mon70,fit06,gal16,fm23}, but they are often much more expressive---and complicated---than what is featured in our logic. Intensional logics  \cite{mon69,mon70,fit06,gal16,fm23}, indeed, often feature, at the level of the logical language, mechanisms for expressing judgements about intensions and explicit terms denoting intensions---something equivalent can be said for explicitly hyperintensional logics or systems \cite{fp21,sed21,pri22}. This is absolutely outside the scope of the present work, which simply aims at providing a system flexible enough to encode predicates and sentential operators that can behave differently on terms  with different intensions even though they have the same extension. In other words, we precisely aim here at a system in which it is possible to have $\know _a  P(t)$ and $t=s$ true at a state $w$ and $\know _a  P(s) $ false in case $\know _a  t=s $ does not hold; and in which, rather obviously, $\know _a  P(t)$ and $\know _a t=s$ together imply $\know _a  P(s) $. Nothing more is required here as far as intensionality is concerned. And this is exactly, as we will show, what our system delivers.

Notice, moreover, that the presented system does not have to rely on a semantics in which objects have counterparts in other states rather than just occurring at more than one state and does not have to rely on non-rigid designators in order to invalidate the necessity of identity principle---\citet[Section 4]{fit06}, for instance, discusses this possibility---and thus to model a properly intensional notion of knowledge. Indeed, the semantics presented here simply admits the possibility of interpreting identity differently at different states. The formalism is then rather simple but the flexibility obtained is on a par with that obtained by more complex semantics. 

% would simply constitute a limit case in the space of all relevant pieces of evidence: ---a rather lucky---or unlucky, depending on the case, depending on the persepctive) for that matters---   
 
%And this connection is not at all arbitrary since $\lambda$-terms can be used as witnesses of the truth of formulae in the context of formalisms---that is, .
%
%Just like, in $\lambda$-calculus \cite{chu32,chu36,bar84}, we consider program variables as programs that we suppose to have but of which we do not really know the structure. 

\paragraph{Structure of the Article} The rest of the article is structured as follows. In Section \ref{sec:logic}, we introduce the logic in which we will formalise our theory of computational trust. In particular, Section \ref{sec:language} we introduce and discuss the language of the logic, in Section \ref{sec:calculus} a natural deduction calculus characterising it, and in Section \ref{sec:semantics} its frame semantics. In Section \ref{sec:hyper}, we formally show that our epistemic operator is intensional with respect to identity and that the justification logic fragment of the system is hyperintensional.
%In Section \ref{sec:express}, we discuss some relevant issues related to the expressive power of the logic and to its peculiar features. 
 In Section \ref{sec:sound-complete}, we prove that an exact correspondence between the calculus and the semantics exists---the proof of soundness is in Section \ref{sec:soundness} and that of completeness in Section \ref{sec:completeness}. In Section \ref{sec:theory}, after the discussion in Section \ref{sec:related-work} of the existing work related to this part of the article, we present our formal theory of computational trust based on the logic just introduced.
% in Section \ref{sec:conclusion}.

\section{An Intensional Epistemic Logic}
\label{sec:logic}

Let us begin by fixing the language of the intensional epistemic  logic that we will use as basis for our formal theory of computational trust.

\subsection{The Language}
\label{sec:language}

%{ Now, the structure of justifications is completely irrelevant. If there is a justification for $A$, then $j:A$ is true for any $j$. This is fine for the intended application, but very odd, I am sure, for anyone who knows justification logic. Do we need to adopt the whole semantics with a justification admissibility function, or can we simplify it a bit? We can also just introduce a simple box-like operator for $A$ is certified. Big problems with sets of formulae for which a justification is admissible and substitution. Terrible things, very poorly thought out for first-order logic.
%
%Ultimately, we just need a reasonable way to associate a justification to its formula in such a way to be able to meaningfully relate the fact  that an agent knows that a justification exists and the fact that the agent knows that {\it a particular formula} is true. } 

The language is rather standard, exception made for the fragment defining the language of terms. Indeed, it contains the whole language of $\lambda$-calculus. The reader which is not favourably inclined towards models of computation need not worry: the language of $\lambda$-calculus is not much different from  the traditional term language of classical propositional logic featuring variables constants and function symbols. The main differences, when only pure logic---and no specific theory---is concerned, are that instead of constants we use variables, and that terms themselves can be applied to other terms and there is no need to distinguish between terms and symbols that can be applied to terms. No part of this will be particularly relevant, though, before we introduce the formal theory of trust: from a purely logical perspective,  a term is a term and all that matters is whether it is or it is not identical to another term. As for the theory, we will take care, in due time, of explaining almost from scratch the features of $\lambda$-calculus that will matter with respect to our topic. 

Here is the formal grammar that defines the language of our logic:
\begin{align*}\mathrm{Variables~} x,y,z\dots   \;\; ::=  \;\; & x_0\;\; \mid\;\;x_1\;\; \mid\;\;x_2\;\;\mid\;\;\dots 
\\
\mathrm{Terms~} t,j,s,u,v,k\dots   \;\; ::=  \;\; & x \;\;\mid\;\; ts  \;\;\mid\;\; \lambda x.t   \;\;\mid\;\; !t
\\
\mathrm{Agent~names~} a,b,c\dots   \;\; ::=  \;\; & a_0 \;\;\mid\;\; a_1  \;\;\mid\;\;  a_2   \;\;\mid\;\; \dots
\\\mathrm{Formulae~}A,B,C\dots   \;\; ::=  \;\; & \bot  \;\;\mid \;\;t\lid s\;\;\mid\;\; P(t_1,\dots ,t_n)\;\;\mid \;\;A\rightarrow B  \;\;\mid\;\; A\wedge B\\
&  \forall x. A  \;\;\mid\;\; \know _a A   \;\;\mid\;\;  j:A
\end{align*}where $P$ is any $n$-ary predicate symbol and $a$ any agent name. 

Now on, we will use $=$ for syntactical identity and $\lid$ for the identity predicate in the language. 

%We use the elements of a designated set of variables as agent names. We denote its elements  by the metavariables $a,b \dots $  

%Our term language is quite simply pure $\lambda$-calculus, which is due to the fact that the logic is primarily meant to formalise  attitudes of agents towards programs and computations. 
The selection of sentential operators extends a rather typical, functionally complete, language for first-order classical logic with identity by adding to it an epistemic operator $\know $ indexed by agent symbols $a,b,c\dots $ and a justification operator $:$ that enables us to construct formulae of the form $j:A$ where $j$ is a term and $A$ is a formula. Intuitively, $\know _a A$ is meant to express that agent $a$ knows that $A$ is true. That is, $a$ believes that $A$ is true and $A$ is indeed true. The operator is formally very similar to a knowledge operator but does not require that $a$ has any sort of reason or justification to believe $A$. This is simply due to the fact that the phenomena we wish to analyse---those related to computational trust---would not be influenced in an interesting way by additional requirements of this kind. Justification formulae $j:A$, on the other hand, are included in the language in order to express that a piece of evidence $j$ of some kind exists for the truth of $A$. As mentioned, this kind of evidence will not be used to analyse the epistemic attitude expressed by formulae of the form $\know_a A$ but, as we will see later, to define a notion of trust. The latter will be indeed analysed in terms of knowledge of the existence of a justification.

%The notion of justification that underlies the logic will not be specified much further, a few remarks about this appear in Section \ref{sec:intro}. This due to the fact that different possible interpretations seem admissible and we do not think that any benefit  would come, for the moment, from a restriction in this sense. 

We do not introduce primitive sentential operators for the negation and existential quantification but simply define $\neg A$ as $A\rightarrow \bot$ and $\exists x . A$ as $\neg \forall x. \neg A$. Moreover, we adopt the usual conventions about parentheses; we assume that $n$-ary sentential operators bind more strictly than $n+1$-ary sentential operators; and, finally, we assume that $\rightarrow $ associates to the right, that is, for instance,  $A\rightarrow B\rightarrow C$ abbreviates $A\rightarrow (B\rightarrow C)$.

As for metavariables, we will use the capital Greek letters $\Gamma , \Delta , \Theta \dots $ for multisets of formulae and the capital Greek letters $\Sigma  , \Xi \dots $ for sets of formulae. Moreover, given a multiset $\Gamma $ or set $\Sigma$ of formulae, we will use the notation $\know _a \Gamma $  and $\know _a \Sigma$ for the multiset $\{\know _a G \mid G\in \Gamma\}$ and, respectively, for the set $\{\know _a G \mid G\in \Sigma\}$. Given a multiset $\Gamma$ and set $\Delta$, we indicate by $\Gamma \subseteq \Delta $ that all elements of $\Gamma$ are also elements of $\Delta$.
%Analogously, the notation $j: \Gamma $ will indicate the multiset or, respectively,  set $\{j:G \mid G\in \Gamma\}$.

Let us now define the two substitutions that we will use in the calculus. The first one, that is $[t/x]$, will be used to obtain particular instantiations of quantified formulae. The second one, that is $(t/x)$, will be used to enforce Leibniz's law in the calculus: when $t\lid s$ and $A(t/x)$ are known to follow from certain assumptions, then $A(s/x)$ follows from them too. The distinction between the two kinds of substitution is due to the fact that  epistemic and justification contexts are supposed to be opaque with respect to identity in our logic.
\begin{definition}[Substitutions]\label{def:substitution}
Quantifier instantiating substitution $[u/v]$ is inductively defined as follows.
\begin{itemize}
\item $t[u/v]=t$ if $t\neq v$ 
\item $t[u/v]=u$ if $t=v$   
%\item $x[u/v]=x$ if $x\neq v$ 
%\item $x[u/v]=u$ if $x=v$   
%\item $ts[u/v]=t[u/v]s[u/v]$ if $ts\neq v$ 
%\item $ts[u/v]=u$ if $ts=v$  
%\item $(\lambda x. t)[u/v]=\lambda x. (t[u/v])$ if $\lambda x. t\neq v$  and $x\neq v$
%\item $(\lambda x. t)[u/v]=\lambda x. t$ if $\lambda x. t\neq v$  and $x=v$
%\item $(\lambda x.t )[u/v]=u$ if $\lambda x.t=v$
\item $\bot[u/v]=\bot$ 
\item $P(t_1,\dots ,t_n)[u/v]=P(t_1[u/v],\dots ,t_n[u/v])$ 
\item $(t_1\lid t_2)[u/v]=t_1[u/v]\lid t_2[u/v]$
\item $(A\rightarrow B)[u/v]=A[u/v]\rightarrow B [u/v]$
\item $(A\wedge B  )[u/v]=A[u/v]\wedge B [u/v]$
\item $(\forall y. A)[u/v]=\forall y. A$  if $v=y$
\item   $(\forall y. A)[u/v]=\forall y.( A[u/v])$ if $v\neq y\neq u$
\item $(\forall y. A)[u/v]=\forall z. (A[z/y][u/v])$ if $v\neq y=u$, for $z$ that does not occur in $\forall y. A$, in $u$ and in $v$
%\item $(j:A)[u/v]=t[u/v]:A[u/v]$
\item $(\know _a A)[u/v]=\know _a (A[u/v])$
\item $(j: A)[u/v]=j: (A[u/v])$
%\item $(Tw(t):A)[u/v]=Tw(t[u/v]):A [u/v]$
%\item $(Tr_t(s):A)[u/v]=Tr_t(s):A$
\end{itemize}
Substitution of identicals $(u/v)$ is inductively defined as follows.
\begin{itemize}
\item $t (u/v)=t$ if $t\neq v$ 
\item $t (u/v)=u$ if $t=v$   
%\item $x (u/v)=x$ if $x\neq v$ 
%\item $x (u/v)=u$ if $x=v$   
%\item $ts (u/v)=t (u/v)s (u/v)$ if $ts\neq v$ 
%\item $ts (u/v)=u$ if $ts=v$  
%\item $(\lambda x. t) (u/v)=\lambda x. (t (u/v))$ if $\lambda x. t\neq v$  and $x\neq v$
%\item $(\lambda x. t) (u/v)=\lambda x. t$ if $\lambda x. t\neq v$  and $x=v$
%\item $(\lambda x.t ) (u/v)=u$ if $\lambda x.t=v$
\item $\bot (u/v)=\bot$ 
\item $P(t_1,\dots ,t_n) (u/v)=P(t_1 (u/v),\dots ,t_n (u/v))$ 
\item $(t_1\lid t_2)(u/v)=t_1(u/v)\lid t_2(u/v)$
\item $(A\rightarrow B ) (u/v)=A (u/v)\rightarrow B  (u/v)$
\item $(A\wedge B  ) (u/v)=A (u/v)\wedge B  (u/v)$
\item $(\forall y. A) (u/v)=\forall y. A$  if $v=y$
\item   $(\forall y. A) (u/v)=\forall y.( A (u/v))$ if $v\neq y\neq u$
\item $(\forall y. A) (u/v)=\forall z. (A (z/y) (u/v))$ if $v\neq y=u$, for $z$ that does not occur in $\forall y. A$, in $u$ and in $v$
%\item $(j:A) (u/v)=t (u/v):A (u/v)$
\item $(\know _a A) (u/v)=\know _a A$
\item $(j: A) (u/v)=j:A$
%\item $(Tw(t):A) (u/v)=Tw(t (u/v)):A  (u/v)$
%\item $(Tr_t(s):A) (u/v)=Tr_t(s):A$
\end{itemize}
\end{definition}
As just mentioned, the substitution $(t/x)$ will be used for substituting identicals inside formulae. The reason why this substitution does not permute inside epistemic operators is that, as discussed, we wish to model possible ignorance of agents with respect to the identity of objects. The reason why this substitution does not permute inside a justification operator is that one description of the object---or, more generally, way of presenting it---might enable us to show, or even prove, that the object has a certain property while another way of presenting it might not.

Finally, let us fix some rather standard definitions about term occurrences and free variables. 
\begin{definition}[Free variables of a formula]\label{def:free-var}
A term $t$ occurs in a formula if the formula contains 
\begin{itemize}
\item a subformula of the form $P(t_1 , \dots , t_n)$ and $t=t_i$ for $i\in \{1, \dots , n\}$, or 
\item a subformula of the form $t_1\lid t_2$ and $t=t_i$ for $i\in \{1,2\}$.
\end{itemize}

The outermost occurrence of the quantifier $\forall x $ in the formula $\forall x . A $ has as range the displayed occurrence of the formula $A$.

A variable $x$ occurs free in a formula $A$ if the term $x$ occurs in $A$ and some of its occurrences are not inside the range of a quantifier of the form  $\forall x$. 
\end{definition}
Notice, in particular, that the instance of the term $t$ displayed in $t:A$ does not count as a term occurrence. It is treated just as the displayed $a$ in $\know _a A $.

\begin{definition}[Free variables of a term]\label{def:free-var-term}
An occurrence of the variable $x$ in a term $t$ is any instance of the symbol $x$ that appears in $t$ not immediately to the right of the symbol $\lambda$. 

The outermost occurrence of the binder $\lambda x $ in the term $\lambda  x . t $ has as range the displayed instance of the term $t$.

A variable $x$ occurs free in a term $t$ if $x$ occurs in $t$ and some of its occurrences are not inside the range of a binder of the form  $\lambda x$.
\end{definition}

\subsection{The Calculus}
\label{sec:calculus}

%NEW IDEA FOR INTENSIONAL/denotationAL IDENTITIES:
%We have necessitation for awareness (and trust? perhaps one can always go through awareness instead of having logical rules also for trust) but only if the derivation of the premiss does not depend on any local assumption. Only logical rules and global assumptions can be used to derive the premiss of necessitation.

%We introduce a zero-premiss rule to introduce global hypotheses that count as hypotheses but do not block necessitation. Or something of the sort.

Let us now present a calculus for the logic. We deviate from the standard practice---well-entrenched among modal logicians---of providing an axiomatic system for the logic since a natural deduction calculus induces much more natural and useful notions of derivation and proof without really hindering the soundness and completeness arguments. In any case, if the reader cannot absolutely do without an axiomatic system, we assure them that a system of this kind can be very easily read---in a standard way---from the presented rules.

%If we wish to model the fact that someone might not be aware that two expressions are intensionally equivalent, we need two equalities in the language. I will not do this though.
We present the rules of the calculus and briefly discuss them. The logical rules, shown in Figure \ref{fig:con-quant-rules}, for connectives and quantifiers are quite standard.

\begin{figure}[h]\centering
\[\infer{A\Rightarrow A}{}\quad \infer{\Gamma  \Rightarrow A\rightarrow B}{\Gamma , A \Rightarrow B}\quad\infer{\Gamma  , \Delta \Rightarrow B}{\Gamma \Rightarrow A\rightarrow  B & \Delta \Rightarrow A}\quad  \infer{\Gamma \Rightarrow P}{\Gamma \Rightarrow \bot}\quad \infer{\Gamma \Rightarrow A}{\Gamma  \Rightarrow \neg\neg A}\]
\[ \infer{\Gamma , \Delta \Rightarrow A_1\wedge A_2}{\Gamma \Rightarrow A_1& \Delta \Rightarrow A_2}\quad \infer[i\in \{1,2\}]{\Gamma \Rightarrow A_i}{ \Gamma \Rightarrow A_1\wedge A_2}\quad \infer{\Xi \Rightarrow \forall x . A }{\Xi \Rightarrow A}\quad \infer{\Gamma \Rightarrow A[t/x]}{\Gamma \Rightarrow \forall x . A}\]where $P$ is any atomic formula and $x$ does not occur free in $\Xi$

%, and 
%
%the free variable occurrences different from $y$ in $A[y/x]$ are free in $\forall x. A$
\caption{Connective and quantifier rules}\label{fig:con-quant-rules}
\end{figure}

The rules for identity are shown in Figure \ref{fig:id-rules} and differ from the usual ones only insofar as the substitution used does not permute inside epistemic and justification operators.

\begin{figure}[h]\centering
\[\infer{ \Rightarrow t\lid t}{}\qquad\infer{\Gamma \Rightarrow s\lid t}{\Gamma \Rightarrow t\lid s}\qquad \infer{\Gamma , \Delta \Rightarrow u\lid v}{\Gamma \Rightarrow u\lid t&  \Delta \Rightarrow t\lid v}\qquad \infer{\Gamma , \Delta \Rightarrow A(s/x)}{\Gamma \Rightarrow t\lid s& \Delta \Rightarrow A(t/x)} \]
\caption{Identity rules}
\label{fig:id-rules}
\end{figure}
This kind of substitution implements---minimally, one could say---a form of intensionality. Indeed, even though $s\lid t$ is assumed to be true, $\know _a A(t/x)$ is not equivalent to $\know _a A(s/x)$ unless also $\know _a s\lid t$ is assumed to be true.
%
%Notice that, if we implement a theory of program equivalence through identity, program identity  does not become universally known. Unless, obviously, we actually introduce 0-premiss rules corresponding to the axioms of this theory.  

The epistemic rules are in Figure \ref{fig:epi-rules} and the justification rules in Figure \ref{fig:just-rules}.
\begin{figure}[h]\centering
\[\infer{\know _a A_1 , \dots \know _a A_n\Rightarrow \know _a A}{\know _a A_1 , \dots \know _a A_n\Rightarrow A}
\qquad
\infer{\Gamma ,\Delta \Rightarrow \know _a B}{\Gamma \Rightarrow \know _a(A\rightarrow B) &\Delta  \Rightarrow \know _a A}
\]\[
\infer{\Gamma \Rightarrow A}{\Gamma \Rightarrow \know _a A}
\qquad
\infer{\Gamma \Rightarrow \know _a \neg \know _a  A }{\Gamma \Rightarrow \neg \know _a A }
%\qquad { CONVERSE BARCAN}
%\infer{\Gamma \Rightarrow \forall x. \know _a A}{\Gamma \Rightarrow \know _a \forall x. A}
\]
%\[ { \infer{\know _a : \know _a :  A }{\know _a :  A } \qquad \infer{\know _a : \neg \know _a  \neg A }{A }}\]
\caption{Rules for the epistemic operator}
\label{fig:epi-rules}
\end{figure}
The epistemic rules exactly correspond to an axiomatic system for the quantified modal logic S5.
% with the {\it converse Barcan formula} \cite{bar46} as an axiom.  
The first rule corresponds to the necessitation rule in presence of axioms T and 4 \citep[Chapter VI, §1]{pra06}, the second rule corresponds to axiom , the third to axiom T, the fourth to axiom 5.
%, and the fifth to the converse Barcan formula. { Technically, it is possible to derive the rule for the converse Barcan formula by a few applications of the rule for axiom T, the rules for $\forall$ and the necessitation rule.}

%The epistemic rules exactly correspond to an axiomatic system for the quantified modal logic S5 with the {\it converse Barcan formula} $(\know _a \forall x. A)\rightarrow( \forall x. \know _a A) $ as an axiom.  The first rule corresponds to the necessitation rule in presence of axioms T ($\know _a A\rightarrow  A$) and 4 ($\know _a A\rightarrow \know _a \know _a A$)  \citep[Chapter VI, §1]{pra06}; the second formula corresponds to axiom K ($\know _a (A\rightarrow B) \rightarrow \know _a A \rightarrow  \know _a B$); the third to axiom 5 ($\diamond  A \rightarrow  \square \diamond A$ in the usual notation of alethic modal logic); and the fourth to the converse Barcan formula.

%(We can make this actually hyperintensional w.r.t. terms by also requiring that no identity rule has been used in the proof of $C$. Then, also provably identical terms can be distinguished by some agent. If we add the requirement, in the semantics, the equivalence relation relative to an agent in the model should not be forced anymore to be a union of intensions.)

\begin{figure}[h]\centering
%%%%NecJustifications
%\[ 
%\infer{ j_1:A_1 , \dots , j_n:A_n \Rightarrow x^{j_1 , \dots , j_n}:A}{j_1:A_1 , \dots , j_n:A_n\Rightarrow A }\]where $x^{j_1 , \dots , j_n}$ is a designated  variable. 

\[ \infer{\Gamma , \Delta\Rightarrow jk:B}{\Gamma \Rightarrow j: (A\rightarrow B) &\Delta  \Rightarrow k: A }
\qquad 
\infer{\Gamma\Rightarrow A }{\Gamma\Rightarrow  j:A}
\qquad
\infer{\Gamma \Rightarrow {}!j:(j:A)}{\Gamma \Rightarrow j:A } 
\]
% =-closure justifications
\[ \infer{\Gamma , \Delta \Rightarrow j:A(s/x)}{\Gamma \Rightarrow k:t\lid s & \Delta  \Rightarrow j:A(t/x)  } \qquad \infer{\Gamma , \Delta \Rightarrow j:A(s/x)}{\Gamma \Rightarrow k:s\lid t & \Delta  \Rightarrow j:A(t/x)  } \]
\caption{Justification rules}
\label{fig:just-rules}
\end{figure}

The first three rules for justifications are rather standard and  correspond, respectively, to axiom K, axiom T and axiom 4. The only atypical rules are those concerning identity and they are meant to enable the substitution of terms involved in justified identities inside a justified formula $A$.

In general, for the sake of simplicity, we assume that
\begin{itemize}
\item the occurrences of formulae to the left of $\Rightarrow $ can be duplicated; 
\item if $n$ copies of the same formula occur to the left of $\Rightarrow $, then we can erase $n-1$ copies of it; and
\item any finite number of formula occurrences can be added to the left of $\Rightarrow $ at any moment in a derivation.
\end{itemize}
Since the proof-theoretical properties of the calculus are not a focus of this work, there is no need to  explicitly formalise structural rules that enable us to execute these operations.

We can finally define a notion of derivability.
\begin{definition}[Derivability]\label{def:der}
Rules are schematic and derivations are inductively defined as usual. Any instance of a zero-premiss rule is a derivation, the application of a rule to the conclusions of derivations is a derivation. For any formula $A$ and possibly infinite set of formulae $\Sigma$, $A$ is derivable from $\Sigma$ (formally, $\Sigma \vdash A$) if, and only if, there exists a finite multiset $\Gamma$ such that $\Gamma \subseteq \Sigma$ and there exists a derivation with conclusion $\Gamma \Rightarrow A $.
\end{definition}

\subsection{The Semantics}
\label{sec:semantics}

We now introduce a semantics for the logic. It is a rather standard combination of a frame semantics for epistemic logic and one for justification logic. We employ constant-domain first-order models, a family of accessibility relations $\romr_a$ for the epistemic operator $\know _a $ and an accessibility relation $\romr_\gamma$ for justifications. Identity is interpreted just as a predicate and its interpretation can vary from state to state. This will enable us to obtain epistemic contexts which are opaque with respect to identity. Justifications are handled by a family of functions $\rome _w$ that associate to any term $t$ the set of  formulae for which $t$ is an admissible justification at state $w$.    

\begin{definition}[Frame]\label{def:frame}
A Frame $\calf$ is a triple  $( \romw , \romr, \romu)$ where
\begin{itemize}
\item $\romw$ is a non-empty set, the elements of which we will call \emph{states}.
\item $\romr : \{a,b\dots \} \Rightarrow \calp(\romw\times \romw) $ is a function mapping each agent name into a binary relation between states which is reflexive and Euclidean and the symbol $\gamma $ into a binary relation between states which is reflexive and transitive.

Hence, for any agent name $a$, we have that 
\begin{itemize}
\item $\romr_a\subseteq \romw\times \romw$,
\item for any $w\in \romw$, $w\romr_aw$,
\item for any $w,u,v\in \romw$, if $w\romr_au$ and $w\romr_av$, then $u\romr_av$.
\end{itemize}

And, moreover, we have that
\begin{itemize}
\item $\romr_\gamma\subseteq \romw\times \romw$,
\item for any $w\in \romw$, $w\romr_\gamma w$,
\item for any $w,v,u\in \romw$, $w\romr_\gamma v$ and $v\romr_\gamma u$ imply that $w\romr_\gamma u$.
\end{itemize}

\item $\romu$ is a non-empty set, the domain on which quantifiers range.
%\item $\romd : \romw \Rightarrow \calp (\romu ) $ is a function mapping each state in $\romw$ to a subset of $\romu$ (hence, for any $w\in \romw$, we have that $\romd_w\subseteq \romu$) such that, if for some agent name $a$, $w\romr_a v$, then $\romd _w \subseteq \romd_v$.
\end{itemize}
\end{definition}

Intuitively, $\romw $ is the set of all possible epistemic states; $\romr_a$, for any agent $a$, associates to each epistemic state all epistemic states that are consistent with it according to agent $a$; and $\romu$ is the set of all
%---existing and non-existing---
objects that populate the epistemic states in $\romw$.

%; and $\romd$ is the function that associates to each state $w$ in $\romw $ the set of objects that populate $w$, that is, the domain of $w$.

\begin{definition}[Model]\label{def:model}
A model $\calm$ is a quintuple  $( \romw , \romr, \romu ,\romi, \rome)$ where $( \romw , \romr, \romu  )$ is a frame, $\romi$ is a function such that
\begin{itemize}
\item for any state $w\in \romw$ and term $t$,  $\romi _w (t) \in \romu $;

\item for any two states $w,v\in \romw$ and term $t$,  $\romi _w (t)=\romi _v (t) $
%\in \romu $ 

(we will sometimes abbreviate $\romi _w (t)$ by $t^w$);

\item for any state $w\in \romw$ and $e\in \romu$, there exists a term $t$ such that $\romi _w (t) =e $;

\item for any state $w\in \romw$, $\romi _w (\lid ) \subseteq \romu ^2$ such that
\begin{itemize}
\item for any $e\in \romu$, $(e,e)\in \romi _w (\lid )$,

\item for any $d,e\in \romu$, if $(d,e)\in \romi _w (\lid )$ then  $(e,d)\in \romi _w (\lid )$,  

%\item for any $c,d,e\in \romu$, if $(c,d), (d,e)\in \romi _w (\lid )$ then $(c,e)\in \romi _w (\lid ) $,

\item for any $c,d,e\in \romu$, if $(c,d), (d,e)\in \romi _w (\lid )$ then $(c,e)\in \romi _w (\lid ) $;
\end{itemize}

\item for any state $w\in \romw$ and predicate symbol $P$ of arity $n$,  $\romi _w (P) \subseteq \romu ^n$ such that, if $(e_1 , \dots  , e_n)\in \romi _w (P)$ and $(e_i,d)\in \romi_w (\lid )$ for some $1\leq i \leq n$, then $(e_1' , \dots , e_n')\in \romi _w (P)$ where, for any $1\leq j\leq n$,  $e_j '=e_j$ in case 
{\ehi $j\neq i$ and $e_j '=d$ in case $j= i$;}
%$e_j\neq e_i$ and $e_j '=d$ in case $e_j= e_i$;

%(Notice that the replacement is uniform, that is, we replace by $d$ all $e_j$s such that $e_j= e_i$ and not only the $i$th element of $e_1 , \dots  , e_n$.)
\end{itemize}
and $\rome$ is a function such that, for any state $w\in \romw$, the following hold:
\begin{itemize}

\item for any state $v\in \romw$ and term $t$, if $w\romr _\gamma v$, then $\rome_w(t)\subseteq \rome_v(t)$;

%%%%NecJustifications
%\item  { if $j_1:A_1 ,  \dots , j_n:A_n \vdash A$ holds, then $A\in \rome_w(x^{j_1 , \dots , j_n})$ for the designated variable $x^{j_1 , \dots , j_n}$,}

%\item for any two terms $t,s$, if $(\romi_w(t),\romi_w(s))\in \romi_w(\lid)$, then $ \rome_w(t)=\rome_w(s)$, 

%\item for any term $t$, if $(\romi_w(t),\romi_w(s))\in \romi_w(\lid)$ and $ A[t/x]\in \rome_w(t)$, then $ A[s/x]\in \rome_w(t)$,

%\item {\ehi for any term $j$, if $A\in \rome_w(j)$, then $A[x/y] \in\rome_w(j)$ for any variables $x$ and $y$;}

\item {\ehi for any term $j$, if $A\in \rome_w(j)$, $A=B[t/x]$ and  $\romi _w (t)=\romi _w (s)$, then $B[s/x]\in \rome_w(j)$;}

\item for any term $j$, if $A\in \rome_w(j)$, then $ j:A \in\rome_w(!j)$;

\item for any two terms $j,k$, if $ A\rightarrow B \in \rome_w(j)$ and $ A \in \rome_w(k)$, then $ B\in \rome_w(jk)$;

% =-closure justifications
\item for any term $j$, if $A(t/x)\in \rome_w(j)$ and there exists a $k$ such that $
t\lid s \in \rome_w(k)$ or $
s\lid t \in \rome_w(k)$, then $A(s/x)\in \rome_w(j)$.
\end{itemize}

For any model $\calm=( \romw , \romr, \romu , \romi , \rome)$, we call $\romi$ the interpretation function of $\calm$ and $\romi _w $ the interpretation function at the state $w$ of $\calm$. Moreover, we say that $\calm$ is \emph{based on the frame} $ \calf=( \romw , \romr, \romu )$.
\end{definition}

Notice that the conditions on $\rome$ are consistent since one can always, for any term $t$ and $w\in \romw$, take the set of all formulae as $\rome_w(t)$.

\begin{definition}[Truth at a state]\label{def:forcing}
For any formula $A$, model $\calm = ( \romw , \romr, \romu ,  \romi)$ and state $w\in \romw$, the fact that a formula $A$ is true at the state $w$ of model $\calm$  (formally, $\calm , w ,\cani _{w} \Vdash A $) is defined by the following clauses.
\begin{itemize}
\item $\calm , w ,\cani _{w} \Vdash t_1 \lid t_2$ if, and only if, $(\cani _{w}(t_1) , \cani _{w}(t_2))\in \cani _{w}(\lid)$

\item $\calm , w ,\cani _{w} \Vdash P(t_1, \dots , t_n)$ if, and only if, $(\cani _{w}(t_1), \dots , \cani _{w}(t_n))\in \cani _{w}(P)$  

\item $\calm , w ,\cani _{w} \not\Vdash \bot$ for any model $\calm$, state $w$ and assignment $f$

%\item { REMOVE: $\calm , w ,\cani _{w} \Vdash \neg A$ if, and only if, $\calm , w ,\cani _{w} \not\Vdash  A$}

%\item $\calm , w ,\cani _{w} \Vdash  A\vee B $ if, and only if, $\calm , w ,\cani _{w} \Vdash  A$ or $\calm , w ,\cani _{w} \Vdash  A$

\item $\calm , w ,\cani _{w} \Vdash  A\wedge B $ if, and only if, $\calm , w ,\cani _{w} \Vdash  A$ and $\calm , w ,\cani _{w} \Vdash  B$

\item $\calm , w ,\cani _{w} \Vdash  A\rightarrow B $ if, and only if, $\calm , w ,\cani _{w} \Vdash  A$ only if $\calm , w ,\cani _{w} \Vdash  B$

%\item $\calm , w ,\cani _{w} \Vdash  \exists x. A $ if, and only if, $\calm , w , g \Vdash  A$, for some $x$-variant $g$ of $f$

\item {\ehi $\calm , w ,\cani _{w} \Vdash  \forall x. A $ if, and only if, $\calm , w ,\cani _{w}\Vdash  A[t/x]$, for any term $t$}

%\item $\calm , w ,\cani _{w} \Vdash  \forall x. A $ if, and only if, $\calm , w , g \Vdash  A$, for any $x$-variant $g$ of $f$

\item $\calm , w ,\cani _{w} \Vdash  \know _a A $ if, and only if, $\calm , v , \cani_v \Vdash  A$, for any $v\in \romw $ such that $w\romr _a v$. 
%and such that $f_{w\hookrightarrow v} (t)\in \romd_v$ for any $t$ occurring in $A$.

\item $\calm , w ,\cani _{w} \Vdash  j: A $ if, and only if, $\calm , v , \cani_v \Vdash  A$, for any $v\in \romw $ such that $w\romr _\gamma  v$, 
%and, moreover, for some list of terms $t_1\dots t_n$ it holds that $A[t_1/x_1\dots t_n/x_n]\in \rome_w(j)$ where $x_1,\dots , x_n$ are all free variables of $A$ and, for $i\in \{1,\dots , n\}$, $f(t_i)= f(x_i)$.
{\ehi and, moreover, it holds that $A\in \rome_w(j)$.}
\end{itemize} 
For any formula $A$ and state $w$, if $A$ is true at $w$, we say that $w$ verifies $A$. Moreover, if $w$ verifies all elements of $\Gamma$, we say that $w$ verifies $\Gamma$.
\end{definition}
%Even though a bit complicated, the last condition on the truth of $j:A$ at $w$ under the assignment $f$ simply guarantees that all formulae that have the same propositional and quantificational structure, and all their terms interpreted in the same way by $f$ can be justified by the same terms---as specified by $\rome_w(j)$. That is, if $j$ can be a justification for $A$ at $w$ under the assignment $f$, then any formula $A'$, obtained from $A$ by substituting some terms for variables that have the same interpretation as these terms under $f$, can also be justified by $j$. For instance, if $f(x)=f(t)\in \romu$, then $P(x)$ and $P(t)$ can be justified by the same justifications at $w$ under $f$. 
%
%
%It is easy to see that this weakening of the obvious condition $A\in \rome_w(j)$ is necessary and rather natural if one thinks that 
%a free variable $x$ occurring in a formula $A$ under an $x$-variant $f$ of some interpretation function is supposed to play the {\it semantical role}, so to say, of the object $f(x)\in \romu $ rather than that of $\romi _w(x)$ itself. Notice, moreover, that this cannot be directly captured by the  content of $\rome_w(j)$ since this set must be  syntactically defined before the definition of the assignment $f$ is known.

\begin{definition}[Satisfiability, logical consequence and validity]\label{def:sem-rel}
%For any formula $A$, model $\calm = ( \romw , \romr, \romu , \romd, \romi)$, $A$ is satisfiable on the model $\calm$ (formally, $ \calm\Vdash A$) if, and only if, there exists a state $w\in \romw$ such that $\calm , w, \romi _w\Vdash A $.

For any formula $A$, $A$ is satisfiable if, and only if, there exists a model $\calm = ( \romw , \romr, \romu , \romi)$ and a state $w\in \romw$ such that $\calm , w, \romi _w\Vdash A $.

%For any formula $A$ and model $\calm$,  $A$ is valid on the model $\calm$ (formally, $ \calm\vDash A$) if, and only if, for any state $w\in \romw$ it holds that $\calm , w, \romi _w\Vdash A $.
For any set of formulae $\Gamma $ and formula $A$, $A$ is a logical consequence of $\Gamma$ (formally, $\Gamma \vDash A$) if, and only if, for any model $\calm = ( \romw , \romr, \romu ,  \romi)$ and state $w\in \romw$, if $\calm , w, \romi _w\Vdash G $ for each $G\in \Gamma$, then $\calm , w, \romi _w\Vdash A$.

For any formula $A$, $A$ is valid (formally, $\vDash A$) if, and only if, for any model $\calm = ( \romw , \romr, \romu ,  \romi)$ and state $w\in \romw$, $\calm , w, \romi _w\Vdash A $.

\end{definition}

%\subsection{Remarks on Expressivity}
%\label{sec:express}
%
%
%{ Semantically, an intensional identity will simply be an identity of which all agents know that it holds, either because the identity follows from rules for identity or because the identity is globally known for contingent reasons. We do not care. This, indeed, (unfortunately?) agrees with the reading according to which intensionally equivalent formulae are modally equivalent formulae. But this works only if agents agree on intensions, that is, if no agent can have an equivalence class which is smaller than an intension. Proof-theoretically, this is the case when necessitation can be used on any theorem, and not only on those proved without using identity rules. Then intensions are closed with respect to provable identity and agents' equivalence classes are unions of intensions.
%
%Semantically, a denotational but not intensional identity is simply an identity which holds but of which some agent is not aware.
%}

\section{Intensionality of $\know$ and Hyperintensionality of Justifications}
\label{sec:hyper}

In Section \ref{sec:intro}, we stated that we wanted to formally define a logic in which the justification logic fragment is hyperintensional and the epistemic logic fragment enables us to distinguish between extensional and intensional identity. Now that we have defined the system, we can formally show that it is indeed a system of this kind.

Let us begin with the hyperintensionality of our justification logic fragment. First of all, we need to formally fix a notion of hyperintensionality. We adapt the one by \citet{cre02} to the case at hand. We will simply say that an operator $\romo $ is hyperintensional if, and only if, there exist two formulae $A$ and $B$, a model $\calm= ( \romw , \romr, \romu , \romi,\rome)$ and a state $w\in \romw$ such that
\begin{itemize}
\item $A\vDash B$,
\item $B\vDash A$,
\item $\calm , w , \romi _w\Vdash \romo A$
\item $\calm , w  , \romi _w \not\Vdash \romo B$
\end{itemize}
Let us then show that $j:\;$ is indeed hyperintensional. Consider the formulae $P(x)\wedge Q(y)$ and $Q(y)\wedge P(x)$. It clearly holds that $P(x)\wedge Q(y)\vDash Q(y)\wedge P(x)$ and $Q(y)\wedge P(x)\vDash P(x)\wedge Q(y) $. Consider then any model $\calm= ( \romw , \romr, \romu , \romi,\rome)$ and state $w\in \romw$ such that 
\begin{enumerate}
\item $I_v(x)\in I_v(P)$ for any $v\in \romw$, 
\item $I_v(y)\in I_v(Q)$ for any $v\in \romw$,
\item $P(x)\wedge Q(y)\in \rome _v (j)$  for any $v\in \romw$, and 
\item  \label{proof-check-only} $ A \notin \rome _v (t)$ for any $v\in \romw$, term $t$ and formula $A$ unless $A=!\dots ! j: \dots : j:P(x)\wedge Q(y)$ where the terms occurring in $:\dots :$ are of the form $!\dots ! j$.
\end{enumerate}Intuitively, the model is defined in such a way that $P(x)\wedge Q(y)$ is true at all states. Moreover, as for admissible justifications, $j$ is admissible for $ P(x)\wedge Q(y)$ at $w$, any justification obtained by applying possibly several times the proof checker $!$ to $j$ is admissible for the corresponding formula $!\dots ! j: \dots : j:P(x)\wedge Q(y)$. No other terms is an admissible justification for any other formula.

Let us first argue that the model $\calm$ is well defined. The only possible problem is the definition of $\rome $ and, in particular, the possibility of defining it in such a way that only formulae of the form $!\dots ! j: \dots : j:P(x)\wedge Q(y)$ have admissible justifications. First, the only conditions on  $\rome $ that force us to consider terms as admissible justifications are the one stating that, for any term $j$, if $A\in \rome_w(j)$, then $ j:A \in\rome_w(!j)$ and the one stating that, for any two terms $j,k$, if $ A\rightarrow B \in \rome_w(j)$ and $ A \in \rome_w(k)$, then $ B\in \rome_w(jk)$. The latter is irrelevant for us since no implication has an admissible justification in the model. As for the second one, if we consider that the only justification that we need to be admissible is $j$, the definition of $\calm$ complies with it due to \cref{proof-check-only} above. The model $\calm$ is therefore well defined.

Clearly, then, the formulae $P(x)\wedge Q(y)$, $ Q(y) \wedge P(x)$ and $j:P(x)\wedge Q(y)$, the model $\calm$ and the state $w\in \romw$ show that justification operators $t:$ are hyperintensional. Indeed, $j:P(x)\wedge Q(y)$ is true at $w$ while no justification of $Q(y) \wedge P(x)$ is admissible at $w$.

Let us now show that any epistemic operator $\know _a$ is intensional with respect to identity. That is, there exist two terms $t$ and $s$, a formula $A$ in which $x$ occurs free, a model $\calm= ( \romw , \romr, \romu , \romi,\rome)$ and a state $w\in \romw$ such that
\begin{itemize}
\item $\calm , w , \romi _w\Vdash t\lid s$,
\item $\calm , w , \romi _w\Vdash \know _a A[t/x]$, and
\item $\calm , w , \romi _w\not\Vdash \know _a A[s/x]$
\end{itemize}

Consider the formula $P(x)$ and any model $\calm= ( \romw , \romr, \romu , \romi,\rome)$ and state $w\in \romw$ such that
\begin{enumerate}
\item \label{wid} $(\romi _w(t), \romi _w(s)) \in \romi _w(\lid )$,
\item  \label{allpt} $\romi _v(t)\in\romi _v (P)$ for any $v\in \romw$,
\item there exists a $v\in \romw$ such that 
\begin{enumerate}
\item  \label{vdiff} $v\neq w$,
\item \label{vrel} $w\romr_a v$,
\item \label{vnid} $(\romi _v(t), \romi _v(s)) \not\in \romi _v(\lid )$, and 
\item \label{vnps} $\romi _v(s)\not\in\romi _v (P)$.
\end{enumerate}
\end{enumerate}Clearly, we have that $\calm , w , \romi _w\Vdash t\lid s$ and $\calm , w , \romi _w\Vdash \know _a P(t)$ hold because of \cref{wid} and \cref{allpt} above. Nevertheless, due to \cref{vrel}, \cref{vnid} and \cref{vnps}, we also have that $\calm, w, \romi _w \not\vDash \know _a P(s)$.

Hence, the  operator $\know _a $, for any agent name $a$, is intensional with respect to identity. In other words, the operator $\know _a $ enables us to distinguish between extensional identity and intensional identity: the fact that $t\lid s$ is true at a state---extensional identity---is not equivalent to the fact that $t\lid s$ is true at all accessible states---intensional identity. In the second case, indeed, we would have that also $\know _a t \lid s $ is true at $w$ and hence---it is easy to see---that $\know _a P(s)$ is true at $w$.

Notice, finally, that both characterisations of intensionality and hyperintensionality used above can be rephrased in purely proof-theoretical terms. We can do so by employing $\vdash $ instead of $\vDash $ to express the equivalence of formulae and the provable identity of terms, and by employing derivability from hypotheses instead of truth at a state. As for the presented arguments, even though it is easier to discuss them in semantical terms since they have the conceptual form of counterexamples.\footnote{To show that the operators are intensional or hyperintensional, indeed, boils down to show that something does not have to hold in all cases in which certain assumptions hold.}, they can be formally run also in the presented proof-theoretical framework. This is formally guaranteed by the soundness and completeness results that we will prove in Section \ref{sec:sound-complete}

\section{Soundness and Completeness}
\label{sec:sound-complete}
We now prove that there exists an exact correspondence between the natural deduction calculus and the semantics of the logic. We will show, in particular, that derivability and logical consequence perfectly match: $\Gamma \vdash A$ if, and only if, $\Gamma \vDash A$.

\subsection{The Soundness of the Calculus}
\label{sec:soundness}

In order to prove the soundness of the calculus, we will read an expression of the form $\Gamma\Rightarrow A $ as a judgement stating that all states that verify all elements of $\Gamma$ must also verify $A$; and then show that, if the premisses of any rule hold under this reading, then also its conclusion holds.

\begin{theorem}\label{thm:der-preservation}
If  $\Gamma \Rightarrow A $ is derivable, then all states that verify $\Gamma$ also verify $A$.
\end{theorem}\begin{proof}
The proof is by induction on the number of rules applied in the derivation of $\Gamma \Rightarrow A$. See Appendix \ref{appendix}.
\end{proof}
\begin{corollary}[Soundness]\label{cor:soundness}
For any possibly empty multiset of formulae $\Gamma$ and formula $A$, if $\Gamma \vdash A $, then $\Gamma \vDash A$. 
\end{corollary}
\begin{proof}By Theorem \ref{thm:der-preservation} and Definitions \ref{def:der} and \ref{def:sem-rel}. 
\end{proof}

\subsection{The Completeness of the Calculus}
\label{sec:completeness}

The completeness proof is reasonably similar to a traditional proof by canonical model construction---see, for instance, \citet{brv01}---and several of its details are inspired by the general framework introduced by \citet{cor02}. Nevertheless, the choice of providing it directly for a natural deduction calculus and other technical subtleties---such as the fact that we formalise a subjective (agent-relative) identity with respect to which epistemic and justification contexts are opaque---makes it impossible to simply apply existing proof techniques.

The idea behind the proof, as usual, is the definition of a model $\canm$ that contains a counterexample for each non provable formula and incorrect derivability judgement, that is, a state that falsifies the formula or the judgement. This proves the contrapositive of the completeness statement: if $\Gamma\not\vdash A$, then $\Gamma\not\vDash A$. The model $\canm$ is usually called {\it canonical model}, its set of states is defined as the set of all sets of formulae which are maximally consistent with respect to the logic, and the interpretation function and of accessibility relation definitions guarantee that the formulae {\it true at} some state $w$ are exactly those {\it contained in} $w$ if we see it as a maximally consistent set of formulae.

To this basic structure of the proof, one needs to add details and requirements related to the particular logic under consideration. In our case, it is enough to require that the set of states of $\canm$ is the set of all maximally consistent sets of formulae that contain 
$\neg A[t/x]$, for some $t$, whenever they contain $\neg \forall x. A$. We call these sets  {\it counterexemplar}. This condition is  completely analogous to the one required to handle existential quantifiers in first-order logic.

Let us begin with the usual definition of maximally consistent sets of formulae.

\begin{definition}[Maximally consistent set of formulae]\label{def:max-cons}
A set of formulae $\Sigma$ is consistent if and only if, $\Sigma \not\vdash \bot$. Rather unsurprisingly, if a set is not consistent, we say that it is inconsistent.

A set of formulae $\Sigma$ is maximally consistent if, and only if, $\Sigma$ is consistent and, for any formula $A\notin \Sigma$, $\Sigma\cup \{A\}$ is inconsistent.
\end{definition}

We now formally define the notion of {\it counterexemplar} set. 

\begin{definition}[Counterexemplar set of formulae]\label{def:counterex}
A set of formulae $\Sigma$ is counterexemplar  if, and only if, for any formula $\neg \forall x. A\in \Sigma$, there exists a term $t$ such that $\neg A[t/x]\in \Sigma$.
\end{definition}

We can now show that we can always obtain a maximally consistent counterexemplar set (MCC) by extending a consistent one. This will guarantee that, for any non provable formula $A$, there exists a state in $\canm$ that contains---and hence verifies---$\neg A$.

\begin{lemma}\label{lem:ultrafilter}
For any consistent set of formulae $\Sigma$, there exists a maximally consistent, counterexemplar  set (MCC) of formulae $\Sigma^*$ such that $\Sigma \subseteq \Sigma ^*$.
\end{lemma}
\begin{proof} The proof is an adaptation of the modal logic version of the classical proof by \citet{hen49}. 

We construct an infinite chain $\Sigma _0\subset \Sigma _1 \subset \Sigma _2 \subset  \Sigma _3\dots $ of consistent sets of formulae
%%%NEW FOR CORRECTION
by standardly using a countable  set of new variables. Let $\Sigma _0=\Sigma $ and, given any ordering $A_1, A_2, A_3\dots $ of all the formulae of our language, we define the set $\Sigma _n$, for $n>0$, as follows:\begin{itemize} 
\item $\Sigma _n = \Sigma _{n-1}\cup \{A_n\}$ if $\Sigma _{n-1}\cup \{A_n\}$ is consistent,

\item $\Sigma _n = \Sigma _{n-1}\cup \{\neg A_n\}$ if  $\Sigma _{n-1}\cup \{A_n\}$ is not consistent and either $A_n\neq \forall x.B$ for any variable $x$ and formula $B$, or $A_n= \forall x.B$ and $\neg B[t/x] \in \Sigma _{n-1}$ for some term $t$,

\item $\Sigma _n = \Sigma _{n-1}\cup \{\neg A_n , \neg B[y/x]\}$ for some variable $y$ that does not occur in $\Sigma _{n-1}\cup \{\neg A_n\}$ if  $\Sigma _{n-1}\cup \{A_n\}$  is not consistent, $A_n= \forall x.B$, and, for each term $t$, $B[t/x] \notin \Sigma _{n-1}$.
\end{itemize}

Let\[\Sigma ^*=\bigcup_{i\in \mathbb{N}} \Sigma _i\]

See Appendix \ref{appendix} for the proofs of the facts that the obtained set is maximally consistent and counterexemplar. 
\end{proof}

We show that MCCs always contain either a formula or its negation. This will be very often used in the following proofs. One might say that this is the most precious feature of MCCs.

\begin{lemma}\label{lem:maximality}
For any formula $A$ and any maximally consistent, counterexemplar set of formulae (MCC) $\Sigma$, either $A\in \Sigma$ or $\neg A\in \Sigma$.
\end{lemma}
\begin{proof} Suppose, by reasoning indirectly, that, for some formula $A$, neither $A$ nor $\neg A $ is an element of  $\Sigma $. Since $\Sigma $ is maximally consistent, we must have that $\Sigma \cup \{A\}$ is inconsistent. By Definition \ref{def:max-cons},  there must exist a derivation with conclusion $\Gamma , A \vdash \bot$ for a multiset $\Gamma $  of elements of $\Sigma $. By an implication introduction, then, we can construct a derivation with conclusion $\Gamma \vdash \neg A$. 
Since $\Sigma $ is maximally consistent, we must have also that $\Sigma \cup \{\neg A \}$ is inconsistent. By Definition \ref{def:max-cons},  there must exist a derivation with conclusion $\Delta  , \neg A \vdash \bot$ for a multiset $\Delta $  of elements of $\Sigma $. By an implication introduction, then, we can construct a derivation with conclusion $\Delta  \vdash \neg\neg  A$. Since we have a derivation with conclusion  $\Gamma \vdash \neg A$ and one with conclusion $\Delta  \vdash \neg\neg  A$, by an implication elimination, we can construct a derivation with conclusion $\Gamma, \Delta  \vdash \bot$. Since $\Gamma \cup \Delta \subseteq \Sigma $, this means that $\Sigma $ is inconsistent.  This goes against the fact that $\Sigma $ is maximally consistent and thus we can conclude, as desired, that either $A$ or $\neg A $ must be an element of $\Sigma $. 
\end{proof}

Another very useful feature of MCCs is that they contain all formulae that can be derived from their elements. Let us prove it.
\begin{lemma}\label{lem:closure}
Maximally consistent, counterexemplar sets of formulae (MCCs) are closed with respect to derivability.
\end{lemma}
\begin{proof} 

Suppose, by reasoning indirectly, that this is not the case. That is, even though $\Gamma \vdash A$,  there is an MCC $\Sigma$ such that $\Gamma\subseteq \Sigma$ but $A\notin \Sigma$. By Lemma \ref{lem:maximality}, since we have $A\notin \Sigma$, we must have  $\neg A \in \Sigma$. Since $\Gamma \vdash A$, it is possible, by implication elimination, to construct a derivation with conclusion $\Gamma , \neg A \Rightarrow \bot $. Since, moreover, $\neg A \in \Sigma$ and  $\Gamma \subseteq \Sigma$, we can thus show that $\Sigma\vdash \bot $. But this contradicts the assumption that $\Sigma$ is an MCC, because MCCs are consistent. Hence, if $\Sigma$ is an MCC, $\Gamma\subseteq \Sigma$ and $\Gamma \vdash A$, then $A\in\Sigma$.
\end{proof}

We can finally define the canonical model $\canm$ and actually prove that it is a model of our logic. Now, the canonical model is defined by following the usual intuition: each state of the model is a maximally consistent---and, in our specific case, {\it counterexemplar}---set of formulae; the domain of the model should contain, for each term $t$ of the language, an object $\underline{t}$; the accessibility relation and the interpretation function should be defined in such a way that each state verifies all its elements and no other formula. Additionally, we have to specify the function $\cane$, which indicates, for each state $w$ and term $j$, the set of formulae for which $j$ is an admissible justification at $w$. This function is defined by following the same intuition that guides the definition of the interpretation function; that is, the interpretation of a predicate $P$ includes $(\underline{t_1} , \dots  , \underline{t_n})$ exactly when $P(t_1 , \dots  , t_n)$ is an element of the state, and the function $\cane$ establishes that  $j$ is an admissible justification for the formula $A$ at a state exactly when $j:A$ is an element of the state. This, quite clearly, enable us to show that an atomic formula is true at a state exactly when the state contains it and to show---along with the definition of $\canr_\gamma$---that an analogous fact holds for justification formulae. Finally, the definitions of accessibility relations are meant to comply with the fact that each state must verify all formulae that are justified or known---depending on the accessibility relation---at the predecessor states. Clearly, these definitions cannot do this by directly requiring that a certain state {\it verifies} certain formulae. This would be circular since $\Vdash $ is defined in terms of $\romr_\gamma $ and $\romr_a$. Hence, instead of directly mentioning $\Vdash $, the definitions of $\canr_\gamma$ and $\canr_a$ simply require that certain states {\it contain} certain formulae: if a state contains $j:A$, for instance, then all its  $\canr_\gamma$-successors must contain $A$. As we mentioned, we will then be able to prove that, for each formula, being an element of a state $w\in \canw$ and being true at $w$ coincide.

\begin{definition}[Canonical Model]\label{def:canonical-model}
The canonical model $\canm$ is the quintuple $( \canw , \canr, \canu ,\cani, \cane)$ where
\begin{itemize}
\item $\canw$ is the set of all maximally consistent, counterexemplar sets of formulae;
\item $\canr$ is defined as follows: 
\begin{itemize}
\item for any agent name $a$ and $w,v\in \canw$, $w\canr_a v$ if, and only if, $\know_a ^-(w)\subseteq v$ where $\know_a ^-(w)=\{A\mid \know _a  A \in w\}$, and 
\item for any $w,v\in \canw$, $w\canr_\gamma v$ if, and only if, $\gamma ^-(w)\subseteq v$ where $\gamma ^-(w)=\{A\mid j:A \in w\}$ for any justification $j$;
\end{itemize}

\item $\canu$ is the set $\{\underline{t}\mid t \text{ is a term}\}$;

\item $\cani$ is the function such that 
\begin{itemize}
\item for any state $w\in \canw$ and term $t$,  $\cani _w (t)=\underline {t}\in \romu ^\romc$,

\item for any state $w\in \canw$, $(\underline{t_1}, \underline{t_2})\in \cani _w (\lid )$ if, and only if, $t_1\lid t_2\in w$, 

\item for any state $w\in \canw$ and predicate symbol $P$ of arity $n$,  $(\underline{t_1} , \dots  , \underline{t_n})\in\cani _w (P)$ if, and only if,  $P(t_1 , \dots  , t_n)\in w $;
\end{itemize}

\item $\cane$ is the function such that 
\begin{itemize}
\item for any state $w\in \canw$, term $j$ and formula $A$,  $A\in \cane _w (j)$ if, and only if, $j:A\in w$.
%\item for any state $w\in \canw$ and formula $A$, if $j_1:A_1,  \dots , j_n:A_n \vdash A$ holds, then $A\in \cane _w (x^{j_1 , \dots , j_n})$.
\end{itemize}
\end{itemize}
\end{definition}

\begin{lemma}\label{lem:canonical-model-model} The canonical model $\canm$ is a model.
\end{lemma}
\begin{proof}The proof is quite long but no surprising technique is employed. See Appendix \ref{appendix}.
\end{proof}

\begin{lemma}[Existence lemma]\label{lem:existence}
For any $w\in \canw$, if $\neg \know  _a A\in w $, then there exists a $v\in \canw$ such that $\neg A\in v$ and $w\canr_a v$.
%%NecJustifications
%For any $w\in \canw$, 
%\begin{itemize}
%\item if $\neg \know  _a A\in w $, then there exists a $v\in \canw$ such that $\neg A\in v$ and $w\canr_a v$; 
%\item if $\neg j: A\in w $, then there exists a $v\in \canw$ such that $\neg A\in v$ and $w\canr_\gamma v$.
%\end{itemize}
\end{lemma}
\begin{proof} 
%%NEW FOR CORRECTION
The proof is an adapted combination of the proofs in \citep[Theorem 4.20]{brv01} and \citep[Theorem 14.1]{hughes-cresswell-1996}. See Appendix \ref{appendix}.
\end{proof}

We now show that the states of $\canm$ verify exactly the formulae that they contain if we see them as MCCs. 

\begin{lemma}[Truth lemma]\label{lem:truth}For any $w\in \canw$ and formula $A$, $A\in w$ if, and only if, $\canm , w, \cani_w \Vdash A$.
\end{lemma}
\begin{proof} The proof is by induction on the structure of $A$. The argument is rather standard, it often employs Lemma \ref{lem:maximality}, and Lemma \ref {lem:existence} is crucial for the case in which  $A=\know_a B$. The argument for the case in which $A=j:B$, on the other hand, almost exclusively relies on the fact that $w$ verifies $j:B$ if, and only if, $B\in \cane_w(j)$ and on the fact that $j:B\in w$ if, and only if, $B\in \cane_w(j)$. See Appendix \ref{appendix} for all the details.
\end{proof}

We can finally prove that the completeness of the calculus follows from the previous results and definitions.

%
%{\ehi
%
%\begin{theorem}[Completeness]\label{thm:completeness}
%If $\Gamma\vDash A$ then $\Gamma\vdash A$ ($\Gamma$ a finite multiset).
%\end{theorem}
%\begin{proof}
%Contrapositive: let $\Gamma\nvdash A$. Then $\Gamma\cup\{\neg A\}$ is consistent (if $\Gamma,\neg A\Rightarrow\bot$ then $\Gamma\Rightarrow\neg\neg A$ and, by double negation elimination, $\Gamma\Rightarrow A$).
%This set is finite, so only finitely many variables occur in it and Lemma~\ref{lem:ultrafilter}(i) gives an MCC $w\supseteq\Gamma\cup\{\neg A\}$; $w\in\canw$. By Lemma~\ref{lem:canonical-model-model} $\canm$ is a model, and by Lemma~\ref{lem:truth}, $\canm,w,\cani_w\Vdash G$ for all $G\in\Gamma$ and $\canm,w,\cani_w\Vdash\neg A$,
%so $\canm,w,\cani_w\not\Vdash A$. Hence $\Gamma\nvDash A$.
%\end{proof}
%
%
%}

\begin{theorem}[Completeness]\label{thm:completeness}
For any multiset of formulae $\Gamma $ and formula $A$, if $\Gamma \vDash A $, then $\Gamma \vdash A$.
\end{theorem}

\begin{proof}
We prove the contrapositive: if $\Gamma \not\vdash A$, then $\Gamma \not\vDash A $.

In order to show that $\Gamma \not\vDash A $, we need to find a state in model that verifies each element of  $\Gamma$ and does not verify $A$. That is, a state that verifies $\Gamma \cup \{\neg A\}$. If $\Gamma \not\vdash A$, then we know that $\Gamma \cup \{\neg A\}$ is a consistent set of formulae. Indeed, otherwise, we could derive $\bot$ from $\Gamma \cup \{\neg A\}$ and by applying a double negation elimination to the derivation witnessing the derivability of $\bot $ from $\Gamma \cup \{\neg A\}$ we could show that $\Gamma \vdash A$. 
Since $\Gamma \cup \{\neg A\}$ is consistent, by Lemma \ref{lem:ultrafilter}, we know that there exists an MCC $\Sigma$ such that $\Gamma \cup \{\neg A\}\subseteq \Sigma$. By Definition \ref{def:canonical-model}, we know that $\Sigma=w\in \canw$. By Lemma \ref{lem:truth}, moreover, we know that $\canm , w, \cani_w\Vdash G $ for each $G\in \Gamma$ and  $\canm , w, \cani_w \not\Vdash A$. Hence, $w$ is exactly what we needed: a state in a model that verifies $\Gamma$ and falsifies $A$. 
\end{proof}

\section{Towards a Theory of Computational Trust}
\label{sec:theory}

As already shown in Section \ref{sec:hyper}, even though the presented system is quite simple and employs reasonably traditional mechanisms---that is, the machinery of epistemic logic and justification logic---it possesses rather interesting and non trivial formal features. As a consequence, its expressive power enables us to draw distinctions concerning extensional, intensional and hyperinstensional identities that are similar to those that we can draw in much more complex and expressive logical systems. This is particularly interesting with respect to the problem of formalising intensional and hyperintensional properties and relations. Indeed, the system enables us to formalise properties and relations of these kinds without the burden of the additional notation required by more complicated systems such as those introduced by \citet{mon69,mon70,fit06,gal16,fm23}. In other words, the presented system provides us with the expressive power required to formalise certain intensional and hyperintensional properties and relations but nothing more, thus yielding a reasonably simple and yet expressive enough framework for formalising them. 

In order to exemplify an application of this kind, we will now consider the problem of formally characterising a relation of computational trust. More detailed and in-depth applications of the system to similar problems are left for future work. Before discussing the technical details related to the formalisation in our system of a predicate for computational trust, though, let us review existing works sharing a similar aim.

\subsection{Related work: formalisation of trust}
\label{sec:related-work}

Before presenting our formalisation of trust in the presented logical system, let us briefly review the existing works concerning formal analyses of {\it trust}, {\it trustworthiness} and related negative notions such as {\it negative trust}, {\it mistrust} and {\it distrust}. The formal endeavours concerning the study of notions of trust in computational settings can be divided in three classes: one concerning the design of formal means to implement trust attitudes in computational systems involving several agents, one concerning the  definition of formal systems for assessing the trustworthiness of computational agents, and one concerning the definition of logical system to reason about the interplay among epistemic attitudes, information exchange and trust-related attitudes. The assignment of a work to a particular class can obviously be non-exclusive, but the obvious reference class of the system introduced in the present work is the third one. 

As for the design of formal means to implement trust attitudes in computational systems, we can certainly mention the system introduced by \citet{arh98}. This work focuses on distributed systems---computational systems consisting of different, more or less autonomous, agents that can communicate and possibly interact in different ways---and presents a formal language that enables the agents of a distributed system to exchange information about the reputation of other agents and their trustworthiness.

%What we aim at here is instead a logical framework in which we can represent and reason about the knowledge of agents and their trust attitudes.

As far as trustworthiness assessment by formal means is concerned, \citet{zl05} presented a class of network-structured  models for the study of the propagation of trust and ditrust in social networks. More recently, a logic approach to the same problem has been adopted by \citet{pcd23}. On the basis of a language similar to the one by \citet{pr14}---which we will discuss later--- \citet{pcd23} introduced a formal model for assessing the trustworthiness of agents on the basis of their interactions.

We now turn to those systems that have been introduced in order to reason about the interplay among epistemic attitudes, information exchange and trust-related attitudes.
First of all, an axiomatic system of this kind has been introduced by \citet{dem04} in order to formally capture properties of trust in a language featuring a belief and a knowledge operator, and an operator indicating that an agent informs another agent about some fact. While the work explores several possible features of the trust relation and ways of formalising them, no semantics or study of the formal properties of the system are provided. 
With a similar intent, \citet{sin11} introduced a propositional logic extended by a modality for trust. The system focuses on formalising the conceptual relations connecting trust, cooperative actions and, more generally, the interactions between different agents. Trust is not defined here by using other operators, though, but introduced as an elementary notion formalised by a primitive operator. A linear time semantics is also defined.
The work by \citet{pt12}, on the other hand, focuses on communication and its epistemic effects on agents. The system is a type theory in which it is possible to formalise trust as a property of communications. The logical system on which the calculus is based could be described as a propositional second-order one---see \cite[Chapter 11]{glt89} for a standard example of a system of this kind---and, indeed, the calculus features polymorphic terms, whose types are commonly interpreted as second-order logical formulae. \citet{pr14} introduced a propositional language in which atomic formulae represent resources that can be exchanged---or pieces of information that can be communicated---between agents. The language also features operators to describe exchanges of this kind. Trust, also in this system, is defined as a property characterising a communication: in this language, an agent can be said to trust an incoming resource. On the basis of a similar definition of trust, the semantics of several, related negative notions---such as mistrust and distrust---are formally introduced by \citet{pk16}. Similar negative notions of trust are analysed by extensions and adaptations of the formal machinery employed by \citet{pr14} to several negative notions of trust \citep{prbt17,pri20}---we also remark that the latter article contains an extensive review of the literature on trust and its formalisations. A language similar to the one used by \citet{pr14} has also been employed to introduce a type system for reasoning about the trustworthiness of nondeterministic computational processes \citep{dp21,dgp22}. A frame semantics for the systems by \citet{dp21} and \citet{dgp22} has also been introduced by \citet{kp24}.

The system introduced in the present work differs quite strongly from all these systems with respect to the formalism employed. Indeed, we use the a first-order system combining an epistemic logic and a justification logic and we define trust in terms of the modal operators of these two fragments. Moreover, also the way in which we will define trust diverges with respect to these works since, as shown below,  trust will be defined as a relation between an agent and another agent relatively to a property of the second agent.

Finally, a non-monotonic semantics has been introduced by \citet{kra15} in order to capture the formal behaviour of the trust relation. While this semantical framework is very expressive, it only provides a semantical analysis of the issue and requires a very rich language and a very complex semantical machinery, especially if compared to the system that we introduce in the present work.

%A non-monotonic semantics for trust has been introduced by \citet{kra15}. While this semantics constitute a very expressive framework, it only provides a semantical analysis of the issue and it is very complex, especially if compared to the system that we introduce in the present work

\subsection{Formalising computational trust}
\label{sec:computational-trust}

%<<<<<<< HEAD
%Let us finally demonstrate the expressive power of the presented system in relation to the problem of the formal characterisation of computational trust. For the sake of clarity, we repeat that by {\it computational trust} we mean here a relation that holds between an agent and a possibly nondeterministic computer program---or, more in general, computational process---when the agent trusts the program or process to achieve a certain task. From a technical perspective, we are going to define an operator $\trust$ that formalises the relation of computational trust: a formula of the form $\trust_a^t \,A$, with $t$ that occurs in $A$, will then express that agent $a$ trusts program $t$ with respect to the task specified by the formula $A$. For instance, suppose that the formula $P(t,u)$ expresses that $t$ correctly classifies all pictures occurring in list $u$ with respect to what is depicted in them. We can express by the formula $\trust _a^t \,P(t,u)$ that agent $a$ trusts $t$ to correctly classify all pictures occurring in list $u$ with respect to what is depicted in them.\footnote{Formally, we treat $\trust_a^t $ in this notation as a sentential operator, just as we do for $\know _a $ and $j:\;$.}
%=======
% {\ehhi smoothen the introduction}
Let us finally demonstrate the expressive power of the presented system
in relation to the problem of the formal characterisation of computational
trust. For the sake of clarity, we repeat that by computational trust we mean
here a relation that holds between an agent and a possibly nondeterministic computer program---or, more in general, computational process---when
the agent trusts the program or process to achieve a certain task. From a
technical perspective, we are going to define an operator $\trust$ that formalises the relation of computational trust: a formula of the form $ \trust^t_a  A$, with $ t$ that
occurs in $A$, will then express that agent $a$ trusts program $t$ with respect to
the task specified by the formula $A$.
% For instance, suppose that the formula
% P (t, u) expresses that t correctly classifies all pictures occurring in list u
% with respect to what is depicted in them. We can express by the formula
% Tt P (t, u) that agent a trusts t to correctly classify all pictures occurring in
% a
% list u with respect to what is depicted in them.11
% Let us now consider, from a technical point of view, the problem of defining an operator that formalises the relation of computational trust. First of all, by {\it computational trust} we mean here a relation holding between an agent and a possibly nondeterministic computer program---or, more in general, computational process---when the agent trusts the program or process to achieve a certain task. Suppose, for instance,
For instance, suppose that the formula $P(t,u)$ expresses that $t$ correctly classifies all pictures occurring in list $u$ with respect to what is depicted in them. We can express by the formula $\trust _a^t P(t,u)$ that agent $a$ trusts $t$ to correctly classify all pictures occurring in list $u$ with respect to what is depicted in them.\footnote{Formally, we treat $\trust_a^t $ in this notation as a sentential operator, just as we do for $\know _a $ and $j:\;$.}

Before introducing the definition of this operator, let us present an intuitive example that will be formalised, part by part, throughout the rest of the article.
\begin{example}\label{ex:ad}
Alice is very active politically and she often organises extreme left-wing gatherings on polarising topics. She also has a passion for victorian age love novels and organises tea parties for enthusiast readers with a similar inclination. Alice advertises the events she organises through social networks, but, rather obviously, she would like the advertising software to suitably target people for the two very different kinds of events. She would never want some of the conservative old ladies attending her tea parties to find out about her extremist political views. Alice has used for quite some time the advertising service $s$ of a social network with excellent results. Hence, she trusts it very much. She does not know, though, that also the advertising service $t$ of another social network on which she has a profile is based on the same software as $s$. As a consequence, one cannot say that she trusts $t$. When Alice finds out that a third social network, on which she has a lot of contacts, has started providing an advertising service $u$, she starts studying the new service experimenting with it and comparing it with her beloved $s$. After enough tests with satisfactory results, she starts trusting $u$ to be the same program as $s$. She is very happy. Indeed, it would be great for her to advertise her events on the third social network. Moreover, since $s$ and $u$ are essentially the same service, she trusts $u$ to be effective just as she trusts $s$ to be so.
%
%
%Formally, we model the situation by the hypotheses $\trust ^t_a C(t) $ where $C(x)$ means that $t$ is a good classifier, and  $\neg \know _a s\lid t $.  From these hypotheses, we cannot infer that $\trust ^t_a C(t) $. But since we also have the hypothesis $\know _a s\lid u $, we can infer that $\trust ^t_a C(t) $.
\end{example}

The operator $\trust$ can be simply introduced in the language of our logic by extending the last line of the formal grammar presented in Section \ref{sec:language} as follows:
\begin{align*}
\\\text{Formulae }A,B,C\dots   \;\; ::=  \;\; & \bot  \;\;\mid \;\;t\lid s\;\;\mid\;\; P(t_1,\dots ,t_n)\;\;\mid \;\;A\rightarrow B  \;\;\mid\;\; A\wedge B\\
&  \forall x. A  \;\;\mid\;\; \know _a A   \;\;\mid\;\;  j:A    \;\;\mid\;\;   \trust _a^u C
\end{align*}where $P$ is an $n$-ary predicate symbol; $a$ is an agent name; $t,s,u, t_1 , \dots , t_n$ are terms; and $u$ must occur in $C$.

The rules defining the $\trust _a^u $ operator can then be defined as shown in Figure \ref{fig:trust-rules}.
\begin{figure}[h]\centering
\[ \infer{\Gamma \Rightarrow \trust ^t _a A  }{\Gamma \Rightarrow \know _a j:A} 
\qquad\qquad \infer{\Delta \Rightarrow  \know _a x:B}{\Delta \Rightarrow  \trust ^t _a B}  
%\infer{\Gamma ,\Delta  \Rightarrow \trust ^s _a B[s/x]  }{\Gamma    \Rightarrow \trust^s _a t\lid s   & \Delta \Rightarrow \trust ^t _a B[t/x]} 
\]where $t$ must occur in $A$ and $x$ does not occur in $B$ or in  $\Delta$
%and $x$ must be free in $B$
\caption{Definition of the trust operator}
\label{fig:trust-rules}
\end{figure}

Notice that $t$ must occur in $A$ and, on the other hand, $j$ does not necessarily occur in $\trust ^t_a A$. This means that we analyse trust of agent $a$ towards program $t$ with respect to the property $A$ of $t$ as an attitude of $a$ towards $t$ based on $a$'s knowledge that a justification of $A $ exists. The justification $j$ is indeed abstracted out, as it were, when we introduce the trust operator. As a consequence, if we want to eliminate $\trust$, we can simply infer from it the formula $\know _a x:B$ in which $x$ is a fresh variable. This indicates that a justification exists but we do not know its form. 

We, moreover, extend the necessitation rule for $\know $ as show in Figure \ref{fig:new-nec}.
\begin{figure}[h]\centering
\[ \infer{ \know _a A _1 , \dots , \know _a A_n , \trust ^{t_1} _a B _1 , \dots , \trust ^{t_n} _a B_n \Rightarrow \know _a A  }{\know _a A _1 , \dots , \know _a A_n , \trust ^{t_1} _a B _1 , \dots , \trust ^{t_n} _a B_n \Rightarrow A} \]
\caption{New version of the necessitation rule for $\know$}
\label{fig:new-nec}
\end{figure}

The reason behind this new version of the rule is simple: $\trust ^t_a B$ is an abbreviation of $\know_a j:B$ and hence it does not create problems if used inside a context that should only contain formulae of the form $\know_a A $.\footnote{The $\trust_a$ operator can be simply seen as an abbreviation of $\know_a j: $ both from a syntactic and from a semantic point of view. This is the reason why we do not introduce it directly in the language of the logic. Nevertheless, for the sake of clarity and precision, and considering the non-obvious notational handling of $j, x$ and $t$---see Figure \ref{fig:trust-rules}---we explicitly present rules defining it.}

Since, as shown in Section \ref{sec:hyper}, our justification logic is hyperintensional and our epistemic logic is intensional with respect to identity---that is, agents might not know that two objects are identical even though they are at the considered state---also our formal analysis of trust regards it as a hyperintensional notion. That is, even though two terms provably denote the same object---that is, they are provably identical---an agent might trust one and not the other. This intuitively reflects the fact that computational trust, as argued in Section \ref{sec:intro} and especially in the case of opaque and nondeterministic computation, intimately depends on the way an object is presented.

\begin{example}\label{ex:ad2}
We formally model the situation presented in Example \ref{ex:ad}. Let $C$ be a designated predicate symbol such that $C(x)$ means that $x$ is a good classifier of targets for advertising campaigns. Then, our formalisation comprises the following hypotheses:
\begin{itemize}
\item $ \trust ^s_a C(s) $ meaning that Alice---denoted by the agent name $a$---trusts $s$ to be a good classifier of suitable targets for advertising campaigns,
\item $\neg \trust ^t _a s\lid t $ meaning that Alice does not trust $t$ to be the same program as $s$, and 
\item $\trust _a ^u s\lid u $ meaning that Alice trusts $u$ to be the same program as $s$.
\end{itemize}
We first show that we can derive from these hypotheses that $a$ trusts $u$ to be a good classifier:
\[
\infer{ \trust _a ^u s\lid u , \trust ^s_a C(s) \Rightarrow \trust ^u_a C(u)}{\infer{\trust _a ^u s\lid u , \trust ^s_a C(s) \Rightarrow \know _a x: C(u)}{ 
\infer{\trust _a ^u s\lid u , \trust ^s_a C(s) \Rightarrow x: C(u)}{
\infer{\trust _a ^u s\lid u  \Rightarrow y: s\lid u }{\infer{\trust _a ^u s\lid u \Rightarrow \know _a y: s\lid u}{\infer{\trust _a ^u s\lid u \Rightarrow \trust _a  s\lid u}{}}}
&
\infer{ \trust ^s_a C(s) \Rightarrow x: C(s)}{\infer{\trust ^s_a C(s) \Rightarrow \know _a x: C(s)}{\infer{\trust ^s_a C(s) \Rightarrow \trust ^s _a  C(s)}{}}}
}
}}
\]
Moreover, there is clearly no way of deriving the sequent $\neg \trust ^t _a s\lid t , \trust ^u _a s\lid u , \trust ^s_a C(s) \Rightarrow \trust ^t_a C(t)$. Indeed, semantically, we could very well have:
\begin{itemize}
\item a state $w$ in a model such that $C(s)$ holds and there exists a justification for it at all states reachable from $w$ by $\romr_a$, but also 
\item a state $v$ such that
\begin{itemize}
\item $w\romr_a v$,
\item $s\lid t$ does not holds at $v$, 
\item there exists no justification for $C(t)$ at $v$  or for any formula $C(r)$ such that $t\lid r$ holds.
\end{itemize}
\end{itemize}
The absence of a justification for $C(t)$---and for $C(r)$ where $t\lid r $---at a successor $v$ of $w$ would make the formula $\trust ^t_aC(t)$ false at $w$.
\end{example}

In order to fully develop an applicable theory of computational trust, we also need to define a suitable language concerning computation. Let us then now introduce some useful predicates to express properties of programs and, in particular, nondeterministic ones. The intended use of these predicates is to formalise theories of computation that represent the knowledge of one or more agents about programs. After showing how to encode such a theory, we will also show how to employ it in the presented proof-theoretical and semantical systems in order to represent a set of truths about programs and computation of which an agent has knowldge. Before proceeding, though, let us introduce some useful notation.
\begin{definition}
A term context $\calt$ is a term $t$ containing any number of free occurrences of a designated variable $x$. For any term $s$, we denote by  $\calt[s]$ the result of uniformly replacing all occurrences of $x$ in $\calt$ by $s$. Let us denote the set of all term contexts by $\mathrm{Term~contexts}$.\end{definition}

Let us now briefly explain what kind of relations we would like to encode in our logic.
The most basic computational relation that we have in $\lambda$-calculus is the one-step {\it evaluation relation} $\mapsto$---also called one-step {\it reduction relation}---which is defined by the following $\beta$-reduction rule:\[\calt [(\lambda x.t )s]\mapsto \calt[ t^{s/x}] \]where $\calt $ is a term context and $t^{s/x}$ is the term obtained by uniformly replacing all free occurrences of $x$ in $t$ by $s$.\footnote{This term is the usual result of a $\beta$-reduction in $\lambda$-calculus  \cite{chu32,chu36,bar84}. We do not denote it by $t[s/x]$, as usually done, only because the definition of $t[s/x]$ employed in this work---which is meant to formalise substitutions inside formulae but not inside terms---is not the one suited for $\lambda$-calculus $\beta$-reductions.} This rule essentially determines how a $\lambda$-application triggers a reduction of the term containing it. Intuitively, in the term to the left of $\mapsto$ we have that the function $(\lambda x.t )$ is applied to its argument $s$ and, in the term to the right of $\mapsto$, we have the result of replacing all free occurrences of $x$ by $s$ inside the body $t$ of the function. This is how functions work in $\lambda$-calculus, nothing more is required as far as deterministic computation is concerned.

On the basis of this, we can also define the reducibility relation $\mapsto^*$ as the reflexive and transitive closure of $\mapsto$. The  relation $t\mapsto ^*t'$ indicates that $t$ can reduce to $t'$ in a number of immediate reduction steps, possibly $0$. In other words, $\mapsto^* $ relates any term $t$ to its output (if one exists\footnote{What formalises outputs  in $\lambda$-calculus are those terms that do not reduce to any other term. These are often called {\it values}. Some $\lambda$-calculus terms, though, are non-terminating; which means that they can never reduce to a value.}) and to all the intermediate results that one might go through while reducing $t$.

%Now, the simplest way of encoding a reduction relation among programs is by employing the identity predicate $\lid$ of our logic. 

Now, we can define these relations in our system by specifying their behaviour in a possibly infinite set of formulae $\Sigma$. In order to know what follows from some assumptions $\Gamma$ and these definitions, it is enough to check, for any formula of interest $A$, whether $\Sigma\cup \Gamma\vdash A$. 

We now formally show how the evaluation relations can be defined in our language, but, first of all, let us define term substitution. 
\begin{definition}[Term substitution]For any terms $s, t$ and $u$,
\begin{itemize}
\item $x^{s/x}=s$
\item $y^{s/x}=y$ when $y\neq x$
\item $(tu)^{s/x}=t^{s/x}u^{s/x}$
\item $(\lambda x. t)^{s/x}=\lambda x. t$
\item $(\lambda y. t)^{s/x}=\lambda y. t^{s/x}$ when $y\neq x$
\item $(!t)^{s/x} = !t^{s/x}$
\end{itemize}As usual, we assume that capture of variables is avoided by renaming bound variables.
\end{definition}

\begin{definition}[Syntactic definitions $\Sigma_{\mapsto}$ and $\Sigma_{\mapsto^*}$ of the deterministic evaluation relations]
Given two designated binary predicate symbols $\mapsto$ and $\mapsto^*$, let $\Sigma _{\mapsto}$ be the set of formulae defined as follows:
\[\Sigma _{\mapsto}=\{\calt [(\lambda x.t )s]\mapsto \calt[ t^{s/x}] \;\mid\; t,s\in \mathrm{Terms}, \;  \calt\in \mathrm{Term~contexts}  \}\]

Let $\Sigma _{\mapsto^*}$ be the smallest set $\Sigma $ verifying the following conditions: 
\begin{itemize}
\item If $t\in \mathrm{Terms}$ then $t\mapsto^* t\in \Sigma $. 

\item If $t\mapsto t' \in \Sigma _{\mapsto}$ then $t\mapsto^* t'\in \Sigma $. 

\item If $t\mapsto ^*t', t'\mapsto^* t''\in \Sigma $ then $t\mapsto^* t''\in \Sigma $.
\end{itemize}

\end{definition}

Let us now define the notion of program identity induced by the relation $\mapsto^*$, which is what matters most to us with respect to  computational trust. As usual, we define two programs as identical if one is reducible to the other. The intuition behind this identification is that one of the two programs can be considered as an intermediate result which can be obtained during a run of the other. Hence, they can be seen as the same program at different stages of the computation. 

\begin{definition}[Syntactic definitions $\Sigma _{\lambda =}$ for program identity]
Let $\Sigma _{\lambda =}$ be the set of formulae defined as follows:
\[\Sigma _{\lambda =}=\{t\lid s \;\mid\; t,s\in \mathrm{Terms}, \;  t\mapsto s, s\mapsto t\in \Sigma _{\mapsto}  \}\]
\end{definition}

%We specify this by a formula of the following form:\[\calt [(\lambda x.t)s]\lid  \calt [t^{s/x}]\]for any term context $\calt$ and terms $t,s$.
%In order to account for renaming of bound variables, which does not influence in any way the behaviour of a term, we can also introduce a formula of the following form for any term $t$:\[t\lid t'\]where $t'$ is obtained by replacing all occurrences of any bound variable of $t$ by another variable, without the new one being captured by another binder, obviously. 
Notice that, since $\lid$ is closed under reflexivity, symmetry and transitivity both in the calculus and in the semantics, we obtain exactly what we had aimed for---that is, an equational theory of computation---even without reference to the reflexive and transitive closure $\mapsto^*$ of the relation $\mapsto$.\footnote{We do not go into details here, but we remark that it is easy to also introduce formulae guaranteeing that the equational theory is also closed with respect to $\alpha$-equivalence and $\eta$-expansion, if required.} To be precise, though, what we have now is an equational theory of deterministic computation. Indeed, the only operator that we have here is $\lambda$, which, considering that pure $\lambda$-calculus is confluent and sequential \citep[Chapters 11.1 and 14.4]{bar84}, encodes a simple, deterministic notion of computation.

In order to be able to express judgements about non-deterministic computation, we need to work out something more. The standard way of extending the evaluation relation  to the non-deterministic case is by generalising the one-step evaluation relation $\mapsto $ to a probabilistic version $\mapsto_p$ in which $p$ is a number in the unit interval $[0,1]$. By $t\mapsto _p t'$ we mean that the term $t$ has a probability of $p$ of reducing to $t'$ in one evaluation step. Usually, this is required when $\lambda$-calculus is extended by particular operators that  encode non-deterministic computational  mechanisms---see \cite{rp02,par03,dhw05,ppt05,dlz12,sco14,dk20,akm21} for some specific examples taken from the literature on theoretical computer science and \cite{dp95} for a conceptual characterisation of non-deterministic computation. For instance, according to a rather common notation, the term $(\frac{1}{3}:x\oplus \frac{2}{3}:y)$ encodes a program that reduces to $x$ with probability $\frac{1}{3}$ and to $y$ with probability $\frac{2}{3}$. The behaviour of such a program, thus, is fully determined by the following relation instances:
\[(\frac{1}{3}:x\oplus \frac{2}{3}:y) \mapsto _{\frac{1}{3}} x \qquad (\frac{1}{3}:x\oplus \frac{2}{3}:y)\mapsto _{\frac{2}{3}} y\]
%Clearly, an obvious way to generalise this to a derivability relation $\mapsto^*_p$ exists.
Nevertheless, since we are mainly interested in opaque non-deterministic computation, we will not actually extend our term language by operators for non-deterministic computation but we will simply exploit the term language of pure $\lambda$-calculus in order to admit the possibility that some variables encode calls to black-box non-deterministic programs.
A variable program $x$ might then be  introduced by the following computational theory specifying its behaviour: 
\[x \mapsto _{\frac{1}{3}} t \qquad x\mapsto _{\frac{2}{3}} s\]
The variable $x$ will then represent a non-deterministic black-box program that reduces to the term $t$ with  probability $\frac{1}{3}$ and to the term $s$ with probability $\frac{2}{3}$. We can now study the behaviour of the program $x$---and, more interestingly, what the agents in our system know and think about the behaviour of $x$---without explicitly formalising how it is constructed; just by describing it from the outside, one might say.

For the sake of simplicity, we can skip the definition of $\mapsto _p$ and directly introduce a non-determistic identity relation between programs. Since non-deterministic computation is involved, this is not a trivial task, but neither an impossible one. Instead of trying to define an identity relation indexed by a probability value---a solution which appears rather natural, but which is not so easy to work out in detail---we adopt the solution of keeping together  all possible outcomes of the reduction of one term and specify an identity relation between sets of terms. Intuitively, a non-deterministic term is not equivalent, in general, to its reducti. Indeed, a non-deterministic reduction can constitute a proper computational choice that cannot be undone during a computation: while in pure $\lambda$-calculus the final result of the computation is the same even though we can make different choices along the computation itself---this is technically referred to as {\it confluence}, see \citep[Chapter 11.1]{bar84} for the technical details---in non-deterministic $\lambda$-calculi, we might have choices along a computation that can lead to properly different final results. In order to define a notion of equivalence in the non-deterministic setting, therefore, we need to keep together all possible reducti of a term. We do this by talking about equivalent sets of terms. In order to suitably handle non-deterministic reduction and not to loose essential information, we also specify, next to the element of the set, the probability with which it would actually be the outcome of the reduction. For instance, the black-box non-deterministic program with behaviour specified by\[x \mapsto _{\frac{1}{3}} t \qquad x\mapsto _{\frac{2}{3}} s\]will be introduced as follows in an equational theory of non-deterministic computation of the kind we will adopt:\[(( 1,x)) \;\;\lid\;\; (( \frac{1}{3}, t),( \frac{2}{3}, s))\]
where ordered pairs and lists constructed by $( \; )$ are encoded in pure $\lambda$-calculus in the standard way.\footnote{A pair operator $\lambda x. \lambda y.(x,y) $ can be encoded, in pure $\lambda$-calculus, by the function $\lambda x. \lambda y.\lambda z.zxy  $. Indeed, the pair operator works as follows: when applied to any two terms $t$ and $s$ yields the pair $(t,s)$---$(\lambda x. \lambda y.(x,y) )ts\mapsto^* (t,s)$---which can then be either projected to obtain the first element---$(t,s)\pi_1\mapsto^* t$---or projected to obtain the second one---$(t,s)\pi _2\mapsto^* s$. The encoding by $\lambda x. \lambda y.\lambda z.zxy  $ works in exactly the same way if we employ $\lambda z_1 . \lambda z_2. z_1$ as first projection and $\lambda z_1 . \lambda z_2. z_2$ as the second one. Indeed, $(\lambda x. \lambda y.\lambda z.zxy  )ts\mapsto^* \lambda z.zts$ and, moreover, on the one hand we have $(\lambda z.zts)\lambda z_1 . \lambda z_2. z_1\mapsto^* (\lambda z_1 . \lambda z_2. z_1)ts \mapsto^* t$ while, on the other hand, we have  $(\lambda z.zts)(\lambda z_1 . \lambda z_2. z_2)\mapsto^* (\lambda z_1 . \lambda z_2. z_2)ts \mapsto^* s$.} 

The following definition specifies how a set $\Sigma $ containing formulae of the form $t\mapsto _p s$---that is, formulae stating that $t$ reduces to $s$ with probability $p$---can be extended to an equational theory of probabilistic computation encoded by lists as explained above. The first item concerns the direct translation of $t\mapsto _p s$ for all possible reducti $s$ of $t$, the second generalises it to  pairs occurring inside existing lists, the third guarantees that lists are treated as multisets, and the last one enables us to collapse elements of a list that concern the same term. 
\begin{definition}[Syntactic definitions $\Sigma _{\romp\lambda =}$ for probabilistic program identity]
Given a theory $\Sigma $ specifying the behaviour of probabilistic terms  by formulae of the form $t\mapsto _p s$, let $\Sigma _{\romp\lambda =}$ be the smallest set satisfying the following conditions:
\begin{itemize}
\item If $t, s_1, \dots , s_n\in \mathrm{Terms}$ and $t\mapsto_{p_i} s_i \in \Sigma $ for $1\leq i\leq n$, then  $((1, t))\lid ((p_1 , s_1 ), \dots , (p_n , s_n ))\in \Sigma _{\romp\lambda =}$.

\item If $t, s_1, \dots , s_n\in \mathrm{Terms}$, $t\mapsto_{q_i} s_i \in \Sigma $ for $1\leq i\leq n$ and $\dots \lid (\dots ,(p, t), \dots ) \in \Sigma _{\romp\lambda =}$, then  $(\dots ,(p, t), \dots )\lid (\dots ,(p\cdot q_1, s_1), \dots , (p\cdot q_n, s_n), \dots )\in \Sigma _{\romp\lambda =}$.

\item If $\dots \lid (\dots ,(p, t), (q,s), \dots ) \in \Sigma _{\romp\lambda =}$, then  $(\dots ,(p, t), (q,s), \dots )\lid(\dots , (q,s), (p, t), \dots )\in \Sigma _{\romp\lambda =}$.

\item If $\dots \lid (\dots ,(p, t), (q,t), \dots ) \in \Sigma _{\romp\lambda =}$, then  $(\dots ,(p, t), (q,t), \dots )\lid(\dots , (p+q,t), \dots )\in \Sigma _{\romp\lambda =}$.

%%\item If $t\in \mathrm{Terms}$, then $t\lid t \in  \Sigma _{\lambda =}$.
%%
%%\item If $t, s\in \mathrm{Terms}$ and $t\lid s \in \Sigma _{\romp\lambda =}$, then  $s\lid t \in \Sigma _{\romp\lambda =}$.
%
%\item If $t, s, u\in \mathrm{Terms}$ and $t\lid s , s\lid u \in \Sigma _{\romp\lambda =}$, then  $t\lid u \in \Sigma _{\romp\lambda =}$.

\end{itemize}
\end{definition}

%
%
%
%In order to use these lists to suitably represent multisets of terms associated to their probability value, we will also include formulae of the following form in all our theories:
%\[ ( \dots , ( p,  t),( q,  s), \dots ) \lid ( \dots , ( q,  s), ( p,  t), \dots )  \]\[ ( \dots , ( p,  t) , ( p,  t), \dots ) \lid ( \dots , ( p+p,  t) , \dots ) \]for any list $( \dots , ( p,  t), ( q,  s), \dots )$ and any list $( \dots , ( p,  t) , ( p,  t), \dots ) $. 
%

Obviously, once adopted the generalised notation for non-deterministic computation, we can use it also to talk about deterministic computation. Indeed, we can express by the new notation something like $\calt [(\lambda x.t)s]\lid  \calt [t^{s/x}]$ by\[((1, \calt [(\lambda x.t)s]))\;\;\lid\;\;  ((1,\calt [t^{s/x}]))\]Notice moreover that, if we do this, we can simply stipulate that our identity on terms is always used on lists of pairs of the form $(p,t )$ where $p$ is a number in the unit interval and $t$ a program, thus avoiding any possible ambiguity concerning the use of pairs to encode programs and their use to encode multisets of terms associated to probability values.

Since, as already shown in Section \ref{sec:hyper}, our logic enables us to handle in a rather different way what is objectively true but unknown to agents and what is, on the contrary, known to agents or even trusted by agents; we need, first of all, to specify whether an agent $a$ actually knows that the theory is true or not. If not, we simply need to use our theory $S$ as a set of hypotheses---and we will better specify below how to do so in the calculus and in a model. In doing so, we just assume that $S$ is a true theory but no agent necessarily knows of it. Some agent might know part of it---an agent might actually know it all---but for contingent reasons depending on the situation at hand. If, on the contrary, we wish to force the fact that some agent knows the theory and make this fact robust with respect to the change of situation, then we will have to prefix all formulae of the theory by $\know _a$ and thus consider the set $\know _a \Sigma=\{\know _a A\mid A\in \Sigma \} $ as part of our hypotheses. This will also mean, by factivity of $\know$, that $\Sigma $ is objectively true. Indeed, $\know _a \Sigma \vdash A$ for any $A\in \Sigma $. Since $\trust _a $ is defined in terms of $\know _a $ and justifications, all that has been said for $\know _a $ also holds for $\trust _a$, with the only caveat that, in defining $\trust _a \Sigma$, we must specify a term for each outermost occurrence of the $\trust _a$ operator in the set $\trust _a \Sigma$. We can simply state that $\trust _a \Sigma=\{\trust ^{t_A} _a A\mid A\in \Sigma\} $ where $t_A$ is any term occurring in $A$, or we can be more specific depending on particular requirements of the considered case.

Let us then explain in detail how to 
use the formulae of our theory as additional hypotheses depending on the formalism we are considering. We will focus on $\know _a \Sigma$ since the same considerations also apply to $\Sigma$.

From a proof-theoretical perspective---that is, if we are working with the calculus introduced in Section \ref{sec:calculus}---we can simply add $\know_a \Sigma $ to the hypotheses of our derivations. This means that, if we are interested in the status of a formula $B$ or of a judgement of the form  $\Gamma\vdash B$ under the hypothesis that agent $a$ has knowledge of all that is stated in $\Sigma $, instead of directly checking whether the judgements $\vdash A$ and $\Gamma\vdash B$ hold, we should check whether $\know_a \Sigma  \vdash A$ and $\know_a \Sigma  \cup\Gamma\vdash B$ hold.
Indeed, for instance, if $\know _a \Sigma \vdash A$ holds, then we established the truth of the formula $A$ under the hypothesis that agent $a$ knows the theory $\Sigma $ and can use it to reason about programs. Something analogous can be stated for $\know_a \Sigma\cup \Gamma\vdash B$. Obviously, this can be done cumulatively for different theories and different agents, as in $\know _{a_1} \Sigma _1 \cup \dots \cup \know _{a_n} \Sigma _n\vdash A$ and $\know _{a_1} \Sigma _1\cup \dots \cup \know _{a_n} \Sigma _n \cup \Gamma \vdash B$.

From a model-theoretical perspective, on the other hand,  $\know _a \Sigma  $ can be added among our hypotheses by simply considering statements of the form $\know_a \Sigma  \vDash A$ and $\know_a \Sigma \cup\Gamma\vDash B$ instead of $ \vDash A$ and $ \Gamma\vDash B$. Indeed, for instance, if $\know _a \Sigma \vDash A$ holds, then we established the truth of the formula $A$ at all states at which agent $a$ knows that $\Sigma $ is true; and, again, something analogous can be stated for $\know_a \Sigma  \cup \Gamma\vDash B$. Finally, rather obviously, in this case as well the method can be used also for multiple theories and different agents.

Notice, in general, that Corollary \ref{cor:soundness} and Theorem \ref{thm:completeness} guarantee that both the calculus and the semantics will behave equivalently with respect to any additional set of hypotheses $\Sigma $, $\know_a \Sigma $ or  $\know _{a_1} \Sigma _1, \dots , \know _{a_n} \Sigma _n$. 

Notice moreover that, if the elements of $\Sigma $ are neither of the form $\know _a T$ nor of the form $j:T$, then adding the elements of $\Sigma $ as additional hypotheses in a derivability statement of the form $\vdash A$ or $\Gamma \vdash B$ will correspond to considering semantical states at which all elements of $\Sigma $ are true. That is to say, the elements of $\Sigma $ will not be necessarily true at all states of the model we are considering, but just at the state that we take to be the actual state and at which we evaluate formulae. To the contrary, adding the elements of $\know _a \Sigma $ as additional hypotheses in a derivability statement corresponds to forcing the truth of all elements of $ \Sigma $ at all states  reachable by $\romr_a$ from the actual state---thus including the actual state itself since $\romr_a$ is always reflexive.

\begin{example}\label{ex:ad3}
The definitions that we just introduced in order to describe properties of nondeterministic computation can be used, for instance, to model the way in which Alice, in Example \ref{ex:ad}, reaches the conclusion that service $s$ and service $u$ are equivalent.

Suppose that $s$ and $u$ are described by the following statements:
\[s\mapsto _\frac{1}{4}o_1 \quad s\mapsto _\frac{3}{4}o_2  \]\[u\mapsto _\frac{3}{4}u_1 \quad u\mapsto _\frac{1}{4}u_2\]\[u_1 \mapsto _\frac{1}{3}o_1\quad u_1 \mapsto _\frac{2}{3}o_2\]\[u_2\mapsto _1 o_2\]where $o_1, o_2$ are possible outputs of $s$ and $u$ and $u_1,u_2$ are intermediate results of program $u$.

We can build a set $\Sigma _{\romp\lambda =}$ by taking
\[\Sigma = \{s\mapsto _\frac{1}{4}o_1 , \; s\mapsto _\frac{3}{4}o_2  , \;  u\mapsto _\frac{3}{4}u_1  , \;  u\mapsto _\frac{1}{4}u_2  , \;  u_1 \mapsto _\frac{1}{3}o_1 , \;  u_1 \mapsto _\frac{2}{3}o_2 , \; u_2\mapsto _1 o_2\}\]as set defining the behaviour of $s$ and $u$.\footnote{Notice that the presented system does not provide a native method for extracting data from experiments on the behaviour of a program, but a system for reasoning about probabilistic computational processes, such as the one presented by by \citet{dp21,dgp22}, can be used to obtain the data describing $s$ and $u$ that Alice uses to reach the conclusion that they are equivalent.}

In  $\Sigma _{\romp\lambda =}$ we will have an element of the form $((1, s))\lid ((\frac{1}{4}, o_1), (\frac{3}{4}, o_2))$, quite obviously, but also a chain of elements of the form $((1, u))\lid r_1, \; r_1 \lid r_2 ,\;  \dots ,\; r_{n-1} \lid r_n ,\; r_n  \lid  ((\frac{1}{4}, o_1), (\frac{3}{4}, o_2))$. The latter is less obvious, but not too complicated to verify.
Therefore, if we assume that Alice---that is, agent $a$---trusts $s$ and $u$ to act according to the theory $\Sigma _{\romp\lambda =}$ by employing hypotheses in the set $\trust  _a \Sigma _{\romp\lambda =} $, we can formally derive that agent $a$ trusts $u$ to be the same program as $s$. Formally, we can derive the sequent 
\[\trust ^{((1, s))} _a ((1, s))\lid ((\frac{1}{4}, o_1), (\frac{3}{4}, o_2)) ,\;\; \trust ^{((1, u))} _a ((1, u))\lid r_1,\;\; \trust ^{r_1} _a r_1 \lid r_2 , \;\;\dots  \]
\[\dots ,\;\; \trust ^{r_{n-1}}  _a r_{n-1} \lid r_n , \;\; \trust  ^{r_n}_a r_n  \lid  ((\frac{1}{4}, o_1), (\frac{3}{4}, o_2))\;\;   \Rightarrow\;\; \trust ^u _a s\lid u\]by a derivation with the same structure as the one shown in Example \ref{ex:ad2}---that is, a derivaiton that starts with $\trust$ eliminations followed by $\know$ eliminations, that continues with applications of the rules for substituting identicals inside justified formulae, and concludes by reintroducing $\know$ and then $\trust$ in order to obtain $\trust ^u _a s\lid u$ as desired.
\end{example}

Obviously, all that has been said about the theory of  computational trust can also be applied, more generally, to less specific notions of trust, as, for instance, trust among humans. What holds true for computational theories, moreover, also holds for theories in general, that can be simply seen as sets of truths that some agent knows.

The particular attention that we devote here to computational trust is due to the fact that the main motivation behind the definition of the presented system is to formalise an analysis of subjective trust and non-deterministic computation based on analyses of computational trustworthiness such as the one conducted by \citet{dp21,dgp22}. The aim of the present section is simply to argue that the language and formal mechanisms provided by the presented logical system is enough to endeavour in such an analysis. We leave a full development of such a theory of computational trust, a formal  study of its metalogical properties and expressive power, and a discussion of its conceptual and practical applications  for future work.

\backmatter

%\bmhead{Supplementary information}
%
%If your article has accompanying supplementary file/s please state so here. 
%
%Authors reporting data from electrophoretic gels and blots should supply the full unprocessed scans for key as part of their Supplementary information. This may be requested by the editorial team/s if it is missing.
%
%Please refer to Journal-level guidance for any specific requirements.

\bmhead{Acknowledgements}

We would like, first of all, to thank the anonymous reviewers for the very useful comments that greatly helped to improve the present work. We would like, moreover, to thank Ekaterina Kubyshkina, Mattia Petrolo and Giuseppe Primiero for several very useful discussions on hyperintensionality and computational trust. 

%We thank, moreover, all reviewers for their comments and suggestions. 
%
This work was supported by the project BRIO (BIAS, RISK, OPACITY in AI: design, verification and development of Trustworthy AI -- PRIN project n. 2020SSKZ7R).

%\section*{Declarations}
%
%Some journals require declarations to be submitted in a standardised format. Please check the Instructions for Authors of the journal to which you are submitting to see if you need to complete this section. If yes, your manuscript must contain the following sections under the heading `Declarations':
%
%\begin{itemize}
%\item Conflict of interest/Competing interests (check journal-specific guidelines for which heading to use)
%\item Ethics approval and consent to participate
%\item Consent for publication
%\item Data availability 
%\end{itemize}

%\noindent
%If any of the sections are not relevant to your manuscript, please include the heading and write `Not applicable' for that section. 

%\bigskip
%\bigskip\noindent
%Springer journals and proceedings: \url{https://www.springer.com/gp/editorial-policies}
%
%\bigskip\noindent
%Nature Portfolio journals: \url{https://www.nature.com/nature-research/editorial-policies}
%
%\bigskip\noindent
%\textit{Scientific Reports}: \url{https://www.nature.com/srep/journal-policies/editorial-policies}
%
%\bigskip\noindent
%BMC journals: \url{https://www.biomedcentral.com/getpublished/editorial-policies}
%\end{flushleft}

\begin{appendices}

\section{ }\label{appendix}

\noindent {\bf Theorem \ref{thm:der-preservation}.}
{\it If  $\Gamma \Rightarrow A $ is derivable, then all states that verify $\Gamma$ also verify $A$.}
\begin{proof}
The proof is by induction on the number of rules applied in the derivation of $\Gamma \Rightarrow A$.

Since no derivation can contain $0$ rule applications, in the base case, the derivation contains one rule application. 

The rule is one of the following:\[\infer{A\Rightarrow A}{}\qquad\qquad\infer{ \Rightarrow t\lid t}{}\]
In case of the first one, the statement is trivial. Let us then consider the second one. According to Definition \ref{def:forcing}, $t\lid t$ is true at a state if, and only if, the interpretation of $\lid $ contains $(t,t)$. But, according to Definition \ref{def:model}, all interpretation functions $\romi _w$ for all states $w$ in all models $\calm$ meet this requirement.

Suppose now that the statement holds for any derivation containing less than $n$ rule applications. We show that it holds also for any derivation that contains $n$ rule applications.  

Let us reason by cases on the last rule applied in the derivation.
\begin{itemize}
%\item $\vcenter{\infer{t\lid t}{}}$

\item $\vcenter{\infer{\know _a A_1 , \dots \know _a A_n\Rightarrow \know _a A}{\know _a A_1 , \dots \know _a A_n\Rightarrow A}}\quad$ By inductive hypothesis, $A$ is true at all states at which each $ \know _a A_1 , \dots \know _a A_n$ is true. Consider a state $w$ of this kind in a model $\calm$. Since $\romr_a$ is Euclidean, all successors of $w$ with respect to $\romr_a$ not only verify $A_1 , \dots A_n$, but also $\know _a A_1 , \dots \know _a A_n$. Indeed, all successors of successors of $w$ are also successors of $w$. Hence, all successors of successors of $w$ verify $A_1 , \dots A_n$, which implies that all successors of $w$ verify $\know _a A_1 , \dots \know _a A_n$. But by our assumptions, this means that all successors of $w$ verify $A$. Which implies that $w$ verifies $\know _a A$, as desired.

%by Definition \ref{def:forcing} for $\know _a$, $\know_a C$ must be true at all states of all models too. Indeed, if the state has no successors with respect to $\romr_a$, all requirements of Definition \ref{def:forcing} for $\know _a$ are vacuously met; if the state has some successors with respect to $\romr_a$, $C$ must be true at all the successors because it is true at all states of all models.

\item $\vcenter{\infer{\Gamma \Rightarrow s\lid t}{\Gamma \Rightarrow t\lid s}}\quad$ By inductive hypothesis, for any state $w$ that verifies $\Gamma$,  $\calm , w, \romi _w\Vdash t\lid s$. According to Definition \ref{def:forcing}, $\calm , w, \romi _w\Vdash t\lid s$ if, and only if, $(t^w,s^w)\in \romi _w(\lid )$. According to Definition \ref{def:forcing}, $\calm , w, \romi _w\Vdash s\lid t$ if, and only if, $(s^w,t^w)\in \romi _w (\lid )$. But according to Definition \ref{def:model}, if $\romi _w(\lid )$ contains $(t^w,s^w)$ then it must contain $(s^w,t^w)$.
Hence, $\calm , w, \romi _w\Vdash s\lid t$, as desired.

\item $\vcenter{\infer{\Gamma ,\Delta \Rightarrow u\lid v}{\Gamma \Rightarrow u\lid t&\Delta \Rightarrow t\lid v}}\quad$ By inductive hypothesis, for any state $w$ that verifies $\Gamma$ we have $\calm , w, \romi _w\Vdash u\lid t$ and for any state $w$ that verifies $\Delta $ we have $\calm , w, \romi _w\Vdash t\lid v$. Let us consider a state $w$ that verifies $\Gamma, \Delta$. According to Definition \ref{def:forcing}, $\calm , w, \romi _w\Vdash u\lid t$ and $\calm , w, \romi _w\Vdash t\lid v$ if, and only if, $(u^w,t^w),(t^w,v^w)\in \romi _w(\lid )$. According to Definition \ref{def:forcing}, $\calm , w, \romi _w\Vdash u\lid v$ if, and only if, $(u^w,v^w)\in \romi _w (\lid )$. But according to Definition \ref{def:model}, if $\romi _w(\lid )$ contains $(u^w,t^w)$ and $(t^w,v^w)$ then it must contain $(u^w,v^w)$. Hence, $\calm , w, \romi _w\Vdash u\lid v$, as desired.

\item $\vcenter{\infer{\Gamma , \Delta\Rightarrow A(s/x)}{\Gamma \Rightarrow t\lid s&\Delta \Rightarrow A(t/x)}}\quad$ By inductive hypothesis, for any state $w$ that verifies $\Gamma$ we have $\calm , w, \romi _w\Vdash t\lid s$ and for any state $w$ that verifies $\Delta $ we have $\calm , w, \romi _w\Vdash A(t/x)$. Let us consider a state $w$ that verifies $\Gamma, \Delta$ and let us show that it verifies $A(s/x)$. Since for $t=s$ the claim trivially holds, we also suppose that $t\neq s$ and reason by induction on the form of $A$ and show that also $\calm , w, \romi _w\Vdash A(s/x)$.
\begin{itemize}
\item  $A=\bot$. Since $\bot (t/x)=\bot (s/x)$, the statement trivially holds.

\item $A=(u\lid v)$. Now, $\calm , w, \romi _w\Vdash t\lid s$, which implies that $(t^w,s^w)\in \romi_w (\lid )$. Moreover, $\calm , w, \romi _w\Vdash (u\lid v)(t/x)$, which implies that  $( (u(t/x))^w,(v(t/x))^w)\in \romi_w (\lid )$. We have four cases with respect to $x$: $u=x=v$, $u\neq x\neq v$, $u= x\neq v$, or $u\neq x= v$. 

If $u=x=v$, then $\calm , w, \romi _w\Vdash (u\lid v)(s/x)$ since $(u\lid v)(s/x)=(s\lid s)$.

If $u\neq x\neq v$, then $\calm , w, \romi _w\Vdash (u\lid v)(s/x)$  since $(u\lid v)(s/x)=(u\lid v)=(u\lid v)(t/x)$. 

If $u= x\neq v$, then $(u\lid v)(s/x)=(s\lid v)$ and $(u\lid v)(t/x)=(t\lid v)$. Moreover, $\calm , w, \romi _w\Vdash t\lid v$ and $\calm , w, \romi _w\Vdash t\lid s$ hold. But this implies, by the transitive and symmetric closure of $\romi ^w(\lid )$, that $\calm , w, \romi _w\Vdash  s\lid v$ with $ (s\lid v)= (u\lid v)(s/x)$. 

If $u\neq x= v$, then $(u\lid v)(s/x)=(u\lid s)$ and $(u\lid v)(t/x)=(u\lid t)$. Moreover, $\calm , w, \romi _w\Vdash u\lid t$ and $\calm , w, \romi _w\Vdash t\lid s$ hold. But this implies, by the transitive closure of $\romi ^w(\lid )$, that $\calm , w, \romi _w\Vdash u\lid s$ with $(u\lid s) =  (u\lid v)(s/x)$.

%---and considering that Definition \ref{def:model} requires that the set of $n$-uples $\romi_w(P)$ is closed under all possible uniform replacements of $u$ by $v$ if $(u^w,v^w)\in \romi_w (\lid )$, we have that  $P(t_1,\dots ,t_n)(t/x)$ and $P(t_1,\dots ,t_n)(s/x)$ must have the same truth value at $w$. hence, also $P(t_1,\dots ,t_n)(s/x)$  is true at $w$.

\item $A= P(t_1,\dots ,t_n)$. Now, $\calm , w, \romi _w\Vdash P(t_1,\dots ,t_n)(t/x)$ and\[P(t_1,\dots ,t_n)(t/x)=P(t_1(t/x),\dots ,t_n(t/x))\] imply that $(t_1(t/x)^w,\dots ,t_n(t/x)^w)\in \romi _w (P)$. Moreover, $\calm , w, \romi _w\Vdash t\lid s$ implies that  $(t^w,s^w)\in \romi_w (\lid )$. Moreover, since $(t^w,s^w)\in \romi_w (\lid )$ holds, Definition \ref{def:model} requires the set of $n$-uples $\romi_w(P)$ to be closed under all possible replacements of $t^w$ by $s^w$. 
Hence, $(t_1(t/x)^w,\dots ,t_n(t/x)^w)\in \romi _w (P)$ implies that $(t_1(s/x)^w,\dots ,t_n(s/x)^w)\in \romi _w (P)$. But since, by Definition \ref{def:forcing}, $(t_1(s/x)^w,\dots ,t_n(s/x)^w)\in \romi _w (P)$ if, and only if, $\calm , w, \romi _w\Vdash P(t_1(s/x),\dots ,t_n(s/x))$ and since \[P(t_1(s/x),\dots ,t_n(s/x))=P(t_1,\dots ,t_n)(s/x)\]we have, as desired, that $\calm , w, \romi _w\Vdash P(t_1,\dots ,t_n)(s/x)$.

%we have that  $P(t_1,\dots ,t_n)(t/x)$ and $P(t_1,\dots ,t_n)(s/x)$ must have the same truth value at $w$. Hence, $\calm , w, \romi _w\Vdash P(t_1,\dots ,t_n)(t/x)$ implies that $\calm , w, \romi _w\Vdash P(t_1,\dots ,t_n)(s/x)$, as desired.

\item $A= B\rightarrow C$. Now, by inductive hypothesis, $\calm , w, \romi _w\Vdash B(t/x)$ if, adn only if, $\calm , w, \romi _w\Vdash B(s/x)$; and $\calm , w, \romi _w\Vdash C(t/x)$ if, and only if, $\calm , w, \romi _w\Vdash C(s/x)$. Therefore, $\calm , w, \romi _w\Vdash B(t/x)\rightarrow C(t/x)$ if, and only if,  $\calm , w, \romi _w\Vdash B(s/x)\rightarrow C(s/x)$. But since $(B\rightarrow C)(t/x)= B(t/x)\rightarrow C(t/x)$ and $(B\rightarrow C)(s/x)= B(s/x)\rightarrow C(s/x)$, we have, as desired, that $\calm , w, \romi _w\Vdash (B\rightarrow C)(t/x)$ implies  $\calm , w, \romi _w\Vdash (B\rightarrow C)(s/x)$.

\item  $A= B\wedge C$. Now, by inductive hypothesis, $\calm , w, \romi _w\Vdash B(t/x)$ if, adn only if, $\calm , w, \romi _w\Vdash B(s/x)$; and $\calm , w, \romi _w\Vdash C(t/x)$ if, and only if, $\calm , w, \romi _w\Vdash C(s/x)$. Therefore, $\calm , w, \romi _w\Vdash B(t/x)\wedge C(t/x)$ if, and only if,  $\calm , w, \romi _w\Vdash B(s/x)\wedge C(s/x)$. But since $(B\wedge C)(t/x)= B(t/x)\wedge C(t/x)$ and $(B\wedge C)(s/x)= B(s/x)\wedge C(s/x)$, we have, as desired, that $\calm , w, \romi _w\Vdash (B\wedge C)(t/x)$ implies  $\calm , w, \romi _w\Vdash (B\wedge C)(s/x)$.

\item  $A= \forall y. B $. If $x=y$, then $(\forall y. B)(t/x)=\forall y. B=(\forall y. B)(s/x)$ and the statement trivially holds. Let us then suppose that $x\neq y$. Since $t/neq s$, have three cases:  $t\neq y\neq s$, $t= y\neq s$ or $t\neq y= s$. 

Let us suppose that $t\neq y\neq s$. By inductive hypothesis, for any term $u$, $\calm , w, \romi _w\Vdash B(t/x)[u/y] $ if, and only if, $\calm , w, \romi _w\Vdash B(s/x)[u/y]$. Hence, by Definitions \ref{def:forcing} and \ref{def:model}  (there exists a term $u$ such that $u^w=e$ for each $e\in \romu$), $\calm , w, \romi _w\Vdash \forall y . (B(t/x)) $ if, and only if, $\calm , w, \romi _w\Vdash \forall y . ( B(s/x))$. But $x\neq y$ and $t\neq y\neq s$ imply that $ \forall y.(B(s/x))= (\forall y.B)(s/x)$ and $ \forall y.(B(t/x))= (\forall y.B)(t/x)$. Therefore, as desired, $\calm , w, \romi _w\Vdash (\forall y.B)(t/x) $ implies $\calm , w, \romi _w\Vdash (\forall y.B)(s/x) $.

If $x\neq y$ but $t= y\neq s$ or $t\neq y= s$, the case is completely analogous to the previous one, but, as detailed in Definition \ref{def:substitution}, before proceeding with the substitution, we rename the bound variable $y$ in such a way that no capture occurs.

\item  $A=\know _a B$. %By Definition \ref{def:substitution}, 
Since $(\know _a B)(t/x)=\know _a B=(\know _a B)(s/x)$, the statement trivially holds.

\item  $A=j: B$. %By Definition \ref{def:substitution}, 
Since $(j: B)(t/x)=j: B=(j: B)(s/x)$, the statement trivially holds.
\end{itemize}

\item $\vcenter{\infer{\Gamma,\Delta \Rightarrow \know _a B}{\Gamma \Rightarrow \know _a( A\rightarrow B) &\Delta \Rightarrow \know _a A}}\quad$ By inductive hypothesis, for any state $w$  in any model $\calm$ that verifies $\Gamma$ we have  $\calm , w, \romi _w\Vdash \know _a (A\rightarrow B)$  and for any state $w$ that verifies $\Delta $ we have  $\calm , w, \romi _w\Vdash \know _a A$. Let us consider a state $w$ that verifies $\Gamma, \Delta$. According to Definition \ref{def:forcing} for $\know _a$, $\calm , u, \romi _u\Vdash A\rightarrow B$ and $\calm , u, \romi _u\Vdash  A$ hold for any $u$ such that $w\romr_a u$. By Definition \ref{def:forcing} for $\rightarrow$, we have that $\calm , u, \romi _u\Vdash  B$ hold for any $u$ such that $w\romr_a u$. But then, by Definition \ref{def:forcing} for $\know_a$,  $\calm , w, \romi _w\Vdash  \know _a  B$, as desired.

\item $\vcenter{\infer{\Gamma \Rightarrow A}{\Gamma \Rightarrow \know  _a A}}\quad$ By inductive hypothesis, for any state $w$  in any model $\calm$ that verifies $\Gamma$, $\calm , w, \romi _w\Vdash \know _a A$. According to Definition \ref{def:forcing} for $\know _a$, $\calm , u, \romi _u\Vdash A$ holds for any $u$ such that $w\romr_a u$.
But, by Definition \ref{def:model}, $\romr_a$ is reflexive and thus $w\romr_a w$. Hence, by Definition \ref{def:forcing} for $\know_a$,  $\calm , w, \romi _w\Vdash  A$, as desired.

\item $\vcenter{\infer{\Gamma \Rightarrow \know _a \neg \know _a   A }{\Gamma \Rightarrow \neg \know _a  A }}\quad$ By inductive hypothesis, for any state $w$ in any model $\calm$ that verifies $\Gamma$, $\calm , w, \romi _w\Vdash \neg \know _a A$. According to Definition \ref{def:forcing} for $\neg$ and $\know _a$, there must exist a state $u$ such that $w\romr_a u$ and  $\calm , u, \romi _u\not\Vdash  A$ hold. Since $\romr _a$ is Euclidean, for any state $v$ such that $w\romr_a v$, we have that $v\romr_a u$. Hence,  by Definition \ref{def:forcing}, for any state $v$ such that $w\romr_a v$, $\calm , v, \romi _v\Vdash \neg \know _a A$  holds. But this implies that 
 $\calm , w, \romi _w\Vdash  \know _a\neg \know _a A$ holds too, as desired.

\item $\vcenter{ \infer{\Gamma , \Delta\Rightarrow jk:B}{\Gamma \Rightarrow j: (A\rightarrow B) &\Delta  \Rightarrow k: A }}\quad $
By inductive hypothesis, for any state $w$  in any model $\calm$ that verifies $\Gamma$ we have  $\calm , w, \romi _w\Vdash j: (A\rightarrow B)$  and for any state $w$ that verifies $\Delta $ we have  $\calm , w, \romi _w\Vdash k: A$. Let us consider a state $w$ that verifies $\Gamma, \Delta$. According to Definition \ref{def:forcing} for justifications, $\calm , u, \romi _u\Vdash A\rightarrow B$ and $\calm , u, \romi _u\Vdash  A$ hold for any $u$ such that $w\romr_\gamma u$. By Definition \ref{def:forcing} for $\rightarrow$, we have that $\calm , u, \romi _u\Vdash  B$ hold for any $u$ such that $w\romr_\gamma u$. But then, by Definition \ref{def:forcing} for justifications and since $B\in \rome _w(jk)$ whenever $A\rightarrow B \in \rome _w(j)$ and $A \in \rome _w(k)$, we have that $\calm , w, \romi _w\Vdash  jk:B$, as desired. 

\item $\vcenter{ \infer{\Gamma\Rightarrow A }{\Gamma\Rightarrow  j:A}}\quad $  
By inductive hypothesis, for any state $w$  in any model $\calm$ that verifies $\Gamma$, $\calm , w, \romi _w\Vdash j: A$. According to Definition \ref{def:forcing} for justifications, $\calm , u, \romi _u\Vdash A$ holds for any $u$ such that $w\romr_\gamma u$. But, by Definition \ref{def:model}, $\romr_\gamma$ is reflexive and thus $w\romr_\gamma w$. Hence, by Definition \ref{def:forcing} for justifications,  $\calm , w, \romi _w\Vdash  A$, as desired.

\item $\vcenter{\infer{\Gamma \Rightarrow {}!j:(j:A)}{\Gamma \Rightarrow j:A }}\quad$ 
By inductive hypothesis, for any state $w$  in any model $\calm$ that verifies $\Gamma$, $\calm , w, \romi _w\Vdash j: A$. According to Definition \ref{def:forcing} for justifications, $\calm , u, \romi _u\Vdash A$ holds for any $u$ such that $w\romr_\gamma u$. Moreover, since $\romr_\gamma$ is transitive, all successors of the successors of $w$, are also successors of $w$. Hence, by Definition \ref{def:forcing} for justifications, also $\calm , u, \romi _u\Vdash j:A$ holds for any $u$ such that $w\romr_\gamma u$. But, since $j:A\in \rome _w(!j)$ whenever $A\in \rome _w(j)$ and since $\rome _w(!j)\subseteq \rome _v(!j)$, this means that also $\calm , w, \romi _w\Vdash !j:j:A$ holds, as desired.

%=-closure of justif.
\item $\vcenter{\infer{\Gamma , \Delta \Rightarrow j:A(s/x)}{\Gamma \Rightarrow k:t\lid s & \Delta  \Rightarrow j:A(t/x)}} \quad$ 
By inductive hypothesis, for any state $w$  in any model $\calm$ that verifies $\Gamma$, it holds that  $\calm , w, \romi _w\Vdash k:t\lid s$ and, for any state $w$  in any model $\calm$ that verifies $\Delta$, it holds that $\calm , w, \romi _w\Vdash  j:A(t/x)$. Suppose that $w$ verifies $\Gamma$ and $\Delta$. First, it must hold that $ A(t/x)\in  \rome _w(j)$, $ t\lid s\in  \rome _w(k)$. Moreover, according to Definition \ref{def:forcing} for justifications, $\calm , u, \romi _u\Vdash A(t/x)$ and $\calm , u, \romi _u\Vdash t\lid s$ hold for any $u$ such that $w\romr_\gamma u$. Hence, also $\calm , u, \romi _u\Vdash A(s/x)$ holds for any $u$ such that $w\romr_\gamma u$. Moreover, since $\rome _w(j)$ must contain $ A(s/x)$ when $ A(t/x)\in  \rome _w(j)$ and there exists a $k$ such that $ t\lid s \in \rome _w(k)$, we have, by Definition \ref{def:forcing} for justifications, that $\calm , w, \romi _w\Vdash  j:A(s/x)
$ holds, as desired.

\item $\vcenter{\infer{\Gamma , \Delta \Rightarrow j:A(s/x)}{\Gamma \Rightarrow k:s\lid t & \Delta  \Rightarrow j:A(t/x)}}\quad$ This case is completely analogous to the previous one.

\item $\vcenter{\infer{\Gamma \Rightarrow A\rightarrow B}{\Gamma ,A \Rightarrow B}} \quad$ By inductive hypothesis, for any state $w$ in any model $\calm$ that verifies $\Gamma$ and $A$, it holds that  $\calm , w, \romi _w\Vdash B$. We show that each state that verifies $ \Gamma$ also verifies $A\rightarrow B$. Now, consider a generic state $v$ in a generic model $\caln$ such that $\caln , v, \romi _v\Vdash  G $ for each $G \in \Gamma$. If $\caln , v, \romi _v\nVdash A$, we have, by Definition \ref{def:forcing}, that $\caln , v, \romi _v\Vdash A\rightarrow B $, as desired. If, otherwise, $\caln , v, \romi _v\Vdash A$, we know, by local hypothesis, that $\caln , v, \romi _v\Vdash B$. But, by Definition \ref{def:forcing}, $\caln , v, \romi _v\Vdash B$ implies, as desired, $\caln , v, \romi _v\Vdash A\rightarrow B $. 

\item $\vcenter{\infer{\Gamma , \Delta \Rightarrow B}{\Gamma \Rightarrow A\rightarrow  B & \Delta \Rightarrow A}} \quad$ By inductive hypothesis, for any state $w$ in any model $\calm$ that verifies $\Gamma$ we have  $\calm , w, \romi _w\Vdash A\rightarrow B$  and for any state $w$ that verifies $\Delta $ we have  $\calm , w, \romi _w\Vdash A$. Let us consider a state $w$ that verifies $\Gamma, \Delta$. According to Definition \ref{def:forcing} for $\rightarrow$, $\calm , w, \romi _w\Vdash A\rightarrow B$ and $\rightarrow$, $\calm , w, \romi _w\Vdash B$ together imply that $\calm , w, \romi _w\Vdash B$, as desired.

\item $\vcenter{\infer{\Gamma , \Delta\Rightarrow A_1\wedge A_2}{\Gamma \Rightarrow A_1& \Delta \Rightarrow A_2}}\quad$ By inductive hypothesis, for any state $w$ in any model $\calm$ that verifies $\Gamma$ and for $i\in \{1,2\}$, we have that $\calm , w, \romi _w\Vdash A_i$. Let us consider a state $w$ that verifies $\Gamma, \Delta$. According to Definition \ref{def:forcing} for $\wedge$, $\calm , w, \romi _w\Vdash A_1$ and $\calm , w, \romi _w\Vdash A_2$ together imply that $\calm , w, \romi _w\Vdash A_1\wedge A_2$, as desired.

\item $\vcenter{\infer[i\in \{1,2\}]{\Gamma \Rightarrow A_i}{\Gamma \Rightarrow A_1\wedge A_2}}\quad$ By inductive hypothesis, for any state $w$ in any model $\calm$ that verifies $\Gamma$, it holds that  $\calm , w, \romi _w\Vdash A_1\wedge A_2$. We show that each state that verifies $ \Gamma$ also verifies $A_i$ for $i\in \{1,2\}$. Consider then a state $w$ in $\calm$  such that $\calm , w, \romi _w\Vdash  G $ for each $G \in \Gamma$. By the local hypothesis, we have that   $\calm , w, \romi _w\Vdash A_1\wedge A_2$ and hence, by Definition \ref{def:forcing} for $\wedge$, we have that $\calm , w, \romi _w\Vdash A_1 $ and $\calm , w, \romi _w\Vdash A_2 $, as desired.

\item $\vcenter{\infer{\Gamma \Rightarrow P}{\Gamma \Rightarrow \bot}}\quad$ By inductive hypothesis, for any state $w$ in any model $\calm$ that verifies $\Gamma$, it holds that  $\calm , w, \romi _w\Vdash \bot $. But since for any state $w$ in any model $\calm$ it holds that $\calm , w, \romi _w\nVdash \bot $, by contraposition, we have that no state $w$ exists such that $\calm , w, \romi _w\Vdash G $ for all $G\in \Gamma$. Therefore, vacuously, for any state  $w$ in any model $\calm$ such that $w$ verifies $\Gamma$, it holds that $\calm , w, \romi _w\Vdash P $, as desired.

\item $\vcenter{\infer{\Xi\Rightarrow \forall x . A }{\Xi \Rightarrow A}}\quad$ By inductive hypothesis, all states that verify $\Xi$ also verify  $A$. Since the premiss is derivable and because of the conditions on the rule application, it is possible to replace $y$ by any other term $t$ in the derivation of $\Xi \Rightarrow A$ and obtain a derivation of $\Xi \Rightarrow A[t/x]$. Since $\Xi \Rightarrow A[t/x]$ is derivable by a derivation containing $n-1$ rule applications, we know that all states that verify $\Xi$ also verify  $A[t/x]$, for any term $t$. Now, {\ehi Definition \ref{def:model} (existence, for any state $w$ and element $e\in \romu$, of a term $t$ such that $\romi _w (t)=e$) and Definition \ref{def:forcing} for $\forall$}, imply that all states that verify $\Xi$ also verify $ \forall x . A$.

\item $\vcenter{\infer{\Gamma \Rightarrow A[t/x]}{\Gamma \Rightarrow \forall x . A}} \quad $  By inductive hypothesis, all states that verify $\Gamma $ also verify  $\forall x . A$. Hence, according to Definition \ref{def:forcing} for $\forall$, for any state $w$ in which $\Gamma$ holds, 
{\ehi for any term $t$, $\calm , w, \romi_w  \Vdash A[t/x]$ must hold. }
%for any $x$-variant $f$ of $\romi _w$, $\calm , w,\cani _{w}\Vdash A$ must hold. 
%According to Definition \ref{def:x-variant},  for any term $t$, there exists an $x$-variant $f'$ among these $x$-variants of $\romi _w$  such that $f'(x)=\romi ^w(t)$. By  Lemma \ref{lem:subst-variant-corr} and since $\calm , w, f' \Vdash A$, we can conclude, as desired, that $\calm , w, \romi_w \Vdash A[t/x] $.

\item $\vcenter{\infer{\Gamma \Rightarrow A}{\Gamma  \Rightarrow \neg \neg A}}\quad$ By inductive hypothesis, for any state $w$ in any model $\calm$ that verifies $\Gamma$, it holds that  $\calm , w, \romi _w\Vdash \neg \neg A$. We show that each state that verifies $ \Gamma$ also verifies $A$. Consider then a state $w$ in $\calm$  such that $\calm , w, \romi _w\Vdash  G $ for each $G \in \Gamma$. By the local hypothesis, we have that $\calm , w, \romi _w\Vdash \neg \neg A$ and hence, by Definition \ref{def:forcing} for $\rightarrow$ (indeed, $\neg A$ abbreviates $(A\rightarrow \bot )\rightarrow \bot $), we have that if $\calm , w, \romi _w\Vdash \neg A $ then $\calm , w, \romi _w\Vdash \bot$, which is impossible. Hence, $\calm , w, \romi _w\not\Vdash \neg A $. By Definition \ref{def:forcing} for $\rightarrow$, this implies that $\calm , w, \romi _w\Vdash A $ and $\calm , w, \romi _w\not\Vdash \bot$. This clearly boils down to $\calm , w, \romi _w\Vdash A $, which is exactly  what we needed to show.
\end{itemize}
\end{proof}

\noindent {\bf Lemma \ref{lem:ultrafilter}.}
{\it For any consistent set of formulae $\Sigma$, there exists a maximally consistent, counterexemplar  set (MCC) of formulae $\Sigma^*$ such that $\Sigma\subseteq \Sigma^*$.
}
\begin{proof}We construct an infinite chain $\Sigma_0\subset \Sigma_1 \subset \Sigma_2 \subset  \Sigma_3\dots $ of consistent sets of formulae
%%NEW FOR CORRECTION
by standardly using a countable set of new variables. Let $\Sigma_0=\Sigma$ and, given any ordering $A_1, A_2, A_3\dots $ of all the formulae of our language, we define the set $\Sigma_n$, for $n>0$, as follows:\begin{itemize} 
\item $\Sigma_n = \Sigma_{n-1}\cup \{A_n\}$ if $\Sigma_{n-1}\cup \{A_n\}$ is consistent,

\item $\Sigma_n = \Sigma_{n-1}\cup \{\neg A_n\}$ if  $\Sigma_{n-1}\cup \{A_n\}$ is not consistent and either $A_n\neq \forall x.B$ for any variable $x$ and formula $B$, or $A_n= \forall x.B$ and $\neg B[t/x] \in \Sigma_{n-1}$ for some term $t$,
 
\item $\Sigma_n = \Sigma_{n-1}\cup \{\neg A_n , \neg B[y/x]\}$ for some variable $y$ that does not occur in $\Sigma_{n-1}\cup \{\neg A_n\}$ if  $\Sigma_{n-1}\cup \{A_n\}$  is not consistent, $A_n= \forall x.B$, and, for each term $t$, $B[t/x] \notin \Sigma_{n-1}$.
\end{itemize}

Let\[\Sigma^*=\bigcup_{i\in \mathbb{N}} \Sigma_i\]

In order to show that $\Sigma^*$ is consistent, we first show that each set $\Sigma_n$ for $n\in \mathbb{N}$ is consistent. We do it by induction on $n$. If $n=0$, this is an assumption. Suppose that $\Sigma_i$ is consistent for any $i<n$, we show that $\Sigma_n$ is consistent. 

We have three cases:
\begin{itemize}

\item $\Sigma_n= \Sigma_{n-1}\cup \{A_n\}$. In this case, $\Sigma_n $ is obviously consistent since $\Sigma_n$ is defined as $\Sigma_{n-1}\cup \{A_n\}$ only if $\Sigma_{n-1}\cup \{A_n\}$ is consistent. 

\item $\Sigma_n=\Sigma_{n-1}\cup \{\neg A_n\}$ because $\Sigma_{n-1}\cup \{A_n\}$ is not consistent and either $A_n\neq \forall x.B$ for any variable $x$ and formula $B$, or $A_n= \forall x.B$ and $B[t/x] \in \Sigma_{n-1}$ for some term $t$. By reasoning indirectly, suppose that $\Sigma_n=\Sigma_{n-1}\cup \{\neg A_n\}$ is not consistent. This means that there exists a derivation of  $\Gamma , \neg A_n \Rightarrow \bot $ for some $\Gamma \subseteq \Sigma_{n-1}$. The formula $\neg A_n$ must be among the hypotheses of the derivation since, by inductive hypothesis, $s_{n-1} $ is consistent. Since there is a derivation with conclusion $\Gamma , \neg A_n \Rightarrow \bot $,  by an implication introduction (since $\neg A_n$ abbreviates $A_n\rightarrow \bot$), we can obtain a derivation with conclusion  $\Gamma \Rightarrow \neg \neg A_n$. Now, because of how $\Sigma_n$ si defined, we also know that $\Sigma_{n-1}\cup \{ A_n\} $ is not consistent. Hence, there exists a derivation with conclusion $\Delta , A_n \Rightarrow \bot $ for $ \Delta \subseteq \Sigma_{n-1}$. Again, the formula $A_n$ must be among the hypotheses of the derivation since, by inductive hypothesis, $\Sigma_{n-1} $ is consistent. Since there is a derivation with conclusion  $\Delta , A_n \Rightarrow \bot $,  by an implication introduction (since $\neg\neg  A_n$ abbreviates $\neg A_n\rightarrow \bot$), we can construct a derivation with conclusion $\Delta \Rightarrow \neg A_n$. Therefore, we have two derivations, one with conclusion $\Gamma \Rightarrow \neg\neg A_n$ and one with conclusion $\Delta \Rightarrow \neg A_n$ for $\Gamma \cup \Delta\subseteq \Sigma_{n-1}$. But this clearly implies that $\Sigma_{n-1}$ is not consistent: by an implication elimination we can show that $\bot$ is derivable from a subset of $\Sigma_{n-1}$. Since this goes against our assumptions about the consistency of $\Sigma_{n-1}$, we have that $\Sigma_n$ must be consistent also in this case.

\item $\Sigma_n = \Sigma_{n-1}\cup \{\neg A_n , \neg B[y/x]\}$ for some variable $y$ that does not occur in $\Sigma_{n-1}\cup \{\neg A_n\}$ because $\Sigma_{n-1}\cup \{A_n\}$ is not consistent, $A_n= \forall x.B$, and, for each term $t$, $B[t/x] \notin \Sigma_{n-1}$. By using exactly the same argument employed in the previous case, we can conclude that $\Sigma_{n-1}\cup \{\neg A_n\}= \Sigma_{n-1}\cup \{\neg \forall . B\}$ is consistent. Therefore, it is enough to show that, if $\Sigma_{n-1}\cup \{\neg \forall . B\}$ is consistent and $y$ dos not occur in $\Sigma_{n-1}\cup \{\neg \forall x . B\}$, then $\Sigma_{n-1}\cup \{\neg \forall . B , \neg B[y/x]\}$ is consistent too. By reasoning indirectly, suppose that $\Sigma_{n-1}\cup \{\neg \forall . B , \neg B[y/x]\}$ is not consistent. Hence, there exists a derivation of  $\Gamma , \neg B[y/x] \Rightarrow \bot $ for some $\Gamma \subseteq \Sigma_{n-1}\cup \{\neg \forall x.B\}$. From this derivation, since $y$ does not occur in $\Gamma \subseteq \Sigma_{n-1}\cup \{\neg \forall x.B\} $, we can construct the following derivation:\[\infer{\Gamma \Rightarrow  \forall x. B}{\infer{\Gamma  \Rightarrow  B[y/x]}{\infer{\Gamma  \Rightarrow \neg \neg B[y/x]}{\infer*{\Gamma , \neg B[y/x] \Rightarrow \bot }{}}}}\]where the first inference displayed is an implication introduction, the second a double negation elimination and the third a $\forall$ introduction.

Moreover, obviously, there exists a derivation with conclusion $ \neg \forall x.B  \Rightarrow \neg \forall x.B $ and $\{ \neg \forall x.B\}\subseteq \Sigma_{n-1}\cup \{\neg \forall . B \}$. Therefore, we have two derivations, one with conclusion $\Gamma \Rightarrow \forall x.B$ and one with conclusion $ \forall x.B\Rightarrow \neg \forall x.B$ where $\Gamma \cup \{\neg \forall . B\}\subseteq \Sigma_{n-1}$. This clearly implies, by an implication elimination, that $ \Sigma_{n-1}\cup \{\neg \forall . B \} \vdash \bot $ and thus that $ \Sigma_{n-1}\cup \{\neg \forall . B \} $  is not consistent. Since this goes against our assumptions, we have also in this case that $\Sigma_n$ must be consistent also in this case.
\end{itemize}
Hence, we can conclude that each set $\Sigma_n$ for $n\in \mathbb{N}$ must be consistent.

We show now that, from the fact that each element of the chain $\Sigma_0\subset \Sigma_1 \subset \Sigma_2 \subset  \Sigma_3\dots $ is consistent, it follows that also $\Sigma^*=\bigcup_{i\in \mathbb{N}} \Sigma_i$ is consistent. Suppose, indeed, that it is not. Then, by Definition \ref{def:max-cons},  $\Theta  \vdash \bot$ for a multiset $\Theta  $  of elements of $\Sigma^*$. But by Definition \ref{def:der}, there must exist a finite multiset  $\Theta '\subseteq\Theta $ such that $\Theta  ' \vdash \bot$. By construction of $\Sigma^*$, there must exist a $\Sigma_n $ such that $\Theta  '$ is a multiset of elements of $\Sigma_n$. If all formulae in $\Theta  '$ are in $\Sigma$, then $n=0$ and $\Sigma_n=\Sigma$. Otherwise, some formulae in $ \Theta  '$ must have been added to some element of the chain  $\Sigma_0\subset \Sigma_1 \subset \Sigma_2 \subset  \Sigma_3\dots $ during the construction of $\Sigma^*$. Consider the last element of $\Theta  '$ considered during the enumeration $A_1, A_2, A_3\dots $ used for the construction of $\Sigma^*$ and suppose it is $A_m$. Then we have our $\Sigma_n$: it is enough to set $n=m$. Thus, $\Sigma_n  \vdash \bot$. But this, obviously, means, against our assumptions, that $\Sigma_n$ is inconsistent. Hence,  $\Sigma^*$ must be consistent.

What is left to show is that $\Sigma^*$ is maximally consistent and counterexemplar. 

Let us first focus on the fact that $\Sigma^*$ is maximally consistent. By reasoning indirectly, let us suppose that there exists a formula $A\notin \Sigma^*$ such that $\Sigma^*\cup \{A\}$ is consistent. Now, even though the formula $A$ is not an element of $\Sigma^*$, it is an element of the enumeration $A_1, A_2, A_3 \dots $ Suppose that $A = A_n$. Therefore, it has been considered in order to define the set $\Sigma_n $ of the chain $\Sigma_0\subset \Sigma_1 \subset \Sigma_2 \subset  \Sigma_3\dots $ Since $ A=A_n$ can be consistently added to $\Sigma^*$ is must be possible to consistently add it to any subset of $\Sigma^*$. Hence, $\Sigma_n = \Sigma_{n-1} \cup \{A\}$ and $A\in \Sigma^*=\bigcup_{i\in \mathbb{N}} \Sigma_i$. But this contradict our assumption that $ A\notin \Sigma^*$. 

Finally, we need to show that $\Sigma^*$ is counterexemplar. Consider any element of $\Sigma^*$ of the form $\neg \forall x.A$. We show that $\Sigma^*$ is counterexemplar with respect to $\neg \forall x.A$, that is, that $\Sigma^*$    contains an element of the form $\neg A[t/x]$. Now, $\neg \forall x.A$ is either an element of $\Sigma$ or  has been added to a set $\Sigma_n $ during the definition of the chain $\Sigma_0\subset \Sigma_1 \subset \Sigma_2 \subset  \Sigma_3\dots $ In the second case, $\Sigma^*$ is clearly counterexemplar with respect to $\neg \forall x.A$ because of how $\Sigma_n$ has been defined. Let us then consider the first case. Now, even though the  formula $\neg \forall x.A$  is an element of $\Sigma$, it is also an element of the enumeration $A_1, A_2, A_3 \dots $ Suppose that $\neg \forall x.A= A_n$. Therefore, it has been considered in order to define the set $\Sigma_n $ of the chain $\Sigma_0\subset \Sigma_1 \subset \Sigma_2 \subset  \Sigma_3\dots $ Rather obviously, since $\neg \forall x.A\in \Sigma=\Sigma_0\subseteq \Sigma_{n-1}$ and since $\Sigma_{n-1}$ is consistent, $\Sigma_{n-1} \cup \{\neg \forall x.A\}=\Sigma_{n-1}$ cannot be inconsistent. Thus, $\Sigma_n$ has been defined as either $\Sigma_{n-1} \cup \{\neg \forall x.A , \neg A[y/x] \}$ for a $y$ that does not occur in $\Sigma_{n-1} \cup \{\neg \forall x.A\}$ or $\Sigma_{n-1} \cup \{\neg \forall x.A \}$ in case a formula of the form $\neg A[t/x]$ was already contained in $\Sigma_{n-1}$. In either case, $\Sigma^*$ is counterexemplar with respect to $\neg \forall x.A$.
\end{proof}

\noindent {\bf Lemma \ref{lem:canonical-model-model}.} {\it The canonical model $\canm$ is a model.}
\begin{proof}
We show that $\canr _a$ is reflexive and Euclidean for any agent name $a$. 
As for reflexivity, consider any agent name $a$ and state $w\in \canw$. Since $\know_a A \vdash  A$ and since Lemma \ref{lem:closure} guarantees us that MCC are closed with respect to derivability, we have that, for any $\know _a A\in w $, $A\in w$. By Definition \ref{def:canonical-model} then, $w\canr _a w$. As for Euclideanness, consider any agent name $a$ and states $w,v,w\in \canw$ and suppose that $w\canr _a v$ and $w \canr _a  u$. We need to show that $v\canr _a u$. Now, if $w\canr _a v$ and $w \canr _a  u$ hold, then $\know _a ^-(w)\subseteq v$ and $\know _a ^-(w)\subseteq u$ hold. Suppose, by reasoning indirectly, that  $v\canr _a u$ does {\it not} hold. By Definition \ref{def:canonical-model}, $\know _a ^-(v)\not\subseteq u$. Hence, there exists a formula $A$ such that $\know _a A\in v$ and $A\notin u$. Since $A\notin u$ and since $\know _a ^-(w)\subseteq u$, $\know _a A\notin  w$. This, the fact that $w$ is an MCC and Lemma \ref{lem:maximality}, imply that $\neg \know _a A\in  w$. But, since $\neg \know _a A\vdash  \know _a \neg \know _a A$ holds and since Lemma \ref{lem:closure} guarantees us that MCC are closed with respect to derivability, we have that $\know _a \neg \know _a A\in w $. This and the fact that $\know _a ^-(w)\subseteq v$ holds imply that $\neg \know _a A\in v$ holds. But this contradicts the fact that $\know _a A\in v$, because $v$ is and MCC and thus consistent. Hence, we can conclude, as desired, that, if $w\canr _a v$ and $w \canr _a  u$ hold, then  $v\canr _a u$ holds.

We show that $\canr _\gamma$ is reflexive and transitive. As for reflexivity, consider any  state $w\in \canw$. Since $j: A \vdash  A$ and since Lemma \ref{lem:closure} guarantees us that MCC are closed with respect to derivability, we have that, for any $j: A\in w $, $A\in w$. By Definition \ref{def:canonical-model} then, $w\canr _\gamma w$. As for transitivity, consider any agent name $a$ and states $w,v,w\in \canw$ and suppose that $w\canr _\gamma v$ and $v \canr _\gamma  u$. We need to show that $w\canr _\gamma u$. Now, if $w\canr _\gamma v$ and $v \canr _\gamma  u$ hold, then $\gamma ^-(w)\subseteq v$ and $ \gamma^-(v)\subseteq u$ hold. Suppose, by reasoning indirectly, that  $w\canr _\gamma u$ does {\it not} hold. By Definition \ref{def:canonical-model}, $\gamma ^-(w)\not\subseteq u$. Hence, there exists a formula $A$ such that $j: A\in w$ and $A\notin u$. But since $j: A\in w$, $j: A\vdash !j:j:A$ and $w$ is an MCC, Lemma \ref{lem:closure} guarantees us that $ !j:j:A\in w$. But we also know that $\gamma ^-(w)\subseteq v$ and hence that $j:A\in v $. And $j:A\in v $ along with the fact that $\gamma ^-(v)\subseteq u$ implies that $A\in u$, which contradicts the fact that $A\notin u$. Hence, we must conclude that $w\canr _\gamma v$ and $v \canr _\gamma  u$ do indeed imply that $w\canr _\gamma u$.

We now show that $\cani$ is suitably defined as a family of interpretation functions on $\canw$. In particular, we have to show, first of all that, for any element $e\in \canu$ and state $w\in \canw$, there exists a term $t$ such that $\cani_w (t)=e$; secondly, that the interpretation of $\lid $ is always a reflexive, symmetric and transitive relation; and, thirdly, that the interpretation of predicates is substitutive, as indicated in Definition \ref{def:model}. 

We first show that, at any state $w\in \canw$,  all elements of $\canu$ are the interpretation of some term. Now, according to the construction of $\canu$ in Definition \ref{def:canonical-model}, it is clear that, for any $e \in \canu$, $e$ must be of the form $\underline{t}$ for some term $t$. Definition \ref{def:canonical-model}, moreover, indicates that $\cani_w (t)=\underline{t}$, for any $w\in \canw$. Hence, we have our term $t$ such that $\cani_w (t)=\underline{t}=e$.

Let us then focus on $\lid $ first. Since Lemma \ref{lem:closure} guarantees us that all elements of $\canw$ are closed with respect to derivability, by the existence of the following derivations\[\infer{ \Rightarrow t\lid t}{}\qquad\infer{t\lid s \Rightarrow s\lid t}{\infer{t\lid s \Rightarrow t\lid s}{}}\qquad \infer{ u\lid t, t\lid v \Rightarrow u\lid v}{\infer{u\lid t \Rightarrow u\lid t}{}&\infer{t\lid v \Rightarrow t\lid v}{}} \]it is easy to conclude that, for any $w\in \canw$, $\cani_w(\lid)$  complies with definition \ref{def:model}.
 
Let us then consider $\cani_w(P)$ for a generic predicate $P$ of arity $n$ and state $w\in \canw$. Now, suppose that $(\underline{t_1} , \dots  , \underline{t_n})\in\cani _w (P)$ and that $(\underline{t_i},\underline{s})\in \cani_w (\lid )$ for some $1\leq i \leq n$. We need to show that 
$(\underline {t_1'} , \dots , \underline {t_n'})\in \cani _w (P)$ where, for any $1\leq j\leq n$,  $\underline {t_j'}=\underline {t_j}$ {\ehi in case $j\neq i$} and $\underline {t_j'}=\underline{s}$ {\ehi in case $j=i$}. By Definition \ref{def:canonical-model}, $(\underline{t_1} , \dots  , \underline{t_n})\in\cani _w (P)$ implies that $P(t_1 , \dots  , t_n)\in w $, and 
$(\underline{t_i},\underline{s})\in \cani_w (\lid )$ implies that $t_i \lid s\in w$. Moreover, Lemma \ref{lem:closure} guarantees us that $w$ is closed under derivability. Hence, the existence of the following derivation
\[\qquad \infer{t_i\lid s , P(t_1 , \dots  , x, \dots , t_n)(t_i/x) \Rightarrow P(t_1 , \dots  , x, \dots , t_n)(s/x)}{\infer{t_i\lid s\Rightarrow t_i\lid s}{}& \infer{P(t_1 , \dots  , x, \dots , t_n)(t_i/x) \Rightarrow P(t_1 , \dots  , x, \dots , t_n)(t_i/x)}{}}\]
%(where the list of terms $t_1 , \dots  , x, \dots , t_n$ is constructed in such a way that  $x$ occurs at any position at which a term $t=t_i$ occurs in $t_1 , \dots , t_n$) 
guarantees us that  $P(t_1 , \dots  , t_n)\in w $ and $t_i \lid s\in w$ imply that
\[P(t_1' , \dots  , t_n')\in w \]
where, for any $1\leq j\leq n$,  $t_j'=t_j$ in case $ t_j\neq t_i$, and $t_j'= s$ in case $ t_j= t_i$. But then, by the clause on $\cani$ in  Definition \ref{def:canonical-model}, we obtain exactly what we need:
\[(\underline {t_1'} , \dots , \underline {t_n'})\in \cani _w (P)\]
where, for any $1\leq j\leq n$,  $\underline {t_j'}=\underline {t_j}$ in case $\underline {t_j}\neq \underline {t_i}$, and $\underline {t_j'}=\underline{s}$ in case $\underline {t_j}=\underline {t_i}$.

We show that, for any state $w\in \canw$, the function $ \cane _w (t)$ complies with the following conditions:
\begin{enumerate}
\item \label{can-subs} for any state $v\in \canw$ and term $t$, if $w \canr _\gamma v$, then $ \cane_w(t)\subseteq  \cane_v(t)$,

%%%%NecJustifications
%\item \label{can-nec-justif} { if $j_1:A_1 ,  \dots , j_n:A_n \vdash A$ holds, then $A\in \rome_w(x^{j_1 , \dots , j_n})$ for the designated variable $x^{j_1 , \dots , j_n}$,}

%\item  \label{can-id-justif}  for any two terms $t,s$, if $( \cani_w(t), \cani_w(s))\in  \cani_w(\lid)$, then $  \cane_w(t)= \cane_w(s)$, 

\item  \label{can-id-justified} for any term $t$, 
if $\cani_w(t)= \cani_w(s)$ and $ A[t/x]\in  \cane_w(t)$, then $ A[s/x]\in  \cane_w(t)$,
%if $( \cani_w(t), \cani_w(s))\in  \cani_w(\lid)$ and $ A[t/x]\in  \cane_w(t)$, then $ A[s/x]\in  \cane_w(t)$,

\item \label{can-checker}  for any term $j$, if $A\in  \cane_w(j)$, then $ j:A \in \cane_w(!j)$.

\item \label{can-applic} for any two terms $j,k$, if $ A\rightarrow B \in  \cane_w(j)$ and $ A \in  \cane_w(k)$, then $ B\in \cane_w(jk)$.

% =-closure justifications
\item \label{can-id-clos-just}  for any term $j$, 
if $A(t/x)\in \cane_w(j)$ and there exists a $k$ such that $
t\lid s \in \cane_w(k)$ or $
s\lid t \in \cane_w(k)$, then $A(s/x)\in \cane_w(j)$.
\end{enumerate}

Let us consider \cref{can-subs}. If $A\in \cane_v(j) $, then $j:A\in  v$. But since $j:A\vdash !j:j:A$, we also have $ !j:j:A\in  v$. By Definition \ref{def:canonical-model}, $ v \canr _\gamma w$ implies that $j:A\in w$ and thus that $A\in \cane_v(j) $.

%%%NecJustifications
%{ Let us consider  \cref{can-nec-justif}. If $j_1:A_1,  \dots , j_n:A_n \vdash A$ holds, then $j_1:A_1,  \dots , j_n:A_n \vdash x^{j_1,  \dots , j_n}: A$ holds. But then, by Lemma \ref{lem:closure}, we have that $x^{j_1,  \dots , j_n}: A\in w$. Hence, by Definition \ref{def:canonical-model}, $A\in \cane _w (x^{j_1,  \dots , j_n})$, as desired. }

%Let us consider \cref{can-id-justified}. If $( \cani_w(t), \cani_w(s))\in  \cani_w(\lid)$, we have that  $t\lid s\in w $. If, moreover, $ A[t/x]\in  \cane_w(j)$, we have that $j:A[t/x]\in w$. Since $t\lid s, j:A[t/x]\vdash j:A[s/x]$ and by Lemma \ref{lem:closure}, $j:A[s/x]\in w $, which implies that $ A[s/x]\in  \cane_w(j)$ by Definition \ref{def:canonical-model}. 

{\ehi Let us consider \cref{can-id-justified}. If $ \cani_w(t)= \cani_w(s)$, we have that  $t=s$. If, moreover, $ A[t/x]\in  \cane_w(j)$, we have that $j:A[t/x]\in w$. Since $t=s$, we have $j:A[s/x]$, and by Lemma \ref{lem:closure}, $j:A[s/x]\in w $, which implies that $ A[s/x]\in  \cane_w(j)$ by Definition \ref{def:canonical-model}. 
}

Let us consider \cref{can-checker}. If $A\in  \cane_w(j)$, we have that $j:A\in w $. Since $j:A\vdash !j:j:A$ and by Lemma \ref{lem:closure}, $!j:j:A \in w $, which implies that $ j:A\in  \cane_w(!j)$ by Definition \ref{def:canonical-model}.

Let us consider \cref{can-applic}. If $ A\rightarrow B \in  \cane_w(j)$ and $ A \in  \cane_w(k)$, we have that $j:A\rightarrow B, k:A\in w$. 
Since $j:A\rightarrow B, k:A\vdash jk:B$ and by Lemma \ref{lem:closure}, $jk:B \in w $, which implies that  $ B\in \cane_w(jk)$ by Definition \ref{def:canonical-model}.

Let us consider \cref{can-id-clos-just}. If $A(t/x)\in \cane_w(j)$, we have that $j:A(t/x)\in w$ by Definition \ref{def:canonical-model}. If, moreover, $t\lid s\in \cane_w(k)$ or $s\lid t\in \cane_w(k)$, we have that $k:t\lid s$ or, respectively, $k:s\lid t$. Thus, by Lemma \ref{lem:closure}, we can conclude that $j:A(s/x)\in w$. But this implies that $A(s/x)\in \cane_w(j)$ by Definition \ref{def:canonical-model}.
\end{proof}

\noindent {\bf Lemma \ref{lem:existence}} (Existence lemma){\bf .} {\it 
For any $w\in \canw$, if $\neg \know  _a A\in w $, then there exists a $v\in \canw$ such that $\neg A\in v$ and $w\canr_a v$.}
\begin{proof}

  Assume that $\neg \know  _a A\in w $. 
  We show that we can construct a maximally consistent and counterexemplar set containing $\neg A$ in the same language as $w$. 
  %Let $\overline{v} = \{\neg A\}\cup \know_a ^-(w)$.
  Suppose that $U_1, U_2, U_3\dots $ is the list of all formulae of the language of the form $ \forall x D$. We define the succession $G _0, G _1 , G_2 ,  \dots $ by taking $G _0 = \neg A$ and $G_{n+1}= G _n \wedge (\neg \forall x D \rightarrow \neg D [y/x]) $ where $\forall x D =U_n$ and $y$ is chosen in such a way that  $\know_a  ^-(w)\cup \{G _n \wedge  (\neg \forall x D \rightarrow \neg D [y/x])\} $ is consistent. We show by induction on $n$ that it is possible to always choose a suitable $y$.

As for the base case,
% in order to prove the consistency of $\know_a  ^-(w)\cup \{G _0\}$ we use the standard argument for propositional modal logic, presented in detail after this inductive argument.
% Assume that $\neg \know  _a A\in w $.
let $\overline{v} = \know_a ^-(w)\cup \{\neg A\}= \know_a ^-(w)\cup \{G_0\}$. We show that $\overline{v}$ is consistent. By reasoning indirectly, let us suppose that $\overline{v}$ is inconsistent. If this is the case then, for a list of formulae $B_1 , \ldots , B_n \in \know_a ^-(w)$, it must be the case that  $B_1 \wedge \ldots \wedge B_n \rightarrow \neg\neg A$ is a theorem of the logic that we are considering. Indeed, if $\{B_1,  \ldots  B_n , \neg A\}$ is contradictory, we can derive $\bot $ from it by a derivation of the form
\[\infer*{B_1 , \ldots , B_n , \neg A \Rightarrow  \bot }{\infer{B_1  \Rightarrow  B_1 }{}&\dots &\infer{B_n  \Rightarrow  B_n }{}&\infer{\neg A\Rightarrow \neg A}{}
}\]
and hence we can also construct the following derivation:
\[\infer{\Rightarrow\know _a (B_1 \wedge \ldots \wedge B_n) \rightarrow\know _a  A}{\infer{\Rightarrow\know _a (B_1 \wedge \ldots \wedge B_n \rightarrow A)}{\infer{\Rightarrow B_1 \wedge \ldots \wedge B_n \rightarrow A}{\infer{ B_1 \wedge \ldots \wedge B_n \Rightarrow A}{\infer{B_1 \wedge \ldots \wedge B_n \Rightarrow \neg \neg A }{\infer*{B_1 \wedge \ldots \wedge B_n , \neg A \Rightarrow  \bot }{\infer{B_1 \wedge \ldots \wedge B_n  \Rightarrow  B_1 }{\infer{B_1 \wedge \ldots \wedge B_n  \Rightarrow  B_1 \wedge \ldots \wedge B_n }{}}&\dots &\infer{B_1 \wedge \ldots \wedge B_n  \Rightarrow  B_n }{\infer{B_1 \wedge \ldots \wedge B_n  \Rightarrow  B_1 \wedge \ldots \wedge B_n }{}}&\infer{\neg A\Rightarrow\neg  A}{}}
}}}}}\]by which we show that $\know _a (B_1 \wedge \ldots \wedge B_n) \rightarrow\know _a   A$ is a theorem of our logic. Let us call this one $\delta_1$ and let us call the following one $\delta _2$:
{\footnotesize \[ \infer{\know _a B_1 \wedge \ldots \wedge \know _a  B_n\Rightarrow \know _a (B_1 \wedge \ldots \wedge B_n)}{\infer*[\rightarrow \textrm{elim.}]{\bigwedge\know _a \overline{B}\Rightarrow \know _a B_n  \rightarrow \know _a (B_1 \wedge \ldots \wedge B_n)}{\infer{\bigwedge\know _a \overline{B}\Rightarrow \know _a B_2  \rightarrow\dots  \rightarrow \know _a B_n  \rightarrow \know _a (B_1 \wedge \ldots \wedge B_n)}{\infer{\Rightarrow \know _a B_1  \rightarrow \know _a B_2  \rightarrow \dots  \rightarrow \know _a B_n  \rightarrow \know _a (B_1 \wedge \ldots \wedge B_n)}{
\infer*[\rightarrow\textrm{intro.}]{\know _a B_1\Rightarrow  \know _a B_2  \rightarrow \dots  \rightarrow \know _a B_n  \rightarrow \know _a (B_1 \wedge \ldots \wedge B_n)}{\infer{
\know _a B_1, \ldots , \know _a  B_{n-1}
 \Rightarrow  \know _a B_n  \rightarrow \know _a (B_1 \wedge \ldots \wedge B_n)}{\infer{ \know _a B_1, \ldots , \know _a  B_{n-1}  , \know _a  B_n\Rightarrow    \know _a (B_1 \wedge \ldots \wedge B_n)}{\infer*[\wedge\textrm{intro.}]{\know _a B_1, \ldots , \know _a  B_{n-1}  , \know _a  B_n \Rightarrow B_1 \wedge \ldots \wedge B_n}{ \infer{\know _a  B_1\Rightarrow B_1}{ \infer{\know _a B_1\Rightarrow\know _a B_1}{}}  &\dots & \infer{\know _a  B_n\Rightarrow B_n}{\infer{\know _a B_n\Rightarrow\know _a B_n}{}}}}}}
}&\infer*[\wedge\textrm{elim.}]{\bigwedge\know _a \overline{B}\Rightarrow \know _a B_1}{\infer{\bigwedge\know _a \overline{B} \Rightarrow\bigwedge\know _a \overline{B}}{}}}}  & \infer*[\wedge\textrm{elim.}]{\bigwedge\know _a \overline{B} \Rightarrow \know _a B_n}{\infer{\bigwedge\know _a \overline{B} \Rightarrow\bigwedge\know _a \overline{B}}{} }}\]}
where we denote by $\bigwedge \know _a \overline{B}$ the formula $\know _a B_1 \wedge \ldots \wedge \know _a  B_n$ and we recall the reader that we associate $\rightarrow $ to the right.

By derivations $\delta_1$ and $\delta_2$ above, we can finally construct the following derivation:\[\infer{\know _a B_1 \wedge \ldots \wedge \know _a  B_n\Rightarrow \know _a A}{\deduce{\know _a B_1 \wedge \ldots \wedge \know _a  B_n \Rightarrow \know _a (B_1 \wedge \ldots \wedge B_n) }{\delta_2}&\deduce{\Rightarrow\know _a (B_1 \wedge \ldots \wedge B_n) \rightarrow\know _a   A}{\delta_1}}\]
which shows that $\know _a B_1 \wedge \ldots \wedge \know _a  B_n\vdash \know _a  A $ and hence that $\know _a  A \in w$, because $\know _a B_1 \wedge \ldots \wedge \know _a  B_n\in w $ and because  Lemma \ref{lem:closure} guarantees that MCC are closed with respect to derivability. But $\know _a  A \in w$  contradicts the hypotheses that $  \neg \know _a A \in w$ since $w$  is consistent. Therefore
% $\overline{v} = \{\neg A\}\cup \know_a ^-(w)$
 $\overline{v} = \know_a ^-(w)\cup \{\neg A\}= \know_a ^-(w)\cup \{G_0\}$ must be  consistent.

% Hence, according to  Lemma \ref{lem:ultrafilter}, we can then extend $\overline{v}$ to the MCC $\overline{v}^*$. Now, by Definition \ref{def:canonical-model}, $\overline{v}^*\in \canw$. Obviously, $\neg A \in \overline{v}^*$ and, again by Definition \ref{def:canonical-model}, $w\canr_a  \overline{v}^*$ since $ \know_a ^-(w)\subseteq \overline{v}^* $. We have thus showed, as desired,  that  $\overline{v}^* \in \canw$ exists, that $\neg A\in \overline{v}^* $ and $w\canr_a \overline{v}^* $.

%%NEW FOR CORRECTION
  Suppose then that $\know_a  ^-(w)\cup \{G _n \}$ is consistent, we prove, by contradiction, that there exists a $y$ for which   $\know_a  ^-(w)\cup \{G _n \wedge(\neg \forall x D \rightarrow \neg D [y/x])\} $ is consistent too. Let us assume that that no such $y$ exists. For each $y$ there must be  a list $B_1, \dots , B_n$ of formulae in $\know_a  ^-(w)$ such that  $B_1\wedge \dots \wedge B_n \wedge G_n \wedge (\neg \forall x D \rightarrow \neg D [y/x])$ is contradictory. Thus, $B_1\wedge \dots \wedge B_n \rightarrow  (G_n \rightarrow \neg (\neg \forall x D \rightarrow \neg D [y/x]))$ is a theorem of the logic.  Hence, by necessitation, also $\know_a( B_1\wedge \dots \wedge B_n \rightarrow  (G_n \rightarrow  \neg (\neg \forall x D \rightarrow \neg D [y/x])))$ is a theorem, and, by the axiom K, also $\know_a( B_1\wedge \dots \wedge B_n) \rightarrow  \know_a(G_n \rightarrow  \neg (\neg \forall x D \rightarrow \neg D [y/x]))$. But since $B_1, \dots ,  B_n\in \know_a  ^-(w) $, we have that $\know_a B_1, \dots ,  \know_a B_n\in w $ and, by conjunction introduction and the axiom K, also that  $ \know_a( B_1\wedge \dots \wedge B_n)\in w$. Hence, by implication elimination, $\know_a (G_n \rightarrow  \neg (\neg \forall x D \rightarrow \neg D [y/x])) \in w$.

Now, we have assumed that $y$ is a generic variable occurring in $w$. Hence, $\know_a (G_n \rightarrow  \neg (\neg \forall x D \rightarrow \neg D [y/x])) \in w$ for each $y$. Thus, $w$ cannot contain $ \neg \forall z \know_a (G_n \rightarrow  \neg (\neg \forall x D \rightarrow \neg D [z/x]))$ because it is counterexemplar. Since $w$ cannot contain $ \neg \forall z \know_a (G_n \rightarrow   \neg (\neg \forall x D \rightarrow \neg D [z/x]))$, it must contain $ \forall z \know_a (G_n \rightarrow  (\neg D [z/x]  \rightarrow \forall x D))$ since it is maximal. Moreover, the fact that $y$ could be any variable occurring in $w$ implies that we can certainly find a variable $z$ that does not appear in $G_n$ such that $\forall z \know_a (G_n \rightarrow \neg (\neg \forall x D \rightarrow \neg D [z/x]))\in w$.

Since, as we will show below, the Barcan Formula (BF) is a theorem of our logic, from $\forall z \know_a (G_n \rightarrow   (\neg D [z/x]  \rightarrow \forall x D))\in w$, we can derive  $ \know_a\forall z (G_n \rightarrow   \neg (\neg \forall x D \rightarrow \neg D [z/x]))$, which implies that $ \know_a\forall z (G_n \rightarrow   \neg (\neg \forall x D \rightarrow \neg D [z/x]))\in w$. Since $z$ does not appear in $G_n$, it also follows that $ \know_a (G_n \rightarrow \forall z   \neg (\neg \forall x D \rightarrow \neg D [z/x]))\in w$.

% , and, by contraposition, that $ \know_a ( \neg \forall z   D[z/x]\rightarrow \neg G_n)\in w$ as well.

%By the axiom K, $ \know_a \neg \forall z D[z/x]\rightarrow \know_a  \neg G_n \in w$.

Now, $\forall z   \neg (\neg \forall x D \rightarrow \neg D [z/x])$ is classically contradictory. Indeed, intuitively, if there exists some $z$ such that $\neg D [z/x]$ is true, then the implication is true in virtue of the fact that the consequent is true for that $z$; if no $z$ verifies $\neg D [z/x]$, then $D [z/x]$ is true for any $z$, and thus the implication is true in virtue of the fact that the antecedent is false. In either case, the formula is false. Since the fragment of the calculus that only includes the classical rules in Figure \ref{fig:con-quant-rules} is a standard, complete calculus for classical logic, we can construct a proof of $\neg \forall z   \neg (\neg \forall x D \rightarrow \neg D [z/x])$ in a rather standard way.
%%%
%%%NEED PROOF? FIND IT HERE
%%%
Indeed, we can translate the instance $\forall x D\vee\neg\forall x D$ of the law of the excluded middle  as $\neg(\neg \forall x D\wedge \neg \neg\forall x D)$ and prove it as follows
\[\infer{\Rightarrow \neg(\neg \forall x D\wedge \neg \neg\forall x D)}{\infer{\neg \forall x D\wedge \neg \neg\forall x D\Rightarrow \bot}{\infer{\neg \forall x D\wedge \neg \neg\forall x D\Rightarrow  \neg\neg\forall x D}{\infer{\neg \forall x D\wedge \neg \neg\forall x D\Rightarrow  \neg \forall x D\wedge \neg \neg\forall x D}{}}&\infer{\neg \forall x D\wedge \neg \neg\forall x D\Rightarrow  \neg\forall x D}{\infer{\neg \forall x D\wedge \neg \neg\forall x D\Rightarrow  \neg \forall x D\wedge \neg \neg\forall x D}{}}}}\]
And, in order to simulate  an elimination of the disjunction $\forall x D\vee\neg\forall x D$ of the form
\[\infer{\Gamma\Rightarrow F}{\deduce{\Rightarrow \forall x D\vee\neg\forall x D}{\dots}&
\deduce{\Gamma , \forall x D\Rightarrow F}{\dots}
&
\deduce{\neg \forall x D , \Gamma \Rightarrow F}{\dots}
}\]we can use the following standard schema:
\[\infer{\Gamma\Rightarrow F}{\infer{\Gamma\Rightarrow \neg\neg F}{\infer{\neg F , \Gamma\Rightarrow \bot}{\deduce{\Rightarrow \neg(\neg \forall x D\wedge \neg \neg\forall x D)}{\dots}&\infer{\neg F , \Gamma\Rightarrow \neg \forall x D\wedge \neg \neg\forall x D}{\infer{\neg F , \Gamma\Rightarrow \neg \forall x D}{\infer{\neg F , \Gamma, \forall x D\Rightarrow\bot}{\infer{\neg F\Rightarrow \neg F}{}&\deduce{ \Gamma , \forall x D  \Rightarrow F}{\dots}
}}&\infer{\neg F , \Gamma\Rightarrow \neg \neg\forall x D}{\infer{\neg F  , \neg \forall x D , \Gamma\Rightarrow\bot}{\infer{\neg F\Rightarrow \neg F}{}&\deduce{\neg \forall x D , \Gamma \Rightarrow F}{\dots}
}}}}}}\]
Therefore, we can translate the following proof in our calculus:
\[\infer{\Rightarrow \neg\forall z   \neg (\neg \forall x D \rightarrow \neg D [z/x])}{\infer{\forall z   \neg (\neg \forall x D \rightarrow \neg D [z/x])\Rightarrow \bot}{\deduce{\Rightarrow\forall x D\vee\neg\forall x D}{\dots}&
\deduce{\forall z   \neg (\neg \forall x D \rightarrow \neg D [z/x]), \forall x D\Rightarrow\bot}{\delta _1}
%
%\infer{\forall z   \neg (\neg \forall x D \rightarrow \neg D [z/x]), \forall x D\Rightarrow\bot}{\infer{\forall z   \neg (\neg \forall x D \rightarrow \neg D [z/x])\Rightarrow\neg (\neg \forall x D \rightarrow \neg D [z/x])}{\infer{\forall z   \neg (\neg \forall x D \rightarrow \neg D [z/x])\Rightarrow \forall z   \neg (\neg \forall x D \rightarrow \neg D [z/x])}{}}&
%\infer{ \forall x D \Rightarrow\neg \forall x D\rightarrow\neg D[z/x]}{
%\infer{\neg \forall x D , \forall x D\Rightarrow\neg D[z/x]}{\infer{\neg \forall x D , \forall x D\Rightarrow\bot}{
%\infer{\neg \forall x D\Rightarrow\neg \forall x D}{} & \infer{ \forall x D\Rightarrow\forall x D}{}}}}
%}
&
\deduce{\neg \forall x D , \forall z   \neg (\neg \forall x D \rightarrow \neg D [z/x]) \Rightarrow\bot}{\delta _2}
%
%\infer{\neg \forall x D , \forall z   \neg (\neg \forall x D \rightarrow \neg D [z/x]) \Rightarrow\bot}{\infer{\neg \forall x D\Rightarrow \neg \forall x D}{}
%&
%\infer{\forall z   \neg (\neg \forall x D \rightarrow \neg D [z/x])\Rightarrow  \forall x D}{\infer{\forall z   \neg (\neg \forall x D \rightarrow \neg D [z/x])\Rightarrow D[y/x]}{\infer{\forall z   \neg (\neg \forall x D \rightarrow \neg D [z/x]) \Rightarrow  \neg\neg D[y/x]}{\infer{\forall z   \neg (\neg \forall x D \rightarrow \neg D [z/x])  , \neg D[y/x]\Rightarrow  \bot}{
%\infer{\forall z   \neg (\neg \forall x D \rightarrow \neg D [z/x]) \Rightarrow   \neg (\neg \forall x D \rightarrow \neg D [z/x])}{\infer{\forall z   \neg (\neg \forall x D \rightarrow \neg D [z/x]) \Rightarrow\forall z   \neg (\neg \forall x D \rightarrow \neg D [z/x])}{}}
%&
%\infer{\neg D[y/x]\Rightarrow \neg \forall x D \rightarrow \neg D [y/x]}{\infer{\neg D[y/x]\Rightarrow\neg D[y/x]}{}}}}}}}
%
}}\]where $\delta_1$ is
\[
\infer{\forall z   \neg (\neg \forall x D \rightarrow \neg D [z/x]), \forall x D\Rightarrow\bot}{\infer{\forall z   \neg (\neg \forall x D \rightarrow \neg D [z/x])\Rightarrow\neg (\neg \forall x D \rightarrow \neg D [z/x])}{\infer{\forall z   \neg (\neg \forall x D \rightarrow \neg D [z/x])\Rightarrow \forall z   \neg (\neg \forall x D \rightarrow \neg D [z/x])}{}}&
\infer{ \forall x D \Rightarrow\neg \forall x D\rightarrow\neg D[z/x]}{
\infer{\neg \forall x D , \forall x D\Rightarrow\neg D[z/x]}{\infer{\neg \forall x D , \forall x D\Rightarrow\bot}{
\infer{\neg \forall x D\Rightarrow\neg \forall x D}{} & \infer{ \forall x D\Rightarrow\forall x D}{}}}}
}
\]and $\delta_2$ is 
\[\infer{\neg \forall x D , \forall z   \neg (\neg \forall x D \rightarrow \neg D [z/x]) \Rightarrow\bot}{\infer{\neg \forall x D\Rightarrow \neg \forall x D}{}
&
\infer{\forall z   \neg (\neg \forall x D \rightarrow \neg D [z/x])\Rightarrow  \forall x D}{\infer{\forall z   \neg (\neg \forall x D \rightarrow \neg D [z/x])\Rightarrow D[y/x]}{\infer{\forall z   \neg (\neg \forall x D \rightarrow \neg D [z/x]) \Rightarrow  \neg\neg D[y/x]}{\infer{\forall z   \neg (\neg \forall x D \rightarrow \neg D [z/x])  , \neg D[y/x]\Rightarrow  \bot}{
\infer{\forall z   \neg (\neg \forall x D \rightarrow \neg D [z/x]) \Rightarrow   \neg (\neg \forall x D \rightarrow \neg D [z/x])}{\infer{\forall z   \neg (\neg \forall x D \rightarrow \neg D [z/x]) \Rightarrow\forall z   \neg (\neg \forall x D \rightarrow \neg D [z/x])}{}}
&
\infer{\neg D[y/x]\Rightarrow \neg \forall x D \rightarrow \neg D [y/x]}{\infer{\neg D[y/x]\Rightarrow\neg D[y/x]}{}}}}}}}
\]

Now, since, as just shown, $ \forall z   \neg (\neg \forall x D \rightarrow \neg D [z/x])$ is a contradiction, $G_n \rightarrow \forall z   \neg (\neg \forall x D \rightarrow \neg D [z/x])$ is logically equivalent to $\neg G_n$. Hence, since $ \know_a (G_n \rightarrow \forall z   \neg (\neg \forall x D \rightarrow \neg D [z/x]))\in w$, we must have that $\know_a  \neg G_n\in w$.
But this means that $\neg G_n\in\know_a  ^-(w)$ and thus that $\know_a  ^-(w)\cup \{G _n \}$ is not consistent. This contradicts our assumption on $\know_a  ^-(w)\cup \{G _n \}$. Hence, there exists a $y$ for which $\know_a  ^-(w)\cup \{G _n \wedge \neg D [y/x]\} $ is consistent.

This concludes our induction and proves that we can consistently extend $\know_a  ^-(w)$ to a set defined as the union $\know_a  ^-(w)\cup \{G_1 , G_2, G_3, \dots \} $ where $G _0 = \neg A$ and $G_{n+1}= G _n \wedge (\neg \forall x D \rightarrow \neg D [y/x]) $. Note, indeed, that the  consistency of $ \{G_1 , G_2, G_3, \dots \} $ is guaranteed by the fact that each $G_n$ is consistent and implies $G_{n-1}$. 
%Since, moreover, all negated universally quantified  formulae of the language have been considered in the construction of $\know_a  ^-(w)\cup \{G_1 , G_2, G_3, \dots \} $, for every possible formula required to make the set  $\know_a  ^-(w)\cup \{G_1 , G_2, G_3, \dots \} $ counterexemplar, a suitable variable in the language of $w$ has already been selected. 
By Lemma \ref{lem:ultrafilter}, we can now extend $\know_a  ^-(w)\cup \{G_1 , G_2, G_3, \dots \} $ to a maximally consistent and counterexemplar set  $\overline{v}^* $ without further extending the language with respect to the variables occurring in $w$. Indeed, whenever we need to add a formula $\neg \forall x D$ to the set which is being constructed by the procedure in the proof of  Lemma \ref{lem:ultrafilter}, the presence of the formula $\neg \forall x D \rightarrow \neg D [y/x]$ in the original set will already give us a suitable instance $ \neg D [y/x]$ that can make the set counterexemplar without introducing a new variable in the language.

Since $\overline{v}^* $ is maximally consistent and counterexemplar,
% set $\overline{v}^* $ without extending its language with new variables but by only using the variables already occurring in $w$.
% Indeed, during the extension defined in the proof of  Lemma \ref{lem:ultrafilter}, formulae of the form $\neg \forall x A$ are only added in case adding $\forall x A $ does not yield a consistent set. Moreover, formulae guaranteeing that the resulting set is counterexemplar are only added when required.
by Definition \ref{def:canonical-model}, $\overline{v}^*\in \canw$. Obviously, $\neg A \in \overline{v}^*$ and, again by Definition \ref{def:canonical-model}, $w\canr_a  \overline{v}^*$ since $ \know_a ^-(w)\subseteq \overline{v}^* $. We have thus showed, as desired,  that  $\overline{v}^* \in \canw$ exists, that $\neg A\in \overline{v}^* $ and $w\canr_a \overline{v}^* $.

We just need to show now that the Barcan formula (BF) $\forall x \know _a A\rightarrow  \know _a\forall x A$, used in the argument above, is a theorem of the calculus:
\[
\infer{\Rightarrow\forall x \know _a A\rightarrow  \know _a\forall x A } {
    \infer{\forall x \know _a A\Rightarrow\know _a\forall x A} {
      \deduce{ \Rightarrow\know _a\neg \know _a\neg \forall x\know _a A \rightarrow  \know _a\forall x A} {\delta_1}
      &
      \deduce{\forall x \know _a A\Rightarrow\know _a\neg \know _a\neg \forall x\know _a A} {\delta _2}
    }
  }\]
where $\delta_1 $ is
\[\infer{ \Rightarrow\know _a\neg \know _a\neg \forall x\know _a A \rightarrow  \know _a\forall x A} {
\infer{ \know_a\neg\know_a \neg \forall x \know_a A\Rightarrow\know _a  \forall xA} {
\infer{  \know_a\neg\know_a \neg \forall x \know_a A\Rightarrow\forall xA} {
\infer{\know_a\neg\know_a \neg \forall x \know_a A\Rightarrow A}{\infer{\know_a\neg\know_a \neg \forall x \know_a A\Rightarrow\know_a A} {
\infer{\know_a\neg\know_a \neg \forall x \know_a A\Rightarrow\neg\neg \know_a A}{\infer{\know_a\neg\know_a \neg \forall x \know_a A,\neg \know_a A\Rightarrow\bot } {\infer{\know_a\neg\know_a \neg \forall x \know_a A \Rightarrow\neg \know _a\neg \know _a A}{\infer{\know_a\neg\know_a \neg \forall x \know_a A , \know_a \neg \know _a A \Rightarrow\bot} {\infer{\know_a\neg\know_a \neg \forall x \know_a A\Rightarrow\neg\know_a \neg \forall x \know_a A }{\infer{\know_a\neg\know_a \neg \forall x \know_a A\Rightarrow \know_a\neg\know_a \neg \forall x \know_a A}{}}&\infer{\know_a \neg \know _a A \Rightarrow\know_a \neg \forall x \know_a A}{ \infer{\know_a \neg \know _a A \Rightarrow\neg \forall x \know_a A}{\infer{\know_a \neg \know _a A , \forall x  \know _a A\Rightarrow\bot}{\infer{\know_a \neg \know _a A\Rightarrow\neg \know _a A}{\infer{\know_a \neg \know _a A\Rightarrow\know_a \neg \know _a A}{}}&\infer{\forall x  \know _a A\Rightarrow\know _a A }{\infer{\forall x  \know _a A\Rightarrow\forall x  \know _a A}{}}}}}}}&\infer{ \neg \know_a A\Rightarrow\know _a\neg \know _a A}{\infer{\neg \know_a A\Rightarrow \neg \know_a A}{}}}}}}}}}\]
and $\delta_2$ is
\[\infer{\forall x \know _a A\Rightarrow\know _a\neg \know _a\neg \forall x\know _a A} {\infer{\forall x \know _a A\Rightarrow\neg \know _a\neg \forall x\know _a A}{\infer{\forall x \know _a A, \know _a \neg \forall x \know_a A\Rightarrow\bot}{\infer{\forall x \know _a A\Rightarrow\neg\neg \forall x \know _a A}{\infer{\neg \forall x \know _a A, \forall x \know _a A\Rightarrow\bot}{\infer{\neg\forall x \know _a A\Rightarrow\neg \forall x \know _a A}{}&\infer{ \forall x \know _a A\Rightarrow \forall x \know _a A}{}}}&\infer{\know _a \neg \forall x \know_a A\Rightarrow\neg \forall x \know _a A}{\infer{\know _a \neg \forall x \know_a A\Rightarrow\know _a \neg \forall x \know_a A}{}}}}}\]
\end{proof}

\noindent {\bf Lemma \ref{lem:truth}} (Truth lemma){\bf .} {\it For any $w\in \canw$ and formula $A$, $A\in w$ if, and only if, $\canm , w, \cani_w \Vdash A$.}
\begin{proof}
The proof is by induction on the number of occurrences of sentential operators---that is, connectives, quantifiers, $\know$ and justifications---in $A$. In the base case, $A$ is atomic. We have then two subcases.
\begin{itemize}
\item $A=(u\lid v)$. Suppose $u\lid v\in w $. Then, by Defintion \ref{def:canonical-model}, $(\underline{u}, \underline{v})\in \cani_w(\lid)$ and hence $\canm , w, \cani _w \Vdash u\lid v$. 

Suppose  $\canm , w, \cani _w \Vdash u\lid v$. This is the case if, and only if, $(\underline{u}, \underline{v})\in \cani_w(\lid)$. But by Defintion \ref{def:canonical-model}, this implies that $u\lid v\in w $.

\item $A= P(t_1,\dots ,t_n)$. Suppose $P(t_1,\dots ,t_n) \in w $. Then, by Defintion \ref{def:canonical-model}, $(\underline{t_1},\dots , \underline{t_n})\in \cani_w(P)$ and hence $\canm , w, \cani _w \Vdash P(t_1,\dots ,t_n)$. 

Suppose  $\canm , w, \cani _w \Vdash P(t_1,\dots ,t_n)$. This is the case if, and only if, $(\underline{t_1},\dots , \underline{t_n})\in \cani_w(P)$. But by Defintion \ref{def:canonical-model}, this implies that $P(t_1,\dots ,t_n)\in w $.
\end{itemize}

Let us then  suppose that the statement holds for any formula containing less than $n$ occurrences of sentential operators. We show it holds also for any formula $A$ that contains $n$ occurrences of sentential operators. Let us reason on the form of $A$ in case it contains at least one occurrence of a sentential operator. We consider $\bot$ as a $0$-ary operator of this kind.

\begin{itemize}
\item $A=\bot $. Obviously, $A\not\in w$  and $\canm , w, \cani_w \not\Vdash \bot$. Indeed, respectively, $w$ is an MCC and thus consistent, and $\canm$ is a model.

\item $A= B\rightarrow C$. Suppose $B\rightarrow C\in w $. Either $B\notin w$ or $C\in w$ holds. Otherwise, $w$ would be contradictory, which is impossible since $w$ is an MCC. By inductive hypothesis, then, either $\canm , w, \cani _w \not\Vdash B$ or $\canm , w, \cani _w \Vdash C$ holds. But this precisely implies that $\canm , w, \cani _w \Vdash B \rightarrow C$. 

Suppose $\canm , w, \cani _w \Vdash A \rightarrow B$. By Definition \ref{def:forcing} for $\rightarrow$, either $\canm , w, \cani _w \not\Vdash B$ or $\canm , w, \cani _w \Vdash A$  holds. 
By inductive hypothesis, then, either $B\notin w$ or $A\in w$ holds. But since $w$ is an MCC, by Lemma \ref{lem:maximality} and by the consistency of MCCs, we can conclude that $B\rightarrow C\in w $. 

\item  $A= B\wedge C$. Suppose $B\wedge C\in w $. Then both $B\in w$ and $C\in w$ hold. Otherwise, $w$ would not be maximally consistent, which is impossible since $w$ is an MCC. By inductive hypothesis, then, both $\canm , w, \cani _w \Vdash B$ and $\canm , w, \cani _w \Vdash C$ hold. But this precisely implies that $\canm , w, \cani _w \Vdash B \wedge C$. 

Suppose $\canm , w, \cani _w \Vdash A \wedge B$. By Definition \ref{def:forcing} for $\wedge$, both $\canm , w, \cani _w \Vdash B$ and $\canm , w, \cani _w \Vdash A$ hold. 
By inductive hypothesis, then, both $B\in w$ and $A\in w$ hold. But since $w$ is an MCC, by Lemma \ref{lem:maximality} and by the consistency of MCCs, we can conclude that $B\wedge C\in w $. 

\item  $A= \forall y. B $. Suppose $\forall y. B\in w $. Then $B[t/y]\in w$ holds for any term $t$. Otherwise, $w$ would not be maximally consistent, which is impossible since $w$ is an MCC. By inductive hypothesis, then, $\canm , w, \cani _w \Vdash B[t/y] $ holds for any term $t$. {\ehi But this 
%along with 
%{\ehi the fact that, for any object $\underline {t}\in \canu$, there exists a term $t$ such that $\cani  _v (t)=\underline{t}$, for any $v\in \canw$, 
%}
%%\item the fact the value of $f(x)$, {\ehi for any assignment $f$}, is the same at all states $v\in \canw$,
%%\end{itemize}
implies that 
%$\canm , w, \cani_w \Vdash B$ holds. And from this we can deduce that
$\canm , w, \cani _w \Vdash \forall y. B$ holds. }

Suppose $\canm , w, \cani _w \Vdash \forall y. B$. By Definition \ref{def:forcing} for $\forall$, {\ehi $\canm , w,\cani _{w}\Vdash B[t/y]$ holds for any term $t$.
%$x$-variant $f$ of $\cani _w $. 
}
%We show, by reasoning indirectly, that $B[t/y] \in w$ must hold for any term $t$. Suppose that this is not the case and thus that $B[s/y] \notin w$ for some term $s$. By taking the contrapositive of one direction of the inductive hypothesis then, we can deduce that $\canm , w, \cani _w \not\Vdash B[s/y]$.
%Since this---along with the fact that the value of $f(x)$, for any $x$-variant $f$ of any assignment, is the same at all states---implies that there is some $x$-variant $f$ of $\cani_w$ such that $\canm , w,\cani _{w}\not\Vdash B$. We can deduce that $\canm , w, \cani _w \not\Vdash \forall y.B$. Which contradicts our assumptions. 
Hence, 
$B[t/y] \in w$ must hold for any term $t$. Since $w$ is an MCC, by Lemma \ref{lem:maximality} and by the counterexemplarity of MCCs, we can conclude that $\forall y. B\in w $. Indeed,  $\neg \forall y. B\in w $ would imply that there exists a term $s$ such that $\neg B[s/y] \in w$.

\item $A=\know _a B$. Suppose $\know _a B\in w $. Then $B\in v$ must hold for any $v\in \canw$ such that $w\canr _a v$ because of Definition \ref{def:canonical-model}. By inductive hypothesis then, we have that $\canm , v, \cani _v \Vdash B$ for any $v\in \canw$ such that $w\canr _a v$. But this is exactly the condition under which
$\canm , w, \cani _w \Vdash \know _a B$ holds. 

Suppose $\canm , w, \cani _w \Vdash \know _a B$. Then it must be the case that $\canm , v, \cani _v \Vdash B$ for any $v\in \canw$ such that $w\canr _a v$. But then, by inductive hypothesis, we must have that $B\in v$ for any $v\in \canw$ such that $w\canr _a v$.
Now, by Lemma \ref{lem:maximality}, since $w$ is an MCC, it must either be the case that $\know _a B\in w$ or be the case that $\neg \know _a B\in w$. Suppose, by reasoning indirectly, that $\neg \know _a B\in w$. But then, by Lemma \ref {lem:existence}, we know that there must  exist a $v\in \canw$ such that $\neg B\in v$ and $w\canr_a v$.
Since this---along with the fact that any $v\in \canw$ is an MCC---contradicts the fact that $B\in v$ for any $v\in \canw$ such that $w\canr _a v$, we can conclude that $\know _a B\in w$.

\item $A=j: B$. Suppose $j: B\in w $. Then, because of Definition \ref{def:canonical-model}, $B\in v$ must hold for any $v\in \canw$ such that $w\canr _\gamma v$ and we must have that $B \in \cane _w(j)$. By inductive hypothesis then, we have that $\canm , v, \cani _v \Vdash B$ for any $v\in \canw$ such that $w\canr _\gamma v$. But this and $B \in \cane _w(j)$---which must hold since $j:B\in w$---are exactly the conditions under which
$\canm , w, \cani _w \Vdash j: B$ holds. 

Suppose $\canm , w, \cani _w \Vdash j: B$. Then it must be the case that $B\in \cane_w(j)$ and that $\canm , v, \cani _v \Vdash B$ for any $v\in \canw$ such that $w\canr _\gamma v$. But then, since  $B\in \cane_w(j)$, we must have, by Definition \ref{def:canonical-model},  that $j:B\in w $, as desired.

%%NecJustifications
%Suppose $\canm , w, \cani _w \Vdash j: B$. Then it must be the case that $B\in \cane_w(j)$ and that $\canm , v, \cani _v \Vdash B$ for any $v\in \canw$ such that $w\canr _\gamma v$. But then, by inductive hypothesis, we must have that $B\in v$ for any $v\in \canw$ such that $w\canr _\gamma v$.
%Now, by Lemma \ref{lem:maximality}, since $w$ is an MCC, it must either be the case that $j: B\in w$ or be the case that $\neg j: B\in w$. Suppose, by reasoning indirectly, that $\neg j: B\in w$. But then, by Lemma \ref {lem:existence}, we know that there must  exist a $v\in \canw$ such that $\neg B\in v$ and $w\canr_\gamma v$.
%Since this---along with the fact that any $v\in \canw$ is an MCC---contradicts the fact that $B\in v$ for any $v\in \canw$ such that $w\canr _\gamma v$, we can conclude that $j: B\in w$.

\end{itemize}
\end{proof}

\end{appendices}

\section*{Acknowledgments}

We thank Mitar Nedeljkovic for pointing out several issues with a previous definition of the semantics that helped greatly improve this version of the work.


\bibliography{bib-hyper.bib}% common bib file

\end{document}